\documentclass[11pt]{article}
\usepackage{amsfonts,amsmath,amssymb,amsthm,boxedminipage,color,url,fullpage,mathtools}
\usepackage{algorithm}
\usepackage{algpseudocode}
\usepackage{bbm}
\usepackage{enumitem}
\usepackage{stmaryrd}
\usepackage{float}
\usepackage[numbers]{natbib} 
\usepackage{footmisc}
\usepackage[labelfont=bf]{caption}
\usepackage{aliascnt} 
\usepackage{graphicx}
\usepackage{makecell}
\usepackage[table]{xcolor}
\usepackage{xparse}
\usepackage[pdfstartview=FitH,colorlinks,linkcolor=blue,filecolor=blue,citecolor=blue,urlcolor=blue]{hyperref} 
\newtheorem{theorem}{Theorem}[section]
\newtheorem{corollary}[theorem]{Corollary}
\newtheorem{fact}[theorem]{Fact}
\newtheorem{lemma}[theorem]{Lemma}
\newtheorem{definition}[theorem]{Definition}
\newtheorem{remark}[theorem]{Remark}
\newtheorem{claim}[theorem]{Claim}

\newcommand{\norm}[1]{\left\| #1 \right\|}

\newcommand{\normalized}[1]{\Dminushalf #1 \Dminushalf}

\newcommand{\EE}{\mathop{\mathbb E}} 
\newcommand{\RR}{\mathop{\mathbb R}}
\newcommand{\NN}{\mathop{\mathbb N}}

\newcommand{\CC}{\mathop{\mathbb C}}
\newcommand{\PP}{\mathop{\mathbb P}}

\newcommand{\1}{\mathbbm{1}}
\newcommand{\eg}{{\it e.g.}}
\newcommand{\ie}{{\it i.e.}}

\newcommand{\etal}{{\it et al.}}
\newcommand{\unitFlow}{Bounded-Distance-Flow}
\newcommand\numeq[1]{\stackrel{\scriptscriptstyle(\mkern-1.5mu#1\mkern-1.5mu)}{=}}
\newcommand\numge[1]{\stackrel{\scriptscriptstyle(\mkern-1.5mu#1\mkern-1.5mu)}{\ge}}
\newcommand\numle[1]{\stackrel{\scriptscriptstyle(\mkern-1.5mu#1\mkern-1.5mu)}{\le}}
\NewDocumentCommand\F{gg}{\ensuremath{F\IfNoValueTF{#1}{}{_{#1}}\IfNoValueTF{#2}{}{(#2)}}}
\NewDocumentCommand\M{gg}{\ensuremath{M\IfNoValueTF{#1}{}{_{#1}}\IfNoValueTF{#2}{}{(#2)}}}
\NewDocumentCommand\N{gg}{\ensuremath{N\IfNoValueTF{#1}{}{_{#1}}\IfNoValueTF{#2}{}{(#2)}}}
\NewDocumentCommand\W{gg}{\ensuremath{W\IfNoValueTF{#1}{}{_{#1}}\IfNoValueTF{#2}{}{(#2)}}}
\NewDocumentCommand\D{gg}{\ensuremath{D\IfNoValueTF{#1}{}{_{#1}}\IfNoValueTF{#2}{}{(#2)}}}
\NewDocumentCommand\A{gg}{\ensuremath{A\IfNoValueTF{#1}{}{_{#1}}\IfNoValueTF{#2}{}{(#2)}}}
\NewDocumentCommand\B{gg}{\ensuremath{B\IfNoValueTF{#1}{}{_{#1}}\IfNoValueTF{#2}{}{(#2)}}}
\NewDocumentCommand\X{gg}{\ensuremath{X\IfNoValueTF{#1}{}{_{#1}}\IfNoValueTF{#2}{}{(#2)}}}
\NewDocumentCommand\Y{gg}{\ensuremath{Y\IfNoValueTF{#1}{}{_{#1}}\IfNoValueTF{#2}{}{(#2)}}}
\NewDocumentCommand\Z{gg}{\ensuremath{Z\IfNoValueTF{#1}{}{_{#1}}\IfNoValueTF{#2}{}{(#2)}}}
\NewDocumentCommand\R{gg}{\ensuremath{R\IfNoValueTF{#1}{}{_{#1}}\IfNoValueTF{#2}{}{(#2)}}}
\NewDocumentCommand\Smat{gg}{\ensuremath{S\IfNoValueTF{#1}{}{_{#1}}\IfNoValueTF{#2}{}{(#2)}}}
\NewDocumentCommand\U{gg}{\ensuremath{U\IfNoValueTF{#1}{}{_{#1}}\IfNoValueTF{#2}{}{(#2)}}}
\NewDocumentCommand\Hmat{gg}{\ensuremath{H\IfNoValueTF{#1}{}{_{#1}}\IfNoValueTF{#2}{}{(#2)}}}
\NewDocumentCommand\Pmat{gg}{\ensuremath{P\IfNoValueTF{#1}{}{_{#1}}\IfNoValueTF{#2}{}{(#2)}}}
\NewDocumentCommand\I{gg}{\ensuremath{I\IfNoValueTF{#1}{}{_{#1}}\IfNoValueTF{#2}{}{(#2)}}}
\NewDocumentCommand\dG{gg}{\ensuremath{d\IfNoValueTF{#1}{}{_{#1}}\IfNoValueTF{#2}{}{(#2)}}}
\NewDocumentCommand\Matrix{mm}{\ensuremath{\left(#1\right)(#2)}}
\NewDocumentCommand\MatrixSimple{mm}{\ensuremath{#1(#2)}}
\NewDocumentCommand\dGG{g}{\ensuremath{d_G\IfNoValueTF{#1}{}{(#1)}}}
\NewDocumentCommand\vv{g}{\ensuremath{v\IfNoValueTF{#1}{}{(#1)}}}
\NewDocumentCommand\Dminushalf{gg}{\ensuremath{D\IfNoValueTF{#1}{}{_{#1}}^{-\frac{1}{2}}\IfNoValueTF{#2}{}{(#2)}}}
\DeclareMathOperator{\vol}{\mathbf{vol}}

\DeclareMathOperator{\tr}{Tr}
\DeclareMathOperator{\diag}{\mathbf{diag}}
\newcommand{\defeq}{\mathrel{\mathop:}=}

\begin{document}
	\title{Expander Decomposition with Fewer Inter-Cluster Edges Using a Spectral Cut Player\thanks{This work was supported in part by Israel Science Foundation grant no.\ 1595/19 and 1156/23 and the Blavatnik Family Foundation.}
	}

	\author{Daniel Agassy\thanks{Tel Aviv University,
danielagassy@mail.tau.ac.il} \and Dani Dorfman\thanks{Max Planck Institute for Informatics,
Saarbrücken, Germany. Part of this work was done in Tel Aviv University,
dani.i.dorfman@gmail.com} 	\and Haim Kaplan\thanks{Tel Aviv University, haimk@tau.ac.il}}

	\maketitle

	\noindent

\begin{abstract}
A $(\phi,\epsilon)$-expander decomposition of a graph $G$ (with $n$ vertices and $m$ edges) is a partition of $V$ into clusters $V_1,\ldots,V_k$ with conductance $\Phi(G[V_i]) \ge \phi$, such that there are at most $\epsilon m$ inter-cluster edges. Such a decomposition plays a crucial role in many graph algorithms. We  give a randomized $\tilde{O}(m/\phi)$ time algorithm for computing a $(\phi, \phi\log^2 {n})$-expander decomposition. This improves upon the $(\phi, \phi\log^3 {n})$-expander decomposition also obtained in $\tilde{O}(m/\phi)$ time by [Saranurak and Wang, SODA 2019] (SW) and brings the number of inter-cluster edges within logarithmic factor of optimal. 

\looseness = -1
One crucial component of SW's algorithm is a non-stop version of the cut-matching game of 
[Khandekar, Rao, Vazirani, JACM 2009] (KRV):
 The cut player does not stop when it gets from the matching player an unbalanced sparse cut, but continues to play on a trimmed part of the large side. The crux of our improvement is the design of a non-stop version of the cleverer cut player of [Orecchia, Schulman, Vazirani, Vishnoi, STOC 2008] (OSVV). The cut player of OSSV uses a more sophisticated random walk, a subtle potential function, and spectral arguments. Designing and analysing a non-stop version of this game was an explicit open question asked by SW.
\end{abstract}

\section{Introduction}
\looseness=-1
The \emph{conductance} of a cut $(S, V\setminus {S})$ is $\Phi_{G}(S,V\setminus {S}) = \frac{|E(S,V\setminus {S})|}{\min(\vol(S), \vol(V\setminus {S}))}$, where $\vol(S)$ is the sum of the degrees of the vertices of $S$. The  conductance of a graph $G$ is the smallest conductance of a cut in $G$. 

A \emph{$(\phi,\epsilon)$-expander decomposition} of a graph $G$ is a partition of the vertices of $G$ into clusters $V_1,\ldots,V_k$ with conductance $\Phi(G[V_i]) \ge \phi$ such that there are at most $\epsilon m$ inter-cluster edges, where $\phi,\epsilon\ge 0$. We consider the problem of computing in almost linear time ($\tilde{O}(m)$ time) a $(\phi,\epsilon)$-expander decomposition for a given graph $G$ and $\phi>0$, while minimizing $\epsilon$ as a function of $\phi$.  
It is known that a $(\phi,\epsilon)$-expander decomposition, with $\epsilon = O(\phi\log n)$, always exists and that $\epsilon = \Theta(\phi\log n)$ is optimal~\cite{saranurak2019expander,alev2017graph}. 

Expander decomposition algorithms have been used in many cutting edge results, such as directed/undirected Laplacian solvers~\cite{spielman2004nearly,cohen2017almost}, graph sparsification~\cite{chu2020graph,chuzhoy2020deterministic}, distributed algorithms~\cite{chang2020deterministic}, and maximum flow algorithms~\cite{kelner2014almost}. Expander decomposition was also used \cite{chuzhoy2020deterministic}  in order to break the $O\left(\sqrt{n}\right)$ deterministic dynamic connectivity bound and achieve an improved running time of $O(n^{o(1)})$ per operation, and  in the recent breakthrough result by Chen et al.~\cite{chen2022maximum}, who showed algorithms for maximum flow and minimum cost flow in almost linear time. 

\looseness = -1
Given an $f(n)$-approximation algorithm for the problem of finding a minimum conductance cut, one can get a $(\phi, O(f(n)\cdot\phi\log n))$-expander decomposition algorithm by recursively computing approximate cuts (and thus splitting $V$) until all components are certified as expanders.  
In particular, using an exact minimum conductance cut algorithm ensures the existence of an expander decomposition with $\epsilon = O\left(\phi\log n\right)$ as mentioned above. Using the polynomial algorithms of \cite{orecchia2012approximating, arora2009expander} which provide the best approximation ratios  of $O\left(\sqrt{\phi}\right)$ and $O\left(\sqrt{\log n}\right)$, respectively, for conductance, gives polynomial time expander decomposition algorithms with $\epsilon =O\left(\phi^{3/2}\log n\right)$ and $\epsilon = O\left(\phi\log^{\frac{3}{2}} n\right)$. 
However, these decomposition algorithms might lead to a linear recursion depth, and therefore have superlinear time complexity.

To get a near linear time algorithm using this recursive  approach, one must be able to efficiently compute low conductance cuts with additional guarantees. We get such cuts  using the cut-matching framework of \cite{khandekar2009graph} (abbreviated as KRV). In order to present our results in the appropriate context we now give a brief background on the cut-matching framework.

\smallskip
\noindent
{\bf Cut-matching:} 
Edge-expansion is a connectivity measure related to conductance. 
The \emph{edge-expansion} of a cut $(S, V\setminus {S})$ is $h_{G}(S,V\setminus {S}) = \frac{|E(S,V\setminus {S})|}{\min(|S|,|V\setminus {S}|)}$ and the \emph{edge-expansion} of a graph $G$ is  the smallest edge-expansion of a cut in $G$.

The cut-matching game is a technique that reduces the approximation task for sparsest  cut (in terms of edge-expansion) to a polylogarithmic number of maximum flow problems.
The resulting approximation algorithm for sparsest cut is remarkably simple and robust. 

The cut-matching game is played between a \emph{cut player} and a \emph{matching player}, as follows. We start with an empty graph $G_0$ on $n$ vertices. At round $t$, the cut player chooses a bisection $(S_t, \overline{S_t})$ of the vertices (we assume $n$ is even). In response, the matching player presents a perfect matching $M_t$ between the vertices of $S_t$ and $\overline{S_t}$ and the game graph is updated to $G_t = G_{t-1} \cup M_t$. Note that this graph may contain parallel edges. The game ends when $G_t$ is a sufficiently good edge-expander. The goal of this game is to devise a strategy for the cut player that maximizes the ratio $r(n)\defeq \phi/T$, where $T$ is the number of rounds and $\phi = h(G_T)$ is the edge-expansion of $G_T$. KRV showed that one can translate a cut strategy of quality $r(n)$ into a sparsest cut algorithm of approximation ratio $1/r(n)$ by  applying a binary search on a sparsity parameter $\phi$ until we certify	that $h(G) \ge \phi$ and $h(G) = O(\phi/ r(n))$.  

KRV devised a randomized cut-player strategy that finds the bisection using a stochastic matrix that corresponds to a random walk on all previously discovered matchings. Their walk traverses the previous matchings in order and with probability half takes a step according to each matching. They showed that the matrix corresponding to this random walk can actually be embedded (as a flow matrix) into $G_t$ with constant congestion.  They terminate when the random walk matrix is close to uniform (i.e.\ having constant edge-expansion), resulting in $G_T$ for $T=O\left( \log^2 n\right)$, having constant edge-expansion. 

Orecchia et al.~\cite{orecchia2008partitioning} (abbreviated as OSVV) took the same approach but devised a more sophisticated random walk and used Cheeger's inequality~\cite{cheeger1970lower} in order to show that $G_T$, for $T = O\left(\log^2 n\right)$, has $\Omega\left(\log n\right)$ edge-expansion. That is, they got a ratio of $r(n) = \Omega\left(\frac{1}{\log n}\right)$.

\smallskip

Equipped with this background we now get back to expander decomposition, and focus on the $\tilde{O}(m/\phi)$ time algorithm by Saranurak and Wang~\cite{saranurak2019expander} (abbreviated as SW). Their algorithm is  randomized, follows the recursive scheme described above, and computes a $(\phi, \phi\log^3 n)$-expander decomposition in $O\left(\frac{m \log^4 n}{\phi}\right)$ time. Its number of inter-cluster edges is off by a factor of $O\left(\log^2 n\right)$ from optimal and off by a factor of $O\left(\log^{\frac{3}{2}} n\right)$ from the aforementioned best achievable polynomial time construction.	

One core component of this algorithm is a variation of the cut-matching game (inspired by  R\"{a}cke et al.~\cite{racke2014computing}).
In this variation, the game graph $G_t = (V_t,E_t)$ may lose vertices (\ie, $V_{t+1} \subseteq V_t$) throughout the game and the objective of the cut player is to make $V_T$ a \emph{near expander} in $G_T$ (see Definition \ref{def:expander-near-expander}). The result of each round does not consist of a perfect matching in $V_t$, but rather a subset to remove from $V_t$ and a matching of the remaining vertices. The game ends either with a balanced cut of low conductance, or with an unbalanced cut of low conductance, such that the larger side is a \emph{near expander}. This allows SW to avoid recurring on the large side of the cut. Indeed, if the cut is balanced, they run recursively on both sides, and if it is unbalanced, they use the fact that the large side is a \emph{near expander} and ``trim'' it by finding a large subset of this side which is an expander. Then, they run recursively on the smaller side combined
with the ``trimmed'' vertices.
SW's analysis of the new cut-matching game is based on the ideas and the potential function of KRV while carefully taking into account of the shrinkage of the game graph.

An open question, raised by SW, was whether one can adapt the technique of the cut-matching strategy of OSVV to improve their decomposition. A major obstacle is how to perform an OSVV-like spectral analysis  when we lose vertices throughout the process and need to bound the near-expansion of the final piece. This is challenging as the analysis of OSVV is already somewhat more complicated than that of KRV: It uses a different lazy random walk and a subtle potential to measure progress towards near expansion. Moreover Cheeger's inequality is suitable to show high expansion but the object we are targeting is a near expander.

\smallskip

{\bf Our contribution:} In this paper we answer this question of SW affirmatively. We present and analyze	an expander decomposition algorithm with a new cut-player inspired by OSVV. This improves the result of SW and gives a randomized 
$\tilde{O}(m/\phi)$ time algorithm
for computing an $(\phi, \phi\log^2 n)$-expander decomposition 
(Theorem~\ref{theorem:our_expander_decomposition}). 
This brings the number of inter-cluster edges to be off  only by  $O(\log n)$ factor from the best possible. 

To achieve this we overcome two main technical challenges: (1) We generalize the lazy random walk of the cut player of OSVV and the subtle potential  tracking its progress, to the setting in which the vertex set shrinks (by ripping off of it small cuts as in SW). (2) We show that when the generalized potential is small the remaining part of the game graph is a near expander. This required a generalization of Cheeger's inequality appropriate
for our purpose (see Lemma \ref{lemma:our_F_expander}).

\looseness=-1

Our techniques may be applied in  similar contexts. One concrete such context is the construction of tree-cut sparsifiers. Specifically, one could try to use our technique to improve the $O\left(\log^4 n\right)$-approximate tree-cut sparsifier construction of~\cite{racke2014computing} by a factor of $\log n$. (Note that~\cite{racke2014computing} in fact construct a tree-flow sparsifier, which is a stronger notion.)    

\smallskip

\looseness=-1
The  cut-matching framework \cite{khandekar2009graph} is formalized for edge-expansion rather than conductance. Consequently, SW and others whose primary objective is conductance had to transform the graph into a \emph{subdivision graph} in order to use this framework. 
The subdivision graph is obtained  by adding a new vertex (called a \emph{split-node}) in the middle of each edge $e$, splitting $e$ into a path of length two. The analysis has to translate cuts of low expansion in the modified graph (the subdivision graph) to cuts of low conductance in the original graph. 
This transformation complicates the algorithms and their analysis.

\looseness=-1
 To avoid this transformation we revisit the seminal results of KRV and OSVV and redo them directly for conductance. This is not trivial and requires subtle changes to the cut players, the matching players, and the potentials measuring progress towards a graph with high conductance.
 In particular the matching player does not produce a matching anymore but rather what we call a $\dGG$-matching, which is a graph with the same degrees as $G$. 

 Our new cut-matching algorithm is then described using this natural reformulation of the cut-matching framework directly for conductance, removing the complications that would have followed from using the subdivision graph.   

 We believe that our clean presentation of the cut-matching framework for conductance would prove useful for other applications of cut-matching that 
 require optimization for conductance rather than expansion.

\smallskip

{\bf Further related work:}     
Computing the expansion and the conductance of a graph $G$ is NP-hard \cite{matula1990sparsest, vsima2006np}, and there is a long line of research on approximating these connectivity measures. 
The best known polynomial algorithms for approximating the minimum conductance cut have either $O\left(\sqrt{\log n}\right)$ \cite{arora2009expander, sherman2009breaking} or $O\left(\sqrt{\Phi(G)}\right)$ approximation ratios~\cite{orecchia2012approximating}.
Approximation algorithms for expansion and conductance play a crucial role in algorithms for expander decomposition \cite{saranurak2019expander,chang2019improved,chuzhoy2020deterministic}, expander hierarchies \cite{goranci2021expander,hua2022maintaining}, and tree flow sparsifiers \cite{racke2014computing}.

\looseness=-1
In his thesis, Orecchia~\cite{orecchia2011fast} elaborates on the two cut-matching strategies described in OSVV, one based on a lazy random walk, called $C_{\text{NAT}}$, and a more sophisticated one based on the \emph{heat-kernel} random walk, called $C_{\text{EXP}}$.
Orecchia proves (Theorem 4.1.5 of~\cite{orecchia2011fast}) that using $C_{NAT}$ or $C_{EXP}$, after $T=\Theta(\log^2 n)$ iterations, the graph $G_T$ has expansion $\Omega(\log n)$ (and thereby conductance $\Omega\left(\frac{1}{\log n}\right)$,
since it is regular with degrees $\Theta(\log^2 n)$). Orecchia also bounds the second largest eigenvalue of the normalized Laplacian of $G_T$.
However, Orecchia does not show how to use cut-matching to get approximation algorithms for the conductance of $G$.

In a recent paper \cite{pmlr-v162-orecchia22a} Ameranis \etal~use a generalized notion of expansion, also mentioned in \cite{orecchia2011fast}, where we normalize the number of edges crossing the cut by a general measure ($\mu$) of the smaller side of the cut. 
They define a corresponding generalized version of the cut-matching game, and show how to use a cut strategy for this game to get an approximation algorithm for two generalized cut problems.
They claim that one can construct a cut strategy for this measure using ideas from \cite{orecchia2011fast}.\footnote{The details of such a cut player do not appear in \cite{pmlr-v162-orecchia22a} or \cite{orecchia2011fast}.}

Both SW and our result can be implemented in $\tilde{O}(m)$ time using the recent result of \cite{LNPSsoda13}, by replacing \emph{\unitFlow}~(Lemma \ref{lemma:unit_flow}) and the ``Trimming Step'' of  \cite{saranurak2019expander} with the algorithm of \cite[Section 8]{LNPSsoda13}. This $\tilde{O}(m)$ hides many log factors and requires more complicated machinery.

\smallskip
\looseness = -1

The structure of this paper is as follows.
Section~\ref{section:preliminaries} contains additional definitions.
In order to provide the appropriate context for our work, Section \ref{section:cut-matching-conductance} gives an overview of the cut-matching games in \cite{khandekar2009graph} and 
\cite{orecchia2008partitioning}
and highlights the differences between them. In Appendices \ref{section:krv} and \ref{section:osvv} we give a complete and self-contained description of these approximation algorithms directly {\bf for conductance}.
A reader knowledgeable in the Cut-Matching game can skip directly to Section \ref{section:our-short}.
In Section \ref{section:our-short} we present our new non-stop spectral cut player and expander decomposition algorithm. Section \ref{section:proof-main-theorem} contains the analysis of our algorithm. 
In Appendix \ref{appendix:unit_flow_trimming} we use the new cut player to prove the improvement of the expander decomposition algorithm (Theorem \ref{theorem:our_expander_decomposition}). 
Appendix \ref{appendix:matrix_inequalities} contains algebraic tools used throughout the paper. Appendix \ref{appendix:projection} contains probabilistic lemmas used in the analyses of the cut players. Finally, Appendix \ref{appendix:unit_flow_matching} gives the full details of the \emph{\unitFlow}~algorithm (called \emph{Unit-Flow} in \cite{henzinger2017flow, saranurak2019expander}) which is used in Subsection \ref{section:our-matching-player}.

To be consistent with common terminology we refer to a graph with conductance at least $\phi$ as a $\phi$-expander (rather than $\phi$-conductor).
No confusion should arise since
in the rest of this paper we focus on conductance and do not use the notion of edge-expansion anymore.
In this paper we only focus on unweighted graphs, although our algorithm can be adapted to the case of integral, polynomially bounded weights. 

\section{Preliminaries}
\label{section:preliminaries}

We denote the transpose of a vector or a matrix $x$ by $x'$. That is, if $v$ is a column vector then $v'$ is the corresponding row vector. 
For a vector $v\in \RR_{\ge 0}^n$, define
$\sqrt{v}$ to be vector whose coordinates are the square roots of those of $v$. Given $A\in \RR^{n\times n}$, we denote by $\A{}{i,j}$ the element at the $i$'th row and $j$'th column of $\A$. We denote by $\A{}{i,}, \A{}{,i}$ the {\em column} vectors corresponding to the $i$'th  row and the $i$'th column of $A$, respectively.
We use the abbreviation $\A{}{i}\defeq \A{}{i,}$ only with respect to the rows of $A$. Given a vector $v\in \RR^n$, we denote its $i$'th element by $\vv{i}$. For disjoint $A,B \subseteq V$, we denote by $E_G(A,B)$ the set of edges connecting $A$ and $B$ (if the graph is directed then we consider edges in both directions). We denote by $|E_G(A,B)|$ the number of edges in $E_G(A,B)$, or the sum of their weights if the graph is weighted. We sometimes omit the subscript when the graph is clear from the context. If $A = V \setminus B $, then we call $(A,B)$ a \emph{cut}.
\begin{definition} [$\dGG$,$\vol_G(S)$]
    \label{def:d_G}
        Given a graph $G$, the vector $\dGG\in \RR^n$ is defined as $\dGG{v} = \deg_G(v)$. To simplify the notation, we denote $\dG \defeq \dGG$ whenever the graph $G$ is clear from the context.
        For $S\subseteq V$, we denote by $\vol_G(S) \defeq \sum_{v\in S} \dGG{v}$ the \emph{volume} of $S$. 
    \end{definition}
\begin{definition} [$G\{A\}$]
	\label{def:subgraph_with_loops}
	    Let $G = (V,E)$ be a graph, and let $A\subseteq V$ be a set of vertices. We define the graph $G\{A\}=(V', E')$ as the graph induced by $A$ with self-loops added to preserve the degrees: $V'=A, E'=\{\{u,v\}\in E : u,v\in A\}\cup\{\{u,u\}:u\in A, v\in V\setminus A, \{u,v\}\in E\}$. 

	\end{definition}
\begin{definition}[$d$-Matching]
    \label{def:d-matching}
    Given a vector $d \in \NN^n$ and a collection of pairs $M = \left\{ \{u_i,v_i\} \right\}_{i=1}^{m}$. 

    We say that $M$ is a $d$-\emph{matching} if the graph defined by $M$
    (\ie, the graph whose edges are $M$)
    satisfies $\dG{M}(v) = \dG{}{v}$, for every $v$.
    \end{definition}
\begin{definition} [$\dGG$-stochastic]
    \label{def:d_G_stochastic}
        A matrix $\F\in \RR^{n\times n}$ is \emph{$\dGG$-stochastic} with respect to a graph $G$ if the following two conditions hold: (1)
$ F\cdot\1_n = \dGG $ and (2) $   \1_n'\cdot F = \dGG' $. 
    \end{definition}
\begin{definition}[Laplacian, Normalized Laplacian]
	\label{def:laplacian_normalized_laplacian}
	    Let $A\in \RR^{n \times n}$ be a symmetric matrix and let $d = A\cdot \1_n,\; D= \diag(d)$. The \emph{Laplacian of $A$} is defined as $\mathcal{L}(A) = D - A$. The \emph{normalized-Laplacian of $A$} is defined as $\mathcal{N}(A) = \Dminushalf \mathcal{L}(A) \Dminushalf = I - \Dminushalf A \Dminushalf$. The (normalized) Laplacian of an undirected graph is defined analogously using its adjacency matrix. 
	\end{definition}
\begin{definition}[Conductance]
\label{def:conductance}
	Let $G= (V,E)$ and $S \subset V$, $S\neq\emptyset$. The \emph{conductance} of the cut $(S, V\setminus {S})$, denoted by  $\Phi_{G}(S,V\setminus {S})$, is 
	\[
	\Phi_{G}(S,V\setminus {S}) = \frac{|E(S,V\setminus {S})|}{\min(\vol(S), \vol(V\setminus {S}))}.
	\]
	The conductance of $G$ is defined to be	$ \Phi(G) = \min_{S\subseteq V}\Phi_{G}(S,V\setminus {S}) $.
	\end{definition}
\begin{lemma} [Cheeger Bound~\cite{cheeger1970lower}]
	\label{lemma:normalized_laplacian_expansion}
	    Let $G=(V,E)$ be an undirected graph with $n=|V|$. Denote by $N\in\RR^{n\times n}$ its normalized Laplacian, and denote by $0=\nu_1\le\nu_2\le \ldots \le \nu_n$ the eigenvalues of the normalized Laplacian. Then, $\Phi(G)\ge \frac{\nu_2}{2}$.
	\end{lemma}

\begin{definition}[Expander, Near-Expander]
	\label{def:expander-near-expander}
	Let $G = (V,E)$. We say that $G$ 
	is a \emph{$\phi$-expander} if $\Phi(G) \ge \phi$. Let $A\subseteq V$. We say that $A$ is a \emph{near $\phi$-expander}
	in $G$ if 
	\[
	\min_{S \subseteq A}\frac{|E(S,V \setminus  S)|}{\min(\vol(S), \vol(A\setminus S))}\ge \phi.
	\]
	\end{definition}

 \looseness=-1
	That is, a near expander is allowed to use cut edges that go outside of $A$.

    When we apply these definitions to directed graphs we simply ignore the directions on the edges (and allow parallel edges).

\begin{definition} [Embedding]
    \label{def:embedding}
        Let $G=(V,E)$ be an undirected  graph. Let $F\in \RR^{V\times V}_{\ge 0}$ be a matrix (not necessarily symmetric). We say that $F$ is \emph{embeddable in $G$ with congestion $c$}, if there exists a multi-commodity flow $f$ in $G$, with $|V|$ commodities, one for each vertex (vertex $v$ is the source of its commodity), such that, simultaneously for each $(u,v)\in V\times V$, $f$ routes $F(u,v)$ units of $u$'s commodity from $u$ to $v$, and the total flow on each edge is at most $c$.\footnote{This definition requires to route $F(u,v)=F(v,u)$ both from $u$ to $v$ and from $v$ to $u$ if $F$ is symmetric.} 

        If $F$ is the weighted adjacency matrix of a graph $H$ on the same vertex set $V$, we say that $H$ is \emph{embeddable in $G$ with congestion $c$} if $F$ is embeddable in $G$ with congestion $c$.
    \end{definition}
\begin{lemma}
\label{lemma:near_expansion_and_embedding}
        Let $G,H$ be two graphs on the same vertex set $V$. Let $A\subseteq V$. Let $\alpha>0$ be a constant such that for each $v\in V$, $\dGG{v} = \alpha\cdot \dG{H}{v}$. Assume that $H$ is embeddable in $G$ with congestion $c$, and that $A$ is a near $\phi$-expander in $H$. Then, $A$ is a near $\frac{\phi}{c\alpha}$-expander in $G$.
    \end{lemma}
    \begin{proof}
        Let $S\subseteq A, \bar{S} = V\setminus S$ be a cut, and assume $\vol(S)\le\vol(A\setminus S)$. 
        In the embedding of $H$ in $G$, each edge $(u,v) \in E(H)$ corresponds to $w_H(u,v)$ units of flow, routed in $G$ from $u$ to $v$. Since each edge of $G$ carries at most $c$ units of flow we get that $|E_G(S,\bar{S})|\ge \frac{1}{c}|E_H(S,\bar{S})|$.\footnote{Recall that $|E_H(S,\bar{S})|$ is the sum of the weights of edges crossing the cut $(S, \bar{S})$ in both directions if  $H$ is directed.}

        Additionally, $\vol_G(S)=\alpha\vol_H(S)$, $\vol_G(\bar{S}) = \alpha\vol_H(\bar{S})$. In particular, $\vol_H(S)\le\vol_H(A\setminus S)$. Since $A$ is a near $\phi$-expander in $H$, we get that
        \begin{align*}
            \Phi_G(S,\bar{S}) = \frac{|E_G(S,\bar{S})|}{\vol_G(S)}\ge\frac{1}{c\alpha}\frac{|E_H(S,\bar{S})|}{\vol_H(S)} \ge \frac{\phi}{c\alpha} .
        \end{align*}
    \end{proof}

Note that if $H$ is undirected then our definition of embedding implies that we  route $w_H(u,v)$ units of flow  in $G$ from $u$ to $v$ and from $v$ to $u$. This means that 
$|E_G(S,\bar{S})|\ge \frac{2}{c}|E_H(S,\bar{S})|$ and $A$ is in fact 
near $\frac{2\phi}{c\alpha}$-expander in $G$. 
\begin{corollary}
    \label{cor:expansion_and_embedding}
        Let $G,H$ be two graphs on the same vertex set $V$. Let $\alpha>0$ be a constant such that for each $v\in V$, $\dGG{v} = \alpha\cdot \dG{H}{v}$. Assume that $H$ is embeddable in $G$ with congestion $c$, and that $H$ is a $\phi$-expander. Then, $G$ is a $\frac{\phi}{c\alpha}$-expander.
    \end{corollary}
    \begin{proof}
        This follows from Lemma \ref{lemma:near_expansion_and_embedding} by choosing $A = V$.
    \end{proof}

\section{Approximating conductance via cut-matching}\label{section:cut-matching-conductance}
In preparation for our expander decomposition algorithm  we 
give a high level overview of the conductance approximation algorithms of \cite{khandekar2009graph} and \cite{orecchia2008partitioning}. 
KRV and OSVV described their results for edge-expansion rather than conductance. In Appendices \ref{section:krv} and \ref{section:osvv}, respectively, we give a complete description and analysis of these algorithms for conductance. This translation from edge-expansion to conductance is not trivial as both the cut player, the matching player, and the analysis have to be carefully modified to take the degrees into account. 
Here we give a high level overview of the key components of these algorithms and the differences between them so one can better absorb our main algorithm in Section \ref{section:our}.

The cut-matching game of 
KRV (in the conductance setting)
works as follows.
\begin{center}
\fbox{\parbox{0.95\textwidth}{
        The Cut-Matching game for conductance, with parameters $T$ and a degree vector $\dG$, such that $\sum_{i\in V}{\dG{}{i}} = 2m$:
        \begin{itemize}
            \item The game is played on a series of graphs $G_i$. Initially, $G_0 = \emptyset$.
            \item In iteration $t$, the cut player produces two multisets of size $m$, $L_t,R_t\subseteq V$, such that each $v\in V$ appears in $L_t\cup R_t$ exactly $\dG{}{v}$ times.
            \item The matching player responds with a $\dG$-matching $\M{t}$ that only matches vertices in $L_t$ to vertices in $R_t$.
            \item We set $G_{t+1} = G_t \cup \M{t}$.
            \item The game ends at iteration $T$, and the \emph{quality} of the game is $r \defeq \Phi(G_T)$\footnotemark. 
            Note that the volume of $G_t$ increases by $2m$ from one iteration to the next.

        \end{itemize}
    }}
\end{center}

\footnotetext{Note that this definition of quality corresponds to \emph{conductance} while in the introduction we defined quality with respect to \emph{edge-expansion}.}

The difference between the results of KRV and OSVV is mainly in the cut player.  They both run the game for $T=\Theta(\log^2 n)$ iterations but  KRV's cut player achieves quality of $r \defeq \Phi(G_t) = \Omega\left(\frac{1}{\log^2 n}\right)$ whereas OSVV's achieves quality of $r = \Omega\left(\frac{1}{\log n}\right)$. Notice that the cut player produces the stated conductance of $G_T$ regardless of the matchings given by the matching player. 

\subsection{KRV's Cut-Matching Game for Conductance}
\label{section:krv-short}

The cut player implicitly maintains a $\dGG$-stochastic flow matrix (\ie, representing flow demands)
$\F{t}\in \RR^{n\times n}$, and the graph $G_t$ which is the union of the $\dGG$-matchings that it obtained so far from the matching player ($t$ is the index of the round). The flow $F_t$ and the graph $G_t$ have two crucial properties. First, we can embed $\F{t}$  in $G_t$ with  $O(1)$ congestion (See Definition \ref{def:embedding}). Second, after $T=\Theta(\log^2 n)$ rounds, with high probability, $\F{T}$  will have constant conductance.\footnote{We think about $\F{t}$ as a weighted graph on $V=[n]$. The definitions of conductance, expander and near-expander for weighted graphs are the same as Definitions \ref{def:conductance}-\ref{def:expander-near-expander} where $|E(S, V\setminus S)|$ is the sum of the weights of the edges crossing the cut.} 
\looseness=-1 
Since the degrees in $G_T$ are factor of $O(\log^2 n)$ larger than the degrees in $F_T$ (when we think of $F_T$ as a weighted graph) then it follows by Corollary \ref{cor:expansion_and_embedding} that $G_T$ is $\Omega(1/\log^2 n)$ expander. Note that the cut player is unrelated to the input graph $G$ in which we approximate the conductance. Its goal is to produce the expander $G_T$.

At the beginning, $\F{0} = \D = \diag(\dG)$, and $G_0$ is the empty graph on $V = [n]$. 
The cut player updates $\F{t}$ as follows.
It draws  a random unit vector  $r\in \RR^n$ orthogonal to $\sqrt{\dG}$ and computes the projections $u_i = \frac{1}{\dG{}{i}} \langle D^{-\frac{1}{2}}\F{t}{i}, r \rangle$.\footnote{Recall that $\F{t}{i}$ is a column vector.} The cut player computes these projections in  $O(m\log^2 n)$ time since the vector of all projections is $u \defeq D^{-1} \F{t}\Dminushalf \cdot r$ and $F_t$ is defined (see below) as a multiplication of $\Theta(\log^2 n)$ sparse matrices, each having $O(m)$ non-zero entries. The cut player sorts the projections as $u_{i_1} \le ... \le u_{i_n}$. Consider the sequence $Q = (i_1, i_1, \ldots, i_1, i_2, i_2, \ldots, i_2, \ldots, i_n, \ldots, i_n)$, where each $i_j$ appears $\dG{}{i_j}$ times. Then, $|Q|=2m$. Take $L_t\subseteq Q$ to be the multi-set containing the first $m$ elements, and $R_t=Q \setminus L_t$ to be the multi-set containing the last $m$ elements. Define $\eta\in\RR$ such that $L_t \subseteq \{i_k : u_{i_k} \le \eta\}$ and $R_t\subseteq \{i_k : u_{i_k} \ge\eta\}$. Note that a vertex can appear both in $L_t$ and in $R_t$, if $u_{i_j} = \eta$.
For a vertex $v\in V$, denote by $m_v$ the number of times $v$ appears in $L_t$, and by $\bar{m}_v$ the number of times $v$ appears in $R_t$. That is, except for (maybe) one vertex, for any $v\in V$, either $m_v = 0$ and $\bar{m}_v = \dG{}{v}$ or $m_v = \dG{}{v}$ and $\bar{m}_v = 0$. 

\looseness=-1
The cut player hands out the partition 
$L_t$, $R_t$ to the matching player who sends back a $\dGG$-matching $\M{t}$ (we think of $\M{t}$ as an $n\times n$ matrix with at most $2m$ non-zero entries that encodes the matching) between $L_t$ and $R_t$. The cut player updates its flow matrix using $M_t$ and sets $\F{{t+1}}(v) = \frac{1}{2}\F{t}(v) + \sum_{\{v,u\}\in \M{t}}{\frac{1}{2\dG{}{u}} \F{t}(u)}$ (in matrix form $\F{{t+1}} = \frac{1}{2}\left(I + \M{t}\cdot D^{-1}\right)\F{t}$).\footnote{ 
    Note that it is possible that some $u\in V$ appears in the sum $\sum_{\{v,u\}\in \M{t}}{\frac{1}{2\dG{}{u}} \F{t}(u)}$ multiple times, if $v$ is matched to $u$ multiple times in $M_t$.}
 This update keeps $\F{t}$ a $\dGG$-stochastic matrix for all $t$ (see Lemma~\ref{lemma:krv_F_t_embeddable_in_G_t}).  The cut player
also defines the graph $G_{t+1}$ as $G_{t+1} \defeq G_t\cup \M{t}$.
This completes the description of  the cut player of KRV adapted for conductance.

The matching player constructs an auxiliary flow problem on $G' \defeq G \cup \{s,t\}$, where $s$ is a new vertex which would be the source and $t$ is a new vertex which would be the sink. It adds an edge $(s,v)$ of capacity $m_v$ for each $v\in L_t$ and an edge $(v,t)$ of capacity $\bar{m}_v$ for each $v\in R_t$. The capacity of each edge $e\in G$ is set to be $c=\Theta\left(\frac{1}{\phi\log^2 n}\right)$, where $c$ is an integer. The matching player computes a maximum flow $g$ from $s$ to $t$ in this network.

If the value of $g$ is less than $m$, then the matching player uses the minimum cut in $G'$ separating the source from the sink to find a cut in $G$ of conductance $O(\phi\log^2 n)$ (see Lemma \ref{lemma:krv_3.7_small_flow}).  Otherwise, it decomposes $g$ to a set of paths, each carrying exactly one unit of flow from a vertex $u\in L_t$ to a vertex $v\in R_t$.\footnote{Note that there can be multiple flow paths between a pair of vertices $u\in L_t$ and $v\in R_t$. Furthermore, if $u\in L_t\cap R_t$  then it is possible that a path starts and ends at $u$. } Then, it
defines the $\dGG$-matching $\M{t}$ as $\M{t} = \{\{v_j, u_j\}\}_{j=1}^m$, where $v_j$ and $u_j$ are the endpoints of path $j$. We view $\M{t}$ as a symmetric $n\times n$ matrix, such that $\M{t}{v, u}$ is the number of paths between $v$ and $u$. 

The matching player connects the game to the input graph $G$. Indeed, by solving the maximum flow problems in $G$ it guarantees that the expander $G_T$ is embeddable
in $G$ with congestion $O(cT)=O(1/\phi)$.
Since the degrees of $G_T$ are a factor of $O(\log^2 n)$ larger than the degrees of $G$ and $G_T$ is $\Omega(1/\log^2 n)$ expander, we get that $G$ is a $\Omega(\phi)$-expander (see Corollary \ref{cor:expansion_and_embedding}).
The algorithm is summarized in Algorithm \ref{algo:krv_cut_matching} in Appendix \ref{appendix:krv}.
The following theorem summarizes the properties of this algorithm. 
\begin{theorem}[\cite{khandekar2009graph}'s cut-matching game for conductance]
    \label{theorem:krv}
        Given a graph $G$ and a parameter $\phi > 0$, there exists a randomized algorithm, whose running time is dominated by computing a polylogarithmic number of maximum flow problems, that either
        \begin{enumerate}
            \item Certifies that $\Phi(G)= \Omega(\phi)$ with high probability; or
            \item Finds a cut $(S, V\setminus {S})$ in $G$ whose conductance is $\Phi_G(S,V\setminus {S}) = O(\phi\log^2 n)$.
        \end{enumerate}    
\end{theorem}

\looseness=-1
If the matching player finds a sparse cut in any iteration then we terminate with Case (2). On the other hand, if  the game continues for
$T=O(\log^2 n)$ rounds then since 
the  cut player can embed $F_T$ in
$G_T$ and the matching player can embed $G_T$ in $G$,
and since $F_t$ is an expander, then we get Case (1).

The running time of the cut player is $O(m\log^4 n)$. The matching player solves $O(\log^2 n)$ maximum flow problems.    By using the most recent maximum flow algorithm of~\cite{chen2022maximum}, we get the matching player to run  in $O\left( m^{1+ o(1)} \right)$  time. 
Alternatively, we can adapt the cut-matching game, and use a version of the \emph{\unitFlow} algorithm described in Appendix \ref{appendix:unit_flow_matching} to get a running time of $\tilde{O}(\frac{m}{\phi})$ for the matching player. We can also get $\tilde{O}(m)$ running time using the recent result \cite{LNPSsoda13}. 

The key part of the analysis is to show that 
$F_T$ is indeed an $\Omega(1)$-expander with high probability for any choice of $\dGG$-matchings of the matching player. To this end, we keep track of the progress of the cut player using the potential function
\[\psi(t) = \sum_{i\in V}{\sum_{j\in V}{\frac{1}{\dG{}{i}\cdot \dG{}{j}}\left(\F{t}{i,j}-\frac{\dG{}{i} \dG{}{j}}{2m}\right)^2}}= \norm{D^{-\frac{1}{2}} \F{t} D^{-\frac{1}{2}}-\frac{1}{2m}\sqrt{\dG}\sqrt{\dG'}}^2_F
\]
 where the matrix norm which we use here is the Frobenius norm (sum of the squares of the entries). 
 This potential represents the distance between the normalized flow matrix $\bar{F_t} = \Dminushalf \F{t} \Dminushalf$ and the (normalized) uniform random walk distribution $\Dminushalf \cdot (\dGG \dGG'/2m) \cdot \Dminushalf$.

 Let  $P = I - \frac{1}{2m}\sqrt{\dG}\sqrt{\dG'}$ be the projection matrix on the orthogonal complement of the span of the vector $\sqrt{\dG}$, then we can also write this potential as 
      \begin{align*}
        \psi(t) =
        \norm{\bar{\F{t}} P}^2_F = 
        \tr \left((\bar{\F{t}}  P)(\bar{\F{t}} P)' \right) 
        = \tr( \bar{\F{t}} P^2 \bar{\F{t}'}) = \tr(  P \bar{\F{t}'} \bar{\F{t}}).
    \end{align*}
The first equality holds by Lemma~\ref{lemma:normalized_commutes_with_P} and the last equality is due to Fact~\ref{fact:book_trace_identities} (and that $P^2=P$ as a projection matrix).

The crux of the proof is to show that after $T$ rounds, with high probability, this potential is smaller than $1/(16m^2)$ which implies that for every pair of vertices $u$ and $v$, $F_T(u,v)\ge d(v)d(u)/(4m)$.
From this we get a lower bound of $1/4$ on the conductance of every cut. 

The full details of the analysis appear in Appendix \ref{section:krv}. 

 \subsection{OSVV's Cut-Matching Game for Conductance}\label{section:osvv-short}

The cut player of OSVV for conductance also maintains (implicitly) a flow matrix $F_t$ and the union $G_t$ of the $\dGG$-matchings it got from the matching player.
Let $P = I - \frac{1}{2m}\sqrt{\dG} \sqrt{\dG'}$ be the projection to the subspace orthogonal to  $\sqrt{\dG}$ as before  (hence $P^2 = P$). 
Let $\delta = \Theta(\log{n})$ be a power of $2$. 
Here the matrix  $\W{t} = (P\normalized{\F{t}}P)^{\delta}$
takes the role of $D^{-\frac{1}{2}}F_tD^{-\frac{1}{2}}$ from the cut player of Section \ref{section:krv-short}.

In round $t$ the cut player samples a random unit vector $r\in \RR^n$, computes the projections $u_i = \frac{1}{\sqrt{\dG{}{i}}}\langle \W{t}{i}, r \rangle$, and 
defines $L_t$ and $R_t$ based on these projections as in the previous section.\footnote{Computing these projections takes $O(m\log^3 n)$ time since $F_t$ is a multiplication of $\Theta(\log^2 n)$ sparse matrices, each with $O(m)$ non-zero entries. Therefore $W_t$ is a multiplication of $\Theta(\log^3 n)$ matrices, each of which is either $P$ or a sparse matrix.}
Note that here, $r$ does not need to be orthogonal to $\sqrt{\dG}$. 

Then it gets from the matching player a $\dGG$-matching $\M{t}$  between
$L_t$ and $R_t$. It defines $\N{t} = \frac{\delta - 1}{\delta}D + \frac{1}{\delta}\M{t}$ and updates the flow to be $\F{{t+1}} = \N{t} D^{-1} \F{t} D^{-1} \N{t}$. If we think of $F_t$ as  a random walk then $D^{-1} \N{t}$ is a ``lazy step'' that we add before and after the walk $F_t$ to get $F_{t+1}$.
It holds that $F_{t+1}$ is $\dGG$-stochastic and moreover that for all rounds $t$, $\F{t}$ is embeddable in $G_t$ with congestion $\frac{4}{\delta}=O(1/\log n)$. Note that here we embed $F_t$ in $G_t$ with smaller congestion  than in Section \ref{section:krv-short}. We can still prove, however, that 
$F_T$ for $T=O(\log^2 n)$ is a $\Omega(1)$-expander with high probability, and therefore, $G_T$ is a $\Omega(1/\log n)$-expander.

The matching player solves the same flow problem as in Section 
\ref{section:krv-short}
but with an integer capacity  value of $c=\Theta(\frac{1}{\phi\log n})$ on the edges of $G$.  If the value of maximum flow is less than $m$ then it finds a cut of   conductance $O(\phi\log n)$, and otherwise it returns the matching that it derives from a decomposition of the flow into paths.
The matching player guarantees that the expander $G_T$ is embeddable
in $G$ with congestion $O(cT)=O(\log n /\phi)$.
Since the degrees of $G_T$ are larger by  a factor of $O(\log^2 n)$ than the degrees of $G$ and $G_T$ is $\Omega(1/\log n)$-expander, we get that $G$ is a $\Omega(\phi)$-expander (see Lemma \ref{lemma:near_expansion_and_embedding}).
The following theorem summarizes the properties of this algorithm. 

\begin{theorem}[\cite{orecchia2008partitioning}'s cut-matching game for conductance]
    \label{theorem:osvv}
        Given a graph $G$ and a parameter $\phi > 0$, there exists a randomized algorithm, whose running time is dominated by computing a polylogarithmic number of maximum flow problems, that either
        \begin{enumerate}
            \item Certifies that $\Phi(G)= \Omega(\phi)$ with high probability; or
            \item Finds a cut $(S, V\setminus {S})$ in $G$ whose conductance is $\Phi_G(S,V\setminus {S}) = O(\phi\log n)$.
        \end{enumerate}    
    \end{theorem}

The running time of the cut player is dominated by computing the projections in $O(m\log^3 n)$ time per iteration, for a total of  $O(m\log^5 n)$ time. 
The matching player solves $O(\log^2 n)$ maximum flow problems. 
Again, we can modify the algorithm so that its running time is $\tilde{O}(\frac{m}{\phi})$ or $\tilde{O}(m)$, similarly to the previous subsection.

As in Section \ref{section:krv-short},
the key part of the analysis is to show that 
$F_T$ is indeed an $\Omega(1)$-expander with high probability for any choice of $\dGG$-matchings of the matching player. Here we keep track of the progress of the cut player using the potential function
\begin{align*}
    \psi(t) &= \norm{(\normalized{\F{t}})^\delta - \frac{1}{2m}\sqrt{\dG} \sqrt{\dG'}}^2_F.
\end{align*}

Recall that $\W{t} = (P\normalized{\F{t}}P)^{\delta}$, so we can rewrite the potential function as
    \begin{align*}
        \psi(t) 
        = \norm{(\normalized{\F{t}})^\delta P}^2_F =
        \tr(P (\normalized{\F{t}})^{2\delta} P) \numeq{4} \tr( (P\normalized{\F{t}}P)^{2\delta}) = \tr(W^2_t) \ ,
    \end{align*}
where equality $(4)$ follows from Lemma~\ref{lemma:normalized_commutes_with_P}, the fact that $\F{t}$ is $\dG$-stochastic (Lemma \ref{lemma:osvv_basic_properties}(1)), and the fact that $P^2 = P$.
A careful argument shows that after $T=O(\log^2 n)$ iterations, $\psi(T)\le 1/n$ with high probability. 
From this we deduce that the second smallest eigenvalue of 
the normalized Laplacian of $F_T$  is at least $1/2$ and then by Cheeger's inequality~\cite{cheeger1970lower} we get that $\Phi(\F{T}) = \Omega(1)$.

The full details of the analysis appear in Appendix \ref{section:osvv}.

\section{Expander decomposition via spectral Cut-Matching}
    \label{section:our-short}
To put our main result in context we first show how SW
\cite{saranurak2019expander}
modified the cut-matching game of KRV for their expander decomposition algorithm.

\subsection{SW's Cut-Matching for expander decomposition}\label{section:SW-overview}
\looseness = -1
SW take a recursive approach to find an expander decomposition. One can use the cut-matching game to find a sparse cut, but if the cut is unbalanced, we want to avoid recursing on the large side.

In order to refrain from recursing on the large side of the cut, 
SW changed the cut-matching game as follows.
The cut player 
  now maintains
a partition of $V$ into a small set $R$ and  a large set $A = V \setminus R$, where initially $R = \emptyset$ and  $A = V$. In each iteration, the cut and the matching player interact as follows.
\begin{itemize}
        \item The cut player computes two disjoint sets $A^l,A^r\subseteq A$ such that $|A^l| \le n/8$ and $|A^r| \ge n/2$.

        \item The matching player  returns a partition  $(S, A\setminus S)$ of $A$, which may be empty ($S=\emptyset$), and a matching of $A^l\setminus S$ to a subset of $A^r\setminus S$.
\end{itemize}

\looseness=-1
The cut player computes the sets $A^l$ and $A^r$  by projecting the rows of a  \emph{flow-matrix} $F$ that it maintains (as in KRV) onto a random unit vector $r$, and applying a result by \cite{racke2014computing} to generate the sets $A^l$ and $A^r$ from the values of the projections. For the matching player, SW use a flow-based algorithm which simultaneously gives a  cut $(S,A\setminus S)$ of conductance $O(\phi \log^2 n)$ of $G[A]$, and a matching of the vertices left in $A^l\setminus S$ to vertices of $A^r\setminus S$ ($S$ may be empty when $G[A]$ has conductance $\ge \phi$). If the matching player found a sparse cut $(S,A\setminus S)$ then the cut player updates the partition $(R,A)$ of $V$ by moving $S$ from $A$ to $R$.

The game terminates either when 
the volume of $R$ gets larger than $\Omega(m/\log^2 n)$ or after 
$O(\log^2 n)$ rounds.
In the latter case, SW proved that the remaining set $A$ (which is large) is  a near $\phi$-expander in $G$ with high probability (see Definition \ref{def:expander-near-expander}).

\looseness=-1
To prove that after $T=\Theta(\log^2 n)$ iterations, the remaining set $A$ is a near $\phi$-expander, SW essentially followed the footsteps of KRV and used a similar potential. The argument is more complicated since they have to take the shrinkage of $A$  into account.
SW did not use a version of KRV suitable to conductance as we give in Appendix \ref{appendix:krv}. Therefore, they had to modify the graph by adding a split node for each edge, essentially reducing conductance to edge-expansion, a reduction that further complicated the analysis.
The following theorem summarizes the properties of the cut-matching game of SW.
\begin{theorem}[Theorem $2.2$ of~\cite{saranurak2019expander}]
	\label{theorem:2.2-SW}
        Given a graph $G=(V,E)$ of $m$ edges and a parameter $0< \phi < 1/\log^2n$,\footnote{The theorem is trivial if $\phi \ge \frac{1}{\log^2 n}$, because any cut $(A, V\setminus {A})$ has conductance $\Phi_G(A, V\setminus {A}) \le 1$.
        We can therefore assume that $\phi < \frac{1}{\log^2 n}$.} there exists a randomized algorithm, called ``the cut-matching step'', which takes $O\left( (m\log n)/\phi\right)$ time and terminates in one of the following three cases:
        \begin{enumerate}
            \item We certify that $G$ has conductance $\Phi(G)=\Omega(\phi)$ with high probability.
            \item We find a cut $(R,A)$ of $G$ of conductance $\Phi_G(R,A)=O(\phi\log^2 n)$, and $\vol(R), \vol(A)$ are both $\Omega(\frac{m}{\log^2 n})$, \ie, we find a relatively balanced low conductance cut.
            \item We find a cut $(R,A)$ of $G$ with $\Phi_G(R,A)\le c_0 \phi \log^2 n$ for some constant $c_0$, and $\vol(R)\le \frac{m}{10c_0\log^2 n}$, and with high probability $A$ is a near $\phi$-expander in $G$.
        \end{enumerate}
    \end{theorem}

SW derived an expander decomposition algorithm from this modified cut-matching game by recursing on both sides of the cut only if Case (2) occurs. In Case (3) they find a large subset $B\subseteq A$ which is an expander (in what they called the \emph{trimming step}), add $A\setminus B$ to $R$ and recur only on $R$.
 The main result of SW is as follows.
\begin{theorem}[Theorem~$1.2$ of~\cite{saranurak2019expander}]
    \label{theorem:SW-main}
    Given a graph $G = (V,E)$ of $m$ edges and a parameter $\phi$, there is a randomized algorithm that with high probability finds a partitioning of $V$ into clusters $V_1,\ldots, V_k$ such that $\forall i: \Phi_{G\{V_i\}}= \Omega(\phi)$ and there are at most $O(\phi m \log^3{n})$ inter cluster edges.\footnote{$G\{V_i\}$ is defined in Definition \ref{def:subgraph_with_loops}.} The running time of the algorithm is $O(m\log^4{n} / \phi)$. 
    \end{theorem}

\subsection{Our contribution: Spectral cut player for expander decomposition}
\label{section:our}
SW left open the question if one can improve their expander decomposition algorithm using tools similar to the ones that allowed OSVV to improve the edge-expansion approximation algorithm of KRV.
We give a positive answer to this question. Specifically, we improve the cut-matching game of SW and derive the following improved version of Theorem \ref{theorem:2.2-SW}.
\begin{theorem}
    \label{theorem:our}
    \label{theorem:our_cut_matching}
    Given a graph $G=(V,E)$ of $m$ edges and a parameter $0<\phi<\frac{1}{\log n}$,\footnote{         The theorem is trivial if $\phi \ge \frac{1}{\log n}$, because any cut $(A, V\setminus {A})$ has conductance $\Phi_G(A, V\setminus {A}) \le 1$.
        We can therefore assume that $\phi < \frac{1}{\log n}$.}
        there exists a randomized algorithm which takes $O\left(m\log^5 n + \frac{m\log^2 n}{\phi}\right)$ time and must end in one of the following three cases:
        \begin{enumerate}
            \item We certify that $G$ has conductance $\Phi(G)=\Omega(\phi)$ with high probability. 
            \item We find a cut $(R,A)$ in $G$ of conductance $\Phi_G(R,A)=O(\phi\log n)$, and $\vol(R), \vol(A)$ are both $\Omega(\frac{m}{\log n})$, i.e, we find a relatively balanced low conductance cut.
            \item We find a cut $(R,A)$ with $\Phi_G(R,A)\le c_0 \phi \log n$ for some constant $c_0$, and $\vol(R)\le \frac{m}{10c_0\log n}$, and with high probability $A$ is a near $\Omega(\phi)$-expander in $G$.
        \end{enumerate}
    \end{theorem}

The proof of Theorem  \ref{theorem:our_cut_matching} is given in Section \ref{section:proof-main-theorem}. 
Theorem \ref{theorem:our_cut_matching} implies the following theorem.
\begin{theorem}
    \label{theorem:our_expander_decomposition}
        Given a graph $G = (V, E)$ of $m$ edges and a parameter $\phi$, there is a randomized algorithm that with high probability finds a partition of $V$ into clusters $V_1,..., V_k$ such that $\forall i: \Phi_{G\{V_i\}} = \Omega(\phi)$ and $\sum_{i}{|E(V_i, V\setminus {V_i})|} = O(\phi m \log^2 n)$. The running time of the algorithm is $O(m\log^7 n + \frac{m \log^4  n}{\phi})$.\footnote{Note that if $\phi \le \frac{1}{\log^3 n}$, then the running time matches the running time of SW in Theorem \ref{theorem:SW-main}. In case that $\phi\ge \frac{1}{\log^3 n}$, we get a slightly worse running time of $O(m\log^7 n)$ instead of $O(\frac{m \log^4 n}{\phi})$.}
    \end{theorem}

    The derivation of Theorem \ref{theorem:our_expander_decomposition}
    using Theorem \ref{theorem:our_cut_matching} uses the same methods as in SW, and we include the proof for completeness in Appendix \ref{appendix:unit_flow_trimming}.

    \begin{remark}
    \label{remark:fair_cut_global}
        We can implement Theorems \ref{theorem:our_cut_matching} and \ref{theorem:our_expander_decomposition} in $\tilde{O}(m)$ time, by replacing \emph{\unitFlow}~(Lemma \ref{lemma:unit_flow}) and Theorem \ref{theorem:SW19_2.1} with the algorithm of \cite[Section 8]{LNPSsoda13}.
    \end{remark}

To get Theorem \ref{theorem:our_cut_matching}
we use the cut player and matching player described in the following sections. Their summary is also presented in Algorithm \ref{algo:our-update}.

\subsection{Cut player}
    \label{section:our_update_F}
    Like in Section \ref{section:cut-matching-conductance}, we consider a $\dG$-stochastic flow matrix $\F{t}\in \RR^{n\times n}$, and a series of graphs $G_t$. $\F{0}$ is initialized as $\F{0} = D \defeq \diag(\dG)$, and $G_0$ is initialized as the empty graph on $V = [n]$.  Here the cut player also maintains a low conductance cut $A_t\subseteq V, R_t = V\setminus A_t$, such that after $T=\Theta(\log^2 n)$ rounds, with high probability, $A_T$ is a near expander in $G_T$. At the beginning, $A_0 = V$ and $R_0 = \emptyset$.

    Since the new cut-matching game consists of iteratively shrinking the domain $\A{t}\subseteq V$, we start by generalizing our matrices from Section \ref{section:cut-matching-conductance} to this context of shrinking domain.
\begin{definition}[$\I{t},\dG{t},\D{t},\Pmat{t},\vol_t$]
\label{def:blocked_variables}
        We define the following variables:\footnote{These variables are the analogs of $I, \dG, \D, \vol(G)$ and $\Pmat$ (respectively) from Section \ref{section:osvv-short} in $G[A_t]$.}
        \begin{enumerate}
            \item $\I{t}\in\RR^{n\times n}$ is the diagonal $0/1$ matrix that has $1$'s on the diagonal entries corresponding to $A_t$.
            \item $\dG{t} = \I{t} \cdot \dG \in \RR^n$, i.e the projection of $\dG$ onto $A_t$.
            \item $\D{t} = \I{t} \cdot D = \diag(\dG{t}) \in \RR^{n\times n}$.
            \item $\vol_t = \vol_G(A_t)$.
            \item $\Pmat{t} = \I{t} - \frac{1}{\vol_t}\sqrt{\dG{t}}\sqrt{\dG{t}'} \in \RR^{n\times n}$.
        \end{enumerate}
    \end{definition}

    We define the matrix  $\W{t} = (\Pmat{t}\normalized{\F{t}}\Pmat{t})^{\delta}$, where 
$\delta = \Theta(\log{n})$ is set in the proof of Lemma \ref{lemma:our_F_expander},  that plays a crucial role in  this section.
  This definition is similar to the definition of $W_t$ in Section \ref{section:osvv-short}, but with $\Pmat{t}$ instead of $\Pmat$. This makes us ``focus'' only on the remaining vertices $A_t$, as any row/column of $\W{t}$ corresponding to a vertex $v\in R_t$ is zero. 
    The matrix $W_t$ is used in this section to define the projections that our algorithm uses to update $F_t$. It is also used in Section \ref{section:our_A_expander_in_F} to define the potential that measures how far is the remaining part of the graph from a  near expander. In particular, we show in Lemma~\ref{lemma:our_F_expander} and Corollary \ref{cor:our_G_expander} that if $W_T^2$ has small eigenvalues (which will be the case when the potential is small) then $A_T$ is a near-expander in $G_T$.

    Let $r\in \RR^n$ be a random unit vector.
    Consider the projections $u_i = \frac{1}{\sqrt{\dG{}{i}}}\langle \W{t}{i}, r \rangle$, for $i\in A_t$. Note that because $\Pmat{t} \sqrt{\dG{t}} = 0$, and $\W{t}$ is symmetric:
    {\small}
    \begin{align*}
        \sum_{i\in A_t}{d(i) u_i} = \sum_{i\in A_t}{\sqrt{d(i)}\left\langle\W{t}{i}, r\right\rangle} = \left\langle \sum_{i\in A_t}{\sqrt{d(i)}\W{t}{i}}, r \right\rangle = \left\langle \W{t}\sqrt{\dG_{t}}, r \right\rangle = 
        \left\langle 0, r \right\rangle = 
        0
    \end{align*}

    We use the following lemma to partition (some of) the remaining vertices into two multisets $A_t^l$ and $A_t^r$.\footnote{Note that this does not produce a bisection of $V$ or of $A_t$.} The lemma follows by applying Lemma 3.3 in \cite{racke2014computing} on the multiset of the $u_i$'s, where each $u_i$ appears with multiplicity of $\dG{}{i}$.

\begin{lemma} [Lemma 3.3 in \cite{racke2014computing}]
	\label{lemma:RST}
	    Given $u_i\in \RR$ for all $i\in A_t$, such that $\sum_{i\in A_t}{\dG{}{i} u_i} = 0$, 
	    we can find in time $O(|A_t|\cdot \log (|A_t|))$ a multiset of source nodes $A^l_t\subseteq A_t$, a multiset of target nodes $A^r_t\subseteq A_t$, and a separation value $\eta$ such that each $i\in A_t$ appears in $A^l_t \cup A^r_t$ at most $\dG{}{i}$ times, 
	    and additionally:
	    \begin{enumerate}
	        \item $\eta$ separates the sets $A^l_t, A^r_t$, \ie,\ either $\max_{i\in A^l_t}{u_i}\le \eta \le \min_{j\in A^r_t}{u_j}$, or $\min_{i\in A^l_t}{u_i}\ge \eta \ge \max_{j\in A^r_t}{u_j}$,
	        \item $|A^r_t|\ge \frac{\vol_t}{2}$, $|A^l_t|\le \frac{\vol_t}{8}$,
	        \item $\forall i\in A^l_t: (u_i-\eta)^2\ge \frac{1}{9}u_i^2$,
	        \item $\sum_{i\in A^l_t}{m_i u_i^2}\ge \frac{1}{80}\sum_{i\in A_t}{\dG{}{i} u_i^2}$, where $m_i$ is the number of times $i$ appears in $A^l_t$.\footnote{Note that even though $A^l_t$ is a multiset, in the sum $\sum_{i\in A^l_t}$, we iterate over each element once.}
	    \end{enumerate}
	\end{lemma}

    Note that a vertex could appear both in $A_t^l$ and in $A_t^r$, if $u_{i_j} = \eta$. The cut player sends $A^l_t,A^r_t$ and $A_t$ to the matching player. In turn, the matching player (see Subsection \ref{section:our-matching-player}) returns a cut $(S_t,A_t\setminus S_t$) and a $\dGG$-stochastic matrix $\M{t}$ ($\M{t}$ is obtained from a matching $\tilde{\M{t}}$ of $A^l_t \setminus S_t$ to $A^r_t \setminus S_t$ by adding self-loops). 

Define $\N{t} = \frac{\delta - 1}{\delta}D + \frac{1}{\delta}\M{t}$. The cut player then updates $\F{{t}}$ similarly to Section \ref{section:osvv-short}: $\F{{t+1}} = \N{t} D^{-1} \F{t} D^{-1} \N{t}$. Like in the previous sections, we also define the graph $G_{t+1}$ as $G_{t+1} = G_t\cup \M{t}$.\footnote{$G_{t+1}$ may have self-loops.} We define $\A{t+1} = \A{t}\setminus S_t$.

    \subsection{Matching player}\label{section:our-matching-player}

    The matching player receives $A^l_t$ and $A^r_t$ and the current $A_t$. 
    For a vertex $v\in V$, denote by $m_v$ the number times $v$ appears in $A_t^l$, and by $\bar{m}_v$ the number of times $v$ appears in $A_t^r$.
    The matching player solves the flow problem on $G[A_t]$, specified by Lemma \ref{lemma:unit_flow} below. 
    This lemma is similar to Lemma B.6 in~\cite{saranurak2019expander} and is proved using the \emph{\unitFlow} algorithm (called \emph{Unit-Flow} by \cite{henzinger2017flow, saranurak2019expander}). For completeness, the proof of this lemma appears in Appendix~\ref{appendix:unit_flow_matching}. Note that we can get a running time of $\tilde{O}(m)$ mentioned in Remark \ref{remark:fair_cut_global} by replacing this subroutine   with a fair-cut computation as shown in \cite[Section 8]{LNPSsoda13}.

\begin{lemma}
	\label{lemma:unit_flow}
	    Let $G=(V,E)$ be a graph with $n$ vertices and $m$ edges, let $A^l, A^r \subseteq V$ be multisets such that $|A^r| \ge \frac{1}{2}m, |A^l| \le \frac{1}{8}m$, and let $0<\phi<\frac{1}{\log n}$ be a parameter. For a vertex $v\in V$, denote by $m_v$ the number times $v$ appears in $A^l$, and by $\bar{m}_v$ the number of times $v$ appears in $A^r$. Assume that $m_v + \bar{m}_v\le \dG{}{v}$ for all $v\in V$. We define the flow problem $\Pi(G)$, as the problem in which a source $s$ is connected to each vertex $v\in A^l$ with an edge of capacity $m_v$ and each vertex $v\in A^r$ is connected to a sink $t$ with an edge of capacity $\bar{m}_v$. Every edge of $G$ has the same capacity $c=\Theta\left(\frac{1}{\phi\log n}\right)$, which is an integer. A feasible flow for $\Pi(G)$ is a maximum flow that saturates all the edges outgoing from $s$.
	    Then, in time $O(\frac{m}{\phi})$, we can find either
	    \begin{enumerate}
	        \item A feasible flow $f$ for $\Pi(G)$; or
	        \item A cut $S$ where $\Phi_{G}(S, V\setminus S)\le\frac{7}{c}=O(\phi\log n)$, $\vol(V\setminus S) \ge \frac{1}{3}m$ and a feasible flow for the problem $\Pi(G-S)$, where we only consider the sub-graph $G[V\setminus S \cup \{s, t\}]$ (that is, vertices $v\in A^l\setminus S$ are sources of $m_v$ units, and vertices $v\in A^r\setminus S$ are sinks of $\bar{m}_v$ units). 
	    \end{enumerate}
	\end{lemma}

	\begin{remark}
	    It is possible that $A^l\subseteq S$, in which case the feasible flow for $\Pi(G-S)$ is trivial (there are no edges leaving $s$). 
	\end{remark}
    Let $S_t$ be the cut returned by the lemma. If the lemma terminates with the first case, we denote $S_t = \emptyset$. Since $c$ is an integer, we can decompose the returned flow into a set of paths (using \eg{} dynamic trees~\cite{ST83}), each carrying exactly one unit of flow from a vertex $u\in A_t^l\setminus S_t$ to a vertex $v\in A^r_t\setminus S_t$. Note that multiple paths can route flow between the same pair of vertices. 
    If $u\in A^l_t\cap A^r_t$ then it is possible that a path starts and ends at $u$. Each $u\in A^l_t\setminus S_t$ is the endpoint of exactly $m_u\le\dG{}{u}$ paths, and each $v\in A^r_t\setminus S_t$ is the endpoint of \emph{at most} $\bar{m}_v\le\dG{}{v}$ paths. Define the ``matching''\footnote{Note that this is \textbf{not} a matching or a $\dG$-matching, but rather a graph that connects vertices of $A_t^l$ to vertices of $A_t^r$, whose degrees are bounded by $\dG$.}
    $\tilde{\M{t}}$ as $\tilde{\M{t}} = \{\{u_i, v_i\}\}_{i=1}^{|A^l_t\setminus S_t|}$, where $u_i$ and $v_i$ are the endpoints of path $i$. We can view $\tilde{\M{t}}$ as a symmetric $n\times n$ matrix, such that $\MatrixSimple{\tilde{\M{t}}}{u, v}$ is the number of paths between $u$ and $v$.
    We turn $\tilde{\M{t}}$ into a $\dG$-stochastic matrix by increasing its diagonal entries by $\dG - \tilde{\M{t}}\1_n$. Formally, we set $\M{t} \defeq \tilde{\M{t}} + \diag(\dG - \tilde{\M{t}}\1_n)$. Notice that $\dG - \tilde{\M{t}}\1_n$ has only non-negative entries, so $\M{t}$ also has non-negative entries. 
    Additionally, note that due to the definition of $\M{t}$, it is embeddable in $G$ with congestion $c$. 

    Intuitively, we can think of $\M{t}$ as the response of the matching player to the subsets $A^l_t$ and $A^r_t$ given by the cut player. 

\section{Analysis of the spectral cut player}\label{section:proof-main-theorem} \label{section:rest-of-analysis}

This section is organized as follows. Subsection \ref{section:our-algorithm} presents in detail the algorithm for Theorem \ref{theorem:our_cut_matching}, using the cut player and matching player described in Subsections \ref{section:our_update_F} and \ref{section:our-matching-player}. Subsection \ref{section:our_F_embeddable_G} shows that $\F{t}$ is embeddable in $G_t$ with congestion $\frac{4}{\delta}$ and that $G_t$ is embeddable in $G$ with congestion $c\cdot t$. Subsection \ref{section:our_A_expander_in_F} shows that if we reach round $T$, then with high probability, $A_T$ is a near $\Omega(\phi)$-expander in $G$. Finally, in Subsection \ref{section:our_theorem_proof} we prove Theorem \ref{theorem:our}.

    \subsection{The Algorithm}\label{section:our-algorithm}
    Similarly  to Section \ref{section:osvv-short}, let $\delta = \Theta(\log{n})$ be a power of $2$, let $T=\Theta(\log^2 n)$ and $c=\Theta(\frac{1}{\phi\log n})$. We choose $c$ to be an integer.
    The algorithm for Theorem \ref{theorem:our} is presented in Algorithm \ref{algo:cut_matching}. It follows along the same lines as the algorithm of SW in Section \ref{section:SW-overview}. 
    The algorithm runs for at most $T$ rounds and stops when $\vol(R_t)>\frac{m\cdot c\cdot \phi}{70}=\Omega(\frac{m}{\log n})$. 
    In each round $t$, we update (the sparse representation of) $\F{t}$ (see Sections \ref{section:our_update_F}, \ref{section:our-matching-player} and Algorithm \ref{algo:our-update}). 

    Like SW, in order to keep the running time near linear, we use the
    flow routine \emph{\unitFlow} \cite{henzinger2017flow, saranurak2019expander}
    which is mentioned in Subsection \ref{section:our-matching-player}. This routine may also return a cut $S_t \subseteq A_t$ with $\Phi_{G[A]}(S_t, A_t\setminus S_t) \le \frac{7}{c}$, in which case we ``move'' $S_t$ to $R_{t+1}$. After $T$ rounds, $\F{T}$ certifies that the remaining part of $A_T$ is a near $\phi$-expander.
\begin{algorithm}[hbt!]
    \caption{Cut-Matching}
    \label{algo:cut_matching}
    \begin{algorithmic}[1]
        \Function{Cut-Matching}{$G, \phi$}
            \State $T\gets\Theta(\log^2 n)$, $c\gets\Theta(\frac{1}{\phi\log n})$. \Comment{$c$ is an integer.}
            \State $t\gets 0$.
            \State $A_0\gets V$,\; $R_0 \gets \emptyset$. \;
            \State $\F{0}\gets D \defeq \diag(\dG)$.
            \While{$\vol(R_t)\le \frac{m\cdot c\cdot \phi}{70}$ and $t < T$} 
                \State $S_t, \F{t+1} \gets$ \Call{Update-F}{$G, \phi, c, \F{t}, \A{t}$} \Comment{See Algorithm \ref{algo:our-update} and Sections \ref{section:our_update_F}, \ref{section:our-matching-player}.}
                \\
                \Comment{It holds that $S_t\subseteq A_t$ where $\Phi_{G[A_t]}(S_t, A_t\setminus S_t)\le \frac{7}{c}$, or $S_t=\emptyset$.}
                \State $A_{t+1} \gets A_t \setminus S_t$, $R_{t+1} \gets R_t \cup S_t$.
                \State $t\gets t+1$.
            \EndWhile
            \If {$t = T$}
                \If {$R = \emptyset$}
                    \State Certify that $\Phi(G)=\Omega(\phi)$. \Comment{Case (1) of Theorem \ref{theorem:our}.}
                \Else 
                    ~\Return {$(A_T, R_T)$}. \Comment{Case (3) of Theorem \ref{theorem:our}.}
                \EndIf
            \Else 
                ~\Return {$(A_T, R_T)$}. \Comment{Case (2) of Theorem \ref{theorem:our}.}
            \EndIf
        \EndFunction
    \end{algorithmic}
\end{algorithm}

    \begin{algorithm}[H]
        \caption{Round Update. See Sections \ref{section:our_update_F} and \ref{section:our-matching-player} for details.}
        \label{algo:our-update}
        \begin{algorithmic}[1]
            \Function{Update-F}{$G, \phi, c, \F{t}, \A{t}$}
                \State Define $\dG{t}, \D{t}, \Pmat{t}$ as in Definition \ref{def:blocked_variables}.
                \State $\delta \gets \Theta(\log n)$. \Comment{$\delta$ is a power of $2$ (see the proof of Lemma \ref{lemma:our_F_expander}).}
                \smallskip
                \State \underline{\textbf{Cut Player:}}
		\smallskip
                \State $\W{t} \gets (\Pmat{t}\normalized{\F{t}}\Pmat{t})^{\delta}$
                \State $r\gets$ Random unit vector of $\RR^{n}$.
                \State $u_i \gets \frac{1}{\sqrt{\dG{t}{i}}}\langle \W{t}{i}, r \rangle$ for $i \in \A{t}$. 

                \State Compute $\A{t}^{l}, \A{t}^{r}$ using Lemma \ref{lemma:RST}.
                \smallskip
                \State \underline{\textbf{Matching Player:}}
                \smallskip
                \State Use Lemma \ref{lemma:unit_flow} to get a cut $(S_t, V\setminus S_t)$ and a feasible flow $f$ for $\Pi(G-S_t)$.
                \State If the lemma terminated with the first case, $S_t \gets \emptyset$.
                \State Decompose $f$ into a set of paths $\{u_j \to v_j\}_j$, where $u_j \in \A{t}^l\setminus S_t$ and $v_j \in \A{t}^r\setminus S_t$.
                \State $\tilde{\M{t}} \gets \{\{u_j, v_j\}\}_{j=1}^{|\A{t}^l\setminus S_t|}$. 

                \\\Comment{$\tilde{\M{t}}$ is a symmetric $n\times n$ matrix, $\tilde{\M{t}}(u,v)$ is the number of paths between $u$ and $v$.}
                \State $\M{t} \gets \tilde{\M{t}} + \diag(\dG - \tilde{\M{t}}\1_n)$.
                \smallskip
                \State \underline{\textbf{Update of $F$:}}
                \smallskip
                \State $\N{t} \gets \frac{\delta - 1}{\delta}D + \frac{1}{\delta}\M{t}$.
                \State $\F{t+1} \gets \N{t} D^{-1} \F{t} D^{-1} \N{t}$.
                \State \Return{$S_t, \F{t+1}$}.
            \EndFunction
        \end{algorithmic}
    \end{algorithm}

\subsection{$\F{t}$ is embeddable in $G$}
\label{section:our_F_embeddable_G}

To begin the analysis of the algorithm, we first define a blocked matrix. This notion will be useful when our matrices ``operate'' only on vertices of $A_t$.
    \begin{definition}
    \label{def:d_t_block_stochastic}
        Let $A \subseteq V$. A matrix $B\in \RR^{n\times n}$ is $A$-\emph{blocked} if $\MatrixSimple{B}{i, j} = 0$ for all $i\neq j$ such that $(i,j)\notin A \times A$.
    \end{definition}
\begin{lemma}  \label{lemma:our_basic_properties}
        The following holds for all $t$:
        \begin{enumerate}
            \item $\M{t}, \N{t}, \F{t}$ and $\W{t}$ are symmetric.
            \item $\M{t}, \N{t}$ and $\F{t}$ are $\dG$-stochastic.
            \item $\M{t}$ and $\N{t}$ are $\A{t+1}$-blocked. 
        \end{enumerate}
    \end{lemma}
\begin{proof}
        \begin{enumerate}
            \item This is clear from the definitions.
            \item For $\M{t}$ this is true because we explicitly update $\tilde{\M{t}}$ to be $\dG$-stochastic. 
            (See the definition of these matrices in the description of the matching player in Section \ref{section:our-matching-player}.)
            For $\N{t}$ and $\F{t}$ the proof is similar to the proof of Lemma \ref{lemma:osvv_basic_properties}(2).

            \item Since $\tilde{\M{t}}$ is a matching of vertices in $A_t^l\setminus S_t \subseteq \A{t+1}$ to vertices in $A_t^r\setminus S_t\subseteq \A{t+1}$, then $\MatrixSimple{\tilde{\M{t}}}{i, j} = 0$ for all $i\neq j$ such that $(i,j) \notin \A{t+1} \times \A{t+1}$. After the addition of $\diag(d - \tilde{\M{t}}\1_n)$, we still have $\M{t}{i, j} = 0$ for all $i\neq j$ such that $(i,j) \notin \A{t+1} \times \A{t+1}$, so $\M{t}$ is $\A{t+1}$-blocked. For $\N{t}$, note that for all $i\neq j$ such that $(i,j) \notin \A{t+1} \times \A{t+1}$, $\N{t}{i, j} = \frac{\delta - 1}{\delta}\MatrixSimple{D}{i, j} + \frac{1}{\delta}\M{t}{i, j} = \frac{\delta - 1}{\delta}\cdot 0 + \frac{1}{\delta}\cdot 0 = 0$.
        \end{enumerate}
    \end{proof}

    The following lemma will allow us to show that for every $t$, $F_t$ can be embedded into $G_{t}$. 

\begin{lemma}
    \label{lemma:our_F_t_embeddable_in_G_t_step1}

    Let $F\in \RR^{n\times n}$ be a $\dG$-stochastic matrix, and let $H$ be a graph. Fix any $t$. Assume that $F$ is embeddable in $H$ such that the congestion on each edge $e\in H$ is $c(e)$, then
    \begin{enumerate}
        \item $\N{t}D^{-1} F$ is embeddable in $H\cup\M{t}$ such that the congestion on edges of $\M{t}$ is $\frac{2}{\delta}$ and the congestion on each edge $e\in H$ is still $c(e)$.
        \item $F D^{-1} \N{t}$ is embeddable in $H\cup\M{t}$ such that the congestion on edges of $\M{t}$ is $\frac{2}{\delta}$ and the congestion on each edge $e\in H$ is still $c(e)$.
        \item $\N{t} D^{-1} \F D^{-1} \N{t}$ is embeddable in $H \cup \M{t}$ such that the congestion on edges of $\M{t}$ is $\frac{4}{\delta}$ and the congestion on each edge $e\in H$ is still $c(e)$.

    \end{enumerate}
    \end{lemma}
\begin{proof}
        \begin{enumerate}
        \item Let $P:E\times V \to \RR_{\ge 0}$ be a routing of $\F$ in $H$ (where $P((u,v), w)$ indicates how much of $w$'s commodity goes through the edge $(u,v)\in E$ from $u$ to $v$). Note that $\N{t} D^{-1} F$ is the flow matrix obtained by performing a weighted average on the rows of $F$ described by $\M{t}$. Explicitly, for every $v\in V, a \in V$, we have 
        \[(\N{t}D^{-1} F)(v,a) = \Matrix{\frac{\delta-1}{\delta} F + \frac{1}{\delta} \M{t}D^{-1} F}{v,a} = \frac{\delta-1}{\delta}F(v,a) + \frac{1}{\delta}\sum_{u : \{u,v\}\in \M{t}}{\frac{1}{\dG{}{u}}F(u,a)}.\footnote{Note that $u$ may appear in the sum multiple times.}\]

        The precise construction of the new embedding is as follows. 

        \begin{algorithm}[H]
            \begin{algorithmic}[]
                \State $P' \leftarrow \frac{\delta-1}{\delta}\cdot P$.
                \For{$\{a,b\}\in \M{t}$} \Comment{Note that $\{a,b\}$ may appear in this loop multiple times.}
                    \State $P'((a,b),a) \leftarrow P'((a, b), a) + \frac{1}{\delta}$.
                    \State $P'((b,a),b) \leftarrow P'((b, a), b) + \frac{1}{\delta}$.
                \EndFor
                \For{$\{a,b\}\in \M{t}$}
                    \For{$(u,v)\in H$} \Comment{The edge $\{u,v\}\in H$ is scanned in both directions.}
                        \State $P'((u,v),a) \leftarrow P'((u,v),a) + \frac{1}{\delta\dG{}{b}}P((u,v),b)$.
                        \State $P'((u,v),b) \leftarrow P'((u,v),b) + \frac{1}{\delta\dG{}{a}}P((u,v),a)$.
                    \EndFor
                \EndFor
            \end{algorithmic}
        \end{algorithm}

    The argument that we indeed obtain an embedding of $N_tD^{-1} F=\frac{\delta-1}{\delta}F+\frac{1}{\delta} M_t D^{-1} F$ is as follows.
    We think of $P'$ in  stages. In the first stage, we scale $P$ by $\frac{\delta-1}{\delta}$. This routes $\frac{\delta-1}{\delta}\MatrixSimple{F}{v,a}$ units of flow (of $v$'s commodity) from $v$ to $a$ for every $v$ and $a$. After this stage, each vertex $v\in V$ currently sends $\frac{(\delta-1)}{\delta}\dG{}{v}$ units of its commodity. In the next stage we wish to route an additional $\frac{1}{\delta}\sum_{u : \{u,v\}\in \M{t}}{\frac{1}{\dG{}{u}}\MatrixSimple{F}{u,a}}$ units from $v$ to $a$. To this end, we first move $\frac{1}{\delta}$ units from $v$'s commodity to each $u$ with $\{u,v\}\in \M{t}$\footnote{Note that the same pair $\{u,v\}$ might appear multiple times in the matching $\M{t}$ and therefore $v$ will send $\frac{1}{\delta}$ units of flow to $u$ multiple times.} (this flow is sent through the edge $\{v,u\}\in H\cup \M{t}$ from $v$ to $u$). Now, each vertex $u\in V$ ``mixes'' the commodities it got from its neighbors and routes the $\frac{\dG{}{u}}{\delta}$ ``new'' units according to $P$, as if they were of its own commodity. 
    Thus, $\frac{\MatrixSimple{F}{u,a}}{\delta}$ is the total flow sent from $u$ to $a$ when it routes the $\frac{\dG{}{u}}{\delta}$ ``new'' units and the ``share'' of $v$'s commodity from this flow is $\frac{1}{\dG{}{u}}$.
    It follows that out of the $\frac{1}{\delta}$ units $v$ sent to $u$, exactly $\frac{1}{\dG{}{u}}\frac{\MatrixSimple{F}{u,a}}{\delta}$ units go to $a$. 

    As for the congestion, note that on each edge of $\M{t}$, we route $\frac{1}{\delta}$ units of flow in each direction, so the congestion on each such edge is $\frac{2}{\delta}$.
    For each (directed) edge $(u, v)\in H$ and $a\in V$, we have $P'((u,v),a) = \frac{\delta - 1}{\delta}\cdot P((u,v),a) + \sum_{b : \{a,b\}\in \M{t}}{\frac{1}{\delta\dG{}{b}}P((u,v),b)}$. Note that some vertices $b\in V$ may appear in the sum multiple times.
    Therefore, on each edge $(u,v)\in H$, the congestion is 
    \begin{align*}
        \sum_{a\in V}{P'((u,v),a)} &= \frac{\delta - 1}{\delta}\cdot \sum_{a\in V}{P((u,v),a)} + \sum_{a\in V}{\sum_{b : \{a,b\}\in \M{t}}{\frac{1}{\delta\dG{}{b}}P((u,v),b)}}
        \\
        &= \frac{\delta - 1}{\delta}\cdot \sum_{a\in V}{P((u,v),a)} + \frac{1}{\delta}\sum_{b\in V}{P((u,v),b)} 
        \\
        &= \sum_{a\in V}{P((u,v),a)} = c(u,v) \ .
    \end{align*}

    \item  Let $P$ be the routing of $F$ in $H$. Note that $F D^{-1}\N{t}= \frac{\delta-1}{\delta}F + \frac{1}{\delta}FD^{-1}M_t$
    is the flow matrix obtained by performing a weighted average on the columns of $F$ described by $\M{t}$, \ie, we average the flow received by matched vertices.

    We define a routing $P'$ of $F D^{-1}\N{t}$ in $H\cup \M{t}$ as follows:
    \begin{algorithm}[H]
        \begin{algorithmic}[]
            \State $P' \leftarrow P$.
            \For{$\{a,b\}\in \M{t}$}
                \For{$u\in V$}
                    \State $P'((a,b),u) \leftarrow \frac{1}{\delta \dG{}{a}}\cdot \MatrixSimple{F}{u,a}$.
                    \State $P'((b,a),u) \leftarrow \frac{1}{\delta \dG{}{b}}\cdot \MatrixSimple{F}{u,b}$.
                \EndFor
            \EndFor
        \end{algorithmic}
    \end{algorithm}

    That is, $P'$ routes the same as $P$ over edges in $H$. For every $\{a,b\}\in \M{t}$, $a$ sends $\frac{1}{\delta}$ units of its received flow, which is a mix of commodities, to $b$ and in turn, $b$ receives $\frac{1}{\delta}$ units of flow from $a$.
    Note that for every $a, u \in V$, $\Matrix{F D^{-1}\N{t}}{a, u} = \frac{\delta - 1}{\delta}\MatrixSimple{F}{a, u} + \frac{1}{\delta}\sum_{v : \{u,v\}\in \M{t}}{\frac{1}{\dG{}{v}}\MatrixSimple{F}{a, v}}$. Thus, $P'$ routes $F D^{-1}\N{t}$ in~$H$. 

    Note that the congestion of $P'$ on each edge $e\in H$ is still $c(e)$, and on each edge $(a,b)$ such that $\{a,b\}\in \M{t}$ we have $\sum_{u\in V}{P'((a,b),u)} = \sum_{u\in V}{\frac{1}{\delta \dG{}{a}}\MatrixSimple{F}{u,a}} = \frac{1}{\delta}$ (the last equality follows as $F$ is $\dG$-stochastic). Therefore $\frac{1}{\delta}$ flow was routed in each direction, so the congestion on the edge $\{a, b\}\in \M{t}$ is $\frac{2}{\delta}$.

    \item Note that because $F$ and $\N{t}$ are $\dG$-stochastic (see Lemma \ref{lemma:our_basic_properties}(2)), $N_{t} D^{-1} F$ is also $\dG$-stochastic, since:  
    \begin{align*}
        N_{t} D^{-1} F \cdot \1_n &= N_{t} D^{-1} \dG = N_{t} \cdot \1_n = \dG
        \\
        \1_n' \cdot N_{t} D^{-1} F &= \dG' \cdot D^{-1} F = \1_n' \cdot F = \dG' \ .
    \end{align*}
    Therefore we can apply Parts (1) and (2) to get the result.
    \end{enumerate}

    \end{proof}

As an immediate corollary of Lemma~\ref{lemma:our_F_t_embeddable_in_G_t_step1}, we have: 

\begin{corollary}
\label{lemma:our_F_t_embeddable_in_G_t}
    For all rounds $t$, $\F{t}$ is embeddable in $G_t$ with congestion $\frac{4}{\delta}$.
\end{corollary}

\begin{lemma}
\label{lemma:our_G_t_embeddable_in_G}
    For all rounds $t$, $G_t$ is embeddable in $G$ with congestion $ct$.
\end{lemma}
\begin{proof}
    For every $t$, by the definition of the flow problem at round $t$, $\M{t}$ is embeddable in $G$ with congestion $c$. Summing these routings gives a routing of $G_t = \bigcup_{i=1}^t{\M{i}}$ in $G$ with congestion $ct$.
\end{proof}

    \subsection{$A_T$ is a near expander in $\F{T}$}
    \label{section:our_A_expander_in_F}
    In this section we prove that after $T=\Theta(\log^2 n)$ rounds, with high probability, $A_T$ is a near $\Omega(1)$-expander in $\F{T}$, which will imply that it is a near $\Omega(\phi)$-expander in $G$. 
    
    The section is organized as follows. Lemma \ref{lemma:submatrix_properties} contains matrix identities and Lemma \ref{lemma:our_X_as_normalized_laplacian} specifies a spectral property that our proof requires.
    We then define a potential function and lower bound the decrease in potential in Lemmas~\ref{lemma:our_potential_step_1}-\ref{cor:our_total_potential}. Finally, in Lemma \ref{lemma:our_F_expander} and Corollary \ref{cor:our_G_expander} we use the upper bound on the potential at round $T$, to show that with high probability $A_T$ is a near $\Omega(1)$-expander in $\F{T}$ and a near $\Omega(\phi)$-expander in $G$. 
\begin{lemma}
    \label{lemma:submatrix_properties}
    The following relations hold for all $t$:
    \begin{enumerate}
        \item For any $\A{t}$-blocked $\dG$-stochastic matrix $B\in \RR^{n\times n}$ we have $\I{t}\normalized{B} = \normalized{B}\I{t}$ and $\Pmat{t} \cdot  \Dminushalf B \Dminushalf =  \Dminushalf B \Dminushalf \cdot \Pmat{t}$.
        \item $\I{t}\Pmat{t} = \Pmat{t}$, $\I{t}^2 = \I{t}$ and $\Pmat{t}^2 = \Pmat{t}$.
        \item $\Pmat{t} \Pmat{{t+1}} = \Pmat{{t+1}} \Pmat{t} = \Pmat{{t+1}}$.
        \item $\Pmat{t} = \Dminushalf \mathcal{L}(\frac{1}{\vol_t} \dG{t} \dG{t}') \Dminushalf$ (recall the Laplacian defined in Definition \ref{def:laplacian_normalized_laplacian}).
        \item for any $v\in \RR^n$, it holds that 
        $v' \mathcal{L}\left(\frac{1}{\vol_t} \dG{t} \dG{t}'\right) v = \norm{\D{t}^{\frac{1}{2}} v}_2^2 - \frac{1}{\vol_t}  \left\langle v, \dG{t}\right\rangle^2$. 
        \item For any $B\in \RR^{n \times n}$, $\tr(\I{t} BB') = \sum_{i\in A_t}\norm{B(i)}_2^2$.
    \end{enumerate}
    \end{lemma}
\begin{proof}
    \begin{enumerate}
        \item Since $B$ is a $\A{t}$-blocked $\dG{t}$-stochastic matrix, $\I{t}\cdot \normalized{B} = \normalized{B}\cdot \I{t}$ is clear. Moreover 
        \begin{align*}
            \Dminushalf B \Dminushalf \sqrt{\dG{t}}\sqrt{\dG{t}'} = 
            \Dminushalf B \1_t \sqrt{\dG{t}'} = \Dminushalf \dG{t} \sqrt{\dG{t}'} = \sqrt{\dG{t}}\sqrt{\dG{t}'}, \\ 
            \sqrt{\dG{t}}\sqrt{\dG{t}'} \Dminushalf B \Dminushalf  = 
            \sqrt{\dG{t}} \1_t'  B  \Dminushalf = \sqrt{\dG{t}}  \dG{t}' \Dminushalf = \sqrt{\dG{t}}\sqrt{\dG{t}'} \ .
        \end{align*}

        Thus,
        \begin{align*}
            \Pmat{t} \cdot \Dminushalf B \Dminushalf &= (\I{t}-\frac{1}{\vol_t}\sqrt{\dG{t}}\sqrt{\dG{t}'})\Dminushalf B \Dminushalf \\
            &= \Dminushalf B \Dminushalf \I{t} - \frac{1}{\vol_t}\sqrt{\dG{t}}\sqrt{\dG{t}'} = \Dminushalf B \Dminushalf \cdot \Pmat{t} \ .
        \end{align*}

        \item The first and second equalities are clear from the definitions. For $\Pmat{t}$,
        \begin{align*}
            \Pmat{t}^2 &= \left(\I{t} - \frac{1}{\vol_t}\sqrt{\dG{t}}\sqrt{\dG{t}'}\right)\cdot\left(\I{t} - \frac{1}{\vol_t}\sqrt{\dG{t}}\sqrt{\dG{t}'}\right)
            \\
            &= \I{t} - \frac{2}{\vol_t}\sqrt{\dG{t}}\sqrt{\dG{t}'} + \frac{1}{\vol_t^2}\sqrt{\dG{t}}\sqrt{\dG{t}'}\sqrt{\dG{t}}\sqrt{\dG{t}'} = \I{t} - \frac{1}{\vol_t}\sqrt{\dG{t}}\sqrt{\dG{t}'} = \Pmat{t}
        \end{align*}
        \item We show $\Pmat{t} \Pmat{{t+1}} = \Pmat{{t+1}}$. The other direction follows by symmetry.
        \begin{align*}
            \Pmat{t} \Pmat{{t+1}} &= \left(\I{t} - \frac{1}{\vol_t}\sqrt{\dG{t}}\sqrt{\dG{t}'}\right)\cdot \left(\I{{t+1}} - \frac{1}{\vol_{t+1}}\sqrt{\dG{t+1}}\sqrt{\dG{t+1}'}\right) \\&=
            \I{{t+1}} - \frac{1}{\vol_{t+1}}\sqrt{\dG{t+1}}\sqrt{\dG{t+1}'} -
            \frac{1}{\vol_t}\sqrt{\dG{t}}\sqrt{\dG{t}'}\I{{t+1}} +
            \frac{1}{\vol_t \cdot \vol_{t+1}}\sqrt{\dG{t}}\sqrt{\dG{t}'}\sqrt{\dG{t+1}}\sqrt{\dG{t+1}'} \\&=
            \I{{t+1}} - \frac{1}{\vol_{t+1}}\sqrt{\dG{t+1}}\sqrt{\dG{t+1}'} -
            \frac{1}{\vol_t}\sqrt{\dG{t}}\sqrt{\dG{t+1}'} +
            \frac{1}{\vol_t }\sqrt{\dG{t}}\sqrt{\dG{t+1}'} \\&=
            \I{{t+1}} - \frac{1}{\vol_{t+1}}\sqrt{\dG{t+1}}\sqrt{\dG{t+1}'} = \Pmat{{t+1}}.
        \end{align*}

        \item Since $\langle \dG{t},\1 \rangle = \vol_t$, we get that the degree matrix of $\frac{1}{\vol_t}\dG{t} \dG{t}'$ is $\D{t}$. Hence
        \[
        \mathcal{L} \left( \frac{1}{\vol_t}\dG{t} \dG{t}'\right) = \D{t} - \frac{1}{\vol_t}\dG{t} \dG{t}' = 
        D^{\frac{1}{2}} (\I{t} - \frac{1}{\vol_t}\sqrt{\dG{t}}\sqrt{\dG{t}'}) D^{\frac{1}{2}} =
        D^{\frac{1}{2}} \Pmat{t} D^{\frac{1}{2}}.
        \]
        \item
        \begin{align*}
            v' \mathcal{L}\left(\frac{1}{\vol_t} \dG{t} \dG{t}'\right) v &= 
            v'  (\D{t} - \frac{1}{\vol_t}\dG{t} \dG{t}')  v =
             v' \D{t} v  - \frac{1}{\vol_t} v' \dG{t} \dG{t}' v \\&=
             (\D{t}^{\frac{1}{2}} v)' (\D{t}^{\frac{1}{2}} v) - \frac{1}{\vol_t} (v' \dG{t}) (\dG{t}' v) =
             \norm{\D{t}^{\frac{1}{2}} v}_2^2 - \frac{1}{\vol_t} \langle v,\dG{t} \rangle^2.
        \end{align*}
        \item Observe that $\sum_{i\in A_t}\norm{B(i)}_2^2 = \norm{\I{t} B}_F^2$. Indeed, $X= \I{t} B$ satisfies $\X{}{i,j} = B(i,j)$ if $i\in A_t$ (and otherwise $\X{}{i,j} = 0$). Therefore,
        \begin{align*}
            \sum_{i\in A_t}\norm{B(i)}_2^2 &= \norm{\I{t} B}_F^2 = \tr((\I{t} B)  (\I{t} B)') = 
            \tr (\I{t} B  B' \I{t}) \\&=
            \tr(\I{t}^2 B  B') = \tr(\I{t} B B').
        \end{align*}
        Where the fourth equality follows from Fact \ref{fact:book_trace_identities} and the last equality follows from (2).
    \end{enumerate}
\end{proof}

We define the potential $\psi(t) = \tr[\W{t}^2]$, where $\W{t}$ was defined as $\W{t} = (\Pmat{t} \Dminushalf \F{t} \Dminushalf \Pmat{t})^\delta$. 
By Lemma~\ref{lemma:submatrix_properties}(6) (and since $\I{t}\W(t) = \W{t}$), the potential can be written as $\psi(t) = \sum_{i\in A_t}{\norm{\W{t}{i}}_2^2}$.
This is the same potential from Section \ref{section:osvv-short} with the new definition of $\W{t}$. Intuitively, by projecting using $\Pmat{t}$ instead of $\Pmat$, the potential only ``cares'' about the vertices of $A_t$. As shown in Lemma~\ref{lemma:our_F_expander}, having small potential will certify that $A_T$ is a near expander in $\F{t}$. 

Before we bound the decrease in potential, we recall Definition~\ref{def:laplacian_normalized_laplacian} of a normalized Laplacian $\mathcal{N}(A) = \normalized{\mathcal{L}(A)} = I - \normalized{A}$, where $A$ is a symmetric $\dG$-stochastic matrix. 

The proof of the following technical lemma is identical to Lemma \ref{lemma:osvv_X_as_normalized_laplacian} in our analysis of OSVV's cut matching game for conductance  in Appendix~\ref{section:osvv}.

\begin{lemma}
\label{lemma:our_X_as_normalized_laplacian}
    For any matrix $A\in \RR^{n\times n}$, $\tr(A'(I-(\normalized{\N{t}})^{4\delta})A)\ge \frac{1}{3}\tr(A' \mathcal{N}(\M{t})A)$.
\end{lemma}

The following lemma bounds the decrease in potential. The bound takes into account both the contribution of the matched vertices and the removal of $S_t$ from $A_t$.
\begin{lemma}
    \label{lemma:our_potential_step_1}
        For each round $t$, \[\psi(t) - \psi(t+1) \ge \frac{1}{3}\sum_{\{i,k\}\in \M{t}} \norm{\left(\frac{\W{t}{i}}{\sqrt{\dG{}{i}}} - \frac{\W{t}{k}}{\sqrt{\dG{}{k}}}\right)}_2^2 + \sum_{j\in S_{t}} \dG{}{j} \norm{\frac{\W{t}{j}}{\sqrt{\dG{}{j}}}}_2^2\]
    \end{lemma}
    \begin{proof}
        To simplify the notation, we denote $\bar{N}_t \defeq \normalized{\N{t}}$ and $\bar{F}_t \defeq \normalized{\F{t}}$.
        We rewrite the potential in the next iteration as follows:
\begin{align*}
            \psi(t+1) &= \tr(\W{{t+1}}^2) = \tr\left( \left( \Pmat{{t+1}} \Dminushalf \F{{t+1}} \Dminushalf \Pmat{{t+1}} \right)^{2\delta} \right) 
            \\
            &= \tr\left( \left( \Pmat{{t+1}} \Dminushalf (\N{{t}} D^{-1} \F{{t}} D^{-1} \N{{t}}) \Dminushalf \Pmat{{t+1}} \right)^{2\delta}\right) 
            \\
            &= \tr\left( \left( \Pmat{{t+1}} \Dminushalf (\N{{t}} \Dminushalf \Dminushalf \F{{t}} \Dminushalf \Dminushalf \N{{t}}) \Dminushalf \Pmat{{t+1}} \right)^{2\delta}\right) 
            \\
            &= \tr\left( \left( \Pmat{{t+1}} \bar{N}_{t} \bar{F}_t \bar{N}_{t}  \Pmat{{t+1}} \right)^{2\delta}\right) 
            \\
            &\numeq{6} \tr\left( \left(  \bar{N}_{t} \Pmat{{t+1}} \bar{F}_t  \Pmat{{t+1}} \bar{N}_{t}  \right)^{2\delta}\right) 
            \\
            &\numeq{7} \tr\left( \left(  \bar{N}_{t} \Pmat{{t+1}} \Pmat{t} \bar{F}_t \Pmat{t} \Pmat{{t+1}} \bar{N}_{t}   \right)^{2\delta}\right) 
            \\
            &= \tr\left( \left(  \bar{N}_{t} \Pmat{{t+1}} (\Pmat{t} \bar{F}_t \Pmat{t}) \Pmat{{t+1}} \bar{N}_{t}   \right)^{2\delta}\right) \ ,
        \end{align*}
        where equality $(6)$ follows from Lemma \ref{lemma:submatrix_properties}(1) for $\N{{t}}$ (which is $\A{t+1}$-blocked $\dG$-stochastic by Lemma \ref{lemma:our_basic_properties}), and equality $(7)$ follows from Lemma \ref{lemma:submatrix_properties}(3). 

        By Properties (1) and (2) of Lemma~\ref{lemma:submatrix_properties} it holds that $\bar{N}_{t+1} \Pmat{{t+1}} = \Pmat{{t+1}} \bar{N}_{t+1}  = \Pmat{{t+1}} \bar{N}_{t+1} \Pmat{{t+1}}$. 
        Therefore, the potential can be written in terms of symmetric matrices:
        \begin{align*}
            \psi(t+1) &= \tr\left( \left(  (\Pmat{{t+1}} \bar{N}_{t} \Pmat{{t+1}}) (\Pmat{t} \bar{F}_t \Pmat{t}) (\Pmat{{t+1}} \bar{N}_{t} \Pmat{{t+1}})   \right)^{2\delta}\right) 
            \\
            &\le \tr ((\Pmat{{t+1}} \bar{N}_{t} \Pmat{{t+1}})^{2\delta} (\Pmat{t} \bar{F}_t \Pmat{t})^{2\delta} (\Pmat{{t+1}} \bar{N}_{t} \Pmat{{t+1}})^{2\delta}) 
            \\
            &\numeq{2} \tr ( (\Pmat{{t+1}} \bar{N}_{t} \Pmat{{t+1}})^{4\delta} (\Pmat{t} \bar{F}_t \Pmat{t})^{2\delta}) = \tr((\bar{N}_{t} \Pmat{{t+1}})^{4\delta} \W{t}^2) 
            \\
            &\numeq{4} \tr(\bar{N}_{t}^{4\delta}\Pmat{{t+1}} \W{t}^2) \numeq{5} \tr(\bar{N}_{t}^{2\delta}\Pmat{{t+1}}\bar{N}_{t}^{2\delta} \W{t}^2) 
            \\
            &\numeq{6} \tr(\W{t} \bar{N}_{t}^{2\delta}\Pmat{{t+1}}\bar{N}_{t}^{2\delta} \W{t})
            \\
            &\numeq{7} \tr(\W{t} \bar{N}_{t}^{2\delta} \Dminushalf \mathcal{L}\left(\frac{1}{\vol_{t+1}} \dG{t+1} \dG{t+1}'\right) \Dminushalf \bar{N}_{t}^{2\delta} \W{t}) 
            \\
            &= \tr\left(\left(\Dminushalf \bar{N}_{t}^{2\delta} \W{t} \right)' \cdot \mathcal{L}\left(\frac{1}{\vol_{t+1}} \dG{t+1} \dG{t+1}'\right) \cdot \left(\Dminushalf \bar{N}_{t}^{2\delta} \W{t}\right)\right)\ ,
        \end{align*}
        where the inequality follows from Theorem \ref{theorem:symmetric_rearrangement} and Equality $(2)$ follows from Fact \ref{fact:book_trace_identities}. Equalities~$(4)$  and $(5)$ follow from Properties (1) and (2) of Lemma \ref{lemma:submatrix_properties} (and from the fact that $\N{{t}}$ is $\A{t+1}$-blocked $\dG$-stochastic, by Lemma \ref{lemma:our_basic_properties}). 
        Equality $(6)$ again uses Fact \ref{fact:book_trace_identities}, and Equality $(7)$ follows from Lemma \ref{lemma:submatrix_properties}(4). 

        Let $\Z{t} = \Dminushalf \bar{N}_{t}^{2\delta} \W{t}$.  By applying Lemma~\ref{lemma:submatrix_properties}(5) we get
        \begin{align}
            \psi(t+1) &\le \tr\left(\Z{t}' \mathcal{L}\left(\frac{1}{\vol_{t+1}} \dG{t+1} \dG{t+1}'\right) \Z{t}\right) = \sum_{i=1}^n (\Z{t}{,i})' \mathcal{L}\left(\frac{1}{\vol_{t+1}} \dG{t+1} \dG{t+1}'\right) \Z{t} (,i) 
            \notag \\
            &\numeq{2} \sum_{i=1}^n \left(\norm{\D{{t+1}}^{\frac{1}{2}} \Z{t}{,i}}_2^2 - \frac{1}{\vol_{t+1}} \left\langle \Z{t}{,i}, \dG{t+1}\right\rangle^2\right) \le  \sum_{i=1}^n \norm{\D{{t+1}}^{\frac{1}{2}} \Z{t}{,i}}_2^2 
            \notag \\
            &= \sum_{i=1}^n \sum_{j\in A_{t+1}} \left(\sqrt{\dG{}{j}} \Z{t}{j,i}\right)^2 = \sum_{j\in A_{t+1}} \norm{\Matrix{\D{{t+1}}^{\frac{1}{2}}\Z{t}}{j}}_2^2 
            \numeq{5} \sum_{j\in A_{t+1}} \norm{\Matrix{\bar{N}_{t}^{2\delta}\W{t}}{j}}_2^2 
            \notag \\
            &= \sum_{j\in A_{t}} \norm{\Matrix{\bar{N}_{t}^{2\delta}\W{t}}{j}}_2^2 - \sum_{j\in S_{t}} \norm{\Matrix{\bar{N}_{t}^{2\delta}\W{t}}{j}}_2^2, \label{eq:psit+1ep}
        \end{align}
        where Equality $(2)$ holds by Property (5) of Lemma~\ref{lemma:submatrix_properties} and Equality $(5)$ holds since we only sum rows in $A_{t+1}$.
        Since $\bar{N}_{t}$ is diagonal outside $A_{t+1}$ (by the definition of $\M{t}$),
        we have that $\Matrix{\bar{N}_{t}^{2\delta} \W{t}}{j} = \W{t}{j}$, for every $j\in S_t$.
        Thus,
        \begin{equation}\label{eq:n(N)_W_correlation}
            \sum_{j\in S_{t}} \norm{\Matrix{\bar{N}_{t}^{2\delta}\W{t}}{j}}_2^2 =
        \sum_{j\in S_{t}} \norm{\W{t}{j}}_2^2. 
        \end{equation}

        By Lemma~\ref{lemma:submatrix_properties}(6), we get
        \begin{align}
            \sum_{j\in A_{t}} \norm{\Matrix{\bar{N}_{t}^{2\delta}\W{t}}{j}}_2^2 &= 
            \tr(\I{{t}} \cdot \bar{N}_{t}^{2\delta} \cdot \W{t}^2 \cdot \bar{N}_{t}^{2\delta}) \notag \\&\numeq{2}
            \tr(\bar{N}_{t}^{2\delta} \cdot \I{{t}} \cdot  \W{t}^2 \cdot \bar{N}_{t}^{2\delta}) \notag \\&\numeq{3}
            \tr(\bar{N}_{t}^{2\delta} \cdot  \W{t}^2 \cdot \bar{N}_{t}^{2\delta}) \notag \\&\numeq{4}
            \tr(   \bar{N}_{t}^{4\delta}\W{t}^2 ) \label{eq:At}
        \end{align}
        where equality $(2)$ holds since $\N{{t}}$ is $\A{t+1}$-blocked $\dG$-stochastic (by Lemma \ref{lemma:our_basic_properties}), so in particular it is $\A{t}$-blocked $\dG$-stochastic, and we can use Lemma \ref{lemma:submatrix_properties}(1). Equality $(3)$ holds because $\I{t} \W{t}  = \I{t} (\Pmat{t} \bar{F}_{t} \Pmat{t})^\delta$ 
        and $\I{t} \Pmat{t}  = \Pmat{t}$ (by Lemma \ref{lemma:submatrix_properties}(2)), and equality $(4)$ follows from Fact \ref{fact:book_trace_identities}. 
        Plugging 
        Equations (\ref{eq:n(N)_W_correlation}) and (\ref{eq:At}) 
        into (\ref{eq:psit+1ep})
 we get the following bound on the decrease in potential: 
        \begin{align*}
            \psi(t) - \psi(t+1) &\ge \tr(   (I-\bar{N}_{t}^{4\delta})\W{t}^2 ) + 
            \sum_{j\in S_{t}} \norm{\W{t}{j}}_2^2 \\&\numeq{1}
            \tr( \W{t}  (I-\bar{N}_{t}^{4\delta})\W{t} ) + 
            \sum_{j\in S_{t}} \norm{\W{t}{j}}_2^2 \\&\numge{2}
            \frac{1}{3}\tr( \W{t}  \mathcal{N}(\M{{t}}) \W{t} ) + 
            \sum_{j\in S_{t}} \norm{\W{t}{j}}_2^2  \\&=
            \frac{1}{3}\tr( (\Dminushalf \W{t})'  \mathcal{L}(\M{{t}}) (\Dminushalf \W{t}) ) + 
            \sum_{j\in S_{t}}  \dG{}{j} \norm{\frac{\W{t}{j}}{\sqrt{\dG{}{j}}}}_2^2 \\&\numeq{3}
            \frac{1}{3}\sum_{\{i,k\}\in \M{t}}{\norm{\frac{\W{t}{i}}{\sqrt{\dG{}{i}}} - \frac{\W{t}{k}}{\sqrt{\dG{}{k}}}}_2^2} + \sum_{j\in S_{t}} \dG{}{j} \norm{\frac{\W{t}{j}}{\sqrt{\dG{}{j}}}}_2^2
        \end{align*}
        where Equality $(1)$ follows from Fact \ref{fact:book_trace_identities}, Inequality~$(2)$ follows by Lemma~\ref{lemma:our_X_as_normalized_laplacian},  and Equality~$(3)$ follows from Lemma \ref{lemma:matching_laplacian}.

    \end{proof}

    The following lemma states that the potential is expected to drop by a factor of $1-\Omega(1/\log n)$.
\begin{lemma}
    \label{lemma:our_potential_step_2}
        For each round $t$, 
        \[\EE\left[\frac{1}{3}\sum_{\{i,k\}\in \M{t}} \norm{\frac{\W{t}{i}}{\sqrt{\dG{}{i}}} - \frac{\W{t}{k}}{\sqrt{\dG{}{k}}}}_2^2 + \sum_{j\in S_{t}} \dG{}{j} \norm{\frac{\W{t}{j}}{\sqrt{\dG{}{j}}}}_2^2\right]  
        \ge
        \frac{1}{3000 \alpha \log n}\psi(t) - \frac{3}{n^{\alpha/16}}\]
        for every $\alpha>48$, where the expectation is over the unit vector $r\in \RR^n$ of the current round and conditioned on $\psi(t)$. 

    \end{lemma}    

    \begin{proof}
        Recall that $u_i = \frac{1}{\sqrt{\dG{}{i}}} \langle \W{t}{i}, r \rangle$ for $i\in A_t$. Notice that $\Matrix{\frac{\W{t}{i}}{\sqrt{\dG{}{i}}} }{j} = \frac{\W{t}{i,j}}{\sqrt{\dG{}{i}}}$. 

        We use Lemma \ref{lemma:projection_pairs} from Appendix \ref{appendix:projection} for the set of vectors $\left\{\frac{1}{\sqrt{\dG{}{i}}}\W{t}{i} \mid i\in A_t\right\}$. 
        By Lemma \ref{lemma:projection_pairs}(2), we have with high probability:
        \begin{align*}
            \forall i, k\in A_t &: \norm{\frac{\W{t}{i}}{\sqrt{\dG{}{i}}} - \frac{\W{t}{k}}{\sqrt{\dG{}{k}}}}_2^2 \ge \frac{n}{\alpha\log n}\cdot(u_i-u_k)^2
            \\
            \forall i\in A_t &: \norm{\frac{\W{t}{i}}{\sqrt{\dG{}{i}}}}_2^2 \ge \frac{n}{\alpha\log n}\cdot u_i^2
        \end{align*}
        for every constant $\alpha \ge 16$. In order to replace the inequality with high probability by an inequality in expected values, we introduce a random variable $z$ that is non-zero only when this inequality fails to hold, such that  
        \begin{align*}
	        \forall i, k\in A_t &: \norm{\frac{\W{t}{i}}{\sqrt{\dG{}{i}}} - \frac{\W{t}{k}}{\sqrt{\dG{}{k}}}}_2^2 \ge \frac{n}{\alpha \log n}\cdot(u_i - u_k)^2 - z
	        \\
            \forall i\in A_t &: \norm{\frac{\W{t}{i}}{\sqrt{\dG{}{i}}}}_2^2 \ge \frac{n}{\alpha\log n}\cdot u_i^2 - z
        \end{align*}
        holds with probability $1$. \emph{I.e.}, we define
        \begin{align*}
            \mathcal{B} = \{0\} &\cup \left\{\frac{n}{\alpha\log n}(u_i - u_k)^2-\norm{\frac{\W{t}{i}}{\sqrt{\dG{}{i}}} - \frac{\W{t}{k}}{\sqrt{\dG{}{k}}}}_2^2 : (i, k)\in A_t\times A_t\right\}
            \\
            &\cup \left\{\frac{n}{\alpha\log n}u_i^2-\norm{\frac{\W{t}{i}}{\sqrt{\dG{}{i}}}}_2^2 : i\in A_t\right\}
        \end{align*}
        and $z = \max(\mathcal{B})$. We get that
        \begin{align*}
             \frac{1}{3}\sum_{\{i,k\}\in \M{t}} \norm{\frac{\W{t}{i}}{\sqrt{\dG{}{i}}} - \frac{\W{t}{k}}{\sqrt{\dG{}{k}}}}_2^2 &\ge
             \frac{n}{3\alpha \log{n}}\sum_{\{i,k\}\in \M{t}} (u_i - u_k)^2 - m\cdot z
             \\
             \sum_{j\in S_{t}} \dG{}{j} \norm{\frac{\W{t}{j}}{\sqrt{\dG{}{j}}}}_2^2 &\ge \frac{n}{\alpha \log{n}} \sum_{j\in S_{t}} \dG{}{j} u_j^2 - 2m\cdot z.
        \end{align*}

        This means that 
        \begin{align*}
            &\frac{1}{3}\sum_{\{i,k\}\in \M{t}} \norm{\left(\frac{\W{t}{i}}{\sqrt{\dG{}{i}}} - \frac{\W{t}{k}}{\sqrt{\dG{}{k}}}\right)}_2^2 + \sum_{j\in S_{t}} \dG{}{j} \norm{\frac{\W{t}{j}}{\sqrt{\dG{}{j}}}}_2^2
            \\
            &\ge \frac{n}{3\alpha \log{n}}\sum_{\{i,k\}\in \M{t}} (u_i - u_k)^2 + 
            \frac{n}{\alpha \log{n}} \sum_{j\in S_{t}} \dG{}{j} u_j^2  - 3m\cdot z\\&\numge{2}
            \frac{n}{3\alpha \log{n}}\sum_{i\in A_t^{l}\setminus S_t} m_i (u_i - \eta)^2 + 
            \frac{n}{\alpha \log{n}} \sum_{j\in S_{t}} \dG{}{j} u_j^2  - 3m\cdot z\\&\ge
            \frac{n}{3\alpha \log{n}}\sum_{i\in A_t^{l}\setminus S_t} m_i (u_i - \eta)^2 + 
            \frac{n}{\alpha \log{n}} \sum_{j\in A_t^{l} \cap S_{t}} \dG{}{j} u_j^2  - 3m\cdot z\\&\numge{4}
            \frac{n}{27\alpha \log{n}}\sum_{i\in A_t^{l}\setminus S_t} m_i u_i^2 + 
            \frac{n}{\alpha \log{n}} \sum_{j\in A_t^{l} \cap S_{t}} \dG{}{j} u_j^2  - 3m\cdot z\\&\numge{5}
            \frac{n}{27\alpha \log{n}}\sum_{i\in A_t^{l}} m_i u_i^2 - 3m\cdot z \ge
            \frac{n}{3000\alpha \log{n}}\sum_{i\in A_t} \dG{}{i} u_i^2 - 3m\cdot z \ .
        \end{align*} 

        Inequality $(2)$ is due to Lemma \ref{lemma:RST}(1) and because each $i\in A^l_t \setminus S_t$ is matched with exactly $m_i$ vertices in $A^r_t\setminus S_t$, inequality $(4)$ follows from Lemma \ref{lemma:RST}(3), inequality $(5)$ is true because $m_i\le \dG{}{i}$ for all $i\in A^l_t$, and the last inequality follows from Lemma \ref{lemma:RST}(4).

        Finally, in expectation,
        \begin{align*}
            &\EE\left[\frac{1}{3}\sum_{\{i,k\}\in \M{t}} \norm{\left(\frac{\W{t}{i}}{\sqrt{\dG{}{i}}} - \frac{\W{t}{k}}{\sqrt{\dG{}{k}}}\right)}_2^2 + \sum_{j\in S_{t}} \dG{}{j} \norm{\frac{\W{t}{j}}{\sqrt{\dG{}{j}}}}_2^2\right]
            \\
            &\ge \frac{n}{3000\alpha \log{n}}\sum_{i\in A_t} \dG{}{i} \EE[u_i^2] - 3m\cdot\EE[z]\\&=
            \frac{1}{3000\alpha \log{n}}\sum_{i\in A_t} \dG{}{i} \norm{\frac{\W{t}{i}}{\sqrt{\dG{}{i}}}}_2^2 - 3m\cdot\EE[z] \\&=
            \frac{1}{3000\alpha \log{n}}\sum_{i\in A_t} \norm{\W{t}{i}}_2^2 - 3m\cdot\EE[z] \\&=
            \frac{1}{3000\alpha \log{n}}\psi(t) - 3m\cdot\EE[z]
        \end{align*}
        where the first equality follows from Lemma \ref{lemma:projection_pairs}(1).

        We note that $z \le \frac{4n}{\alpha \log n} \le n$. Indeed, for every $i \in \A{t}$
        \[
        u_i = \left\langle \frac{\W{t}{i}}{\sqrt{\dG{}{i}}}, r \right\rangle \le
        \norm{\frac{\W{t}{i}}{\sqrt{\dG{}{i}}}}_2 \le 
        \norm{\W{t}{i}}_2 \le 1,
        \]
        where the first inequality holds due to the Cauchy-Schwartz inequality, since $r$ is a unit vector, and the last inequality holds since all eigenvalues of $\W{t}$ are in $[0,1]$: By Lemma~\ref{lemma:normalized-laplacian-eigenvalues} we get that $\Dminushalf \F{t} \Dminushalf = I - \mathcal{N}(\F{t})$ has eigenvalues in $[-1,1]$. Since $\Pmat{t}$ is a projection matrix, it has eigenvalues in $[0,1]$. Therefore, $\Pmat{t} \Dminushalf \F{t} \Dminushalf \Pmat{t}$ has eigenvalues in $[-1,1]$. Finally, $\W{t}= (\Pmat{t} \Dminushalf \F{t} \Dminushalf \Pmat{t})^\delta$ is positive semi-definite since $\delta$ is a power of $2$, so its eigenvalues are in $[0,1]$.

        By  Lemma~\ref{lemma:projection_pairs}, $z$ is non-zero with probability at most $\frac{1}{n^{\alpha/8}}$, so 
        $3m \EE[z] \le \frac{3n^3}{n^{\alpha/8}}  = \frac{3}{n^{(\alpha-24)/8}} \le \frac{3}{n^{\alpha/16}}$ (since $\alpha > 48)$, completing the proof.
    \end{proof}

    The following two corollaries follow by Lemmas \ref{lemma:our_potential_step_1} and \ref{lemma:our_potential_step_2} (see corollaries \ref{cor:krv_potential_reduction} and \ref{cor:krv_total_potential} for similar arguments).

\begin{corollary}
    \label{cor:our_potential_reduction}
        For each round $t$, $\EE[\psi(t+1)] \le
        \left(1-\frac{1}{3000 \alpha \log n}\right)\psi(t) + \frac{3}{n^{\alpha/16}}$, where the expectation is over the unit vector $r\in \RR^n$ of iteration $t+1$ and conditioned on $\psi(t)$.

\end{corollary}

\begin{corollary} [Decrease in Potential]
    \label{cor:our_total_potential}
        With high probability over the choices of $r$, $\psi(T) \le \frac{1}{n}$.
    \end{corollary} 

    The following lemma uses the low potential to derive the near-expansion of $A_T$ in $F_T$.

    \begin{lemma}[Variation of Cheeger's inequality]
    \label{lemma:our_F_expander}
    Let $H=(V,\bar{E})$ be a graph on $n$ vertices, such that $\F{T}$ is its  weighted adjacency matrix. Assume that $\psi(T) \le \frac{1}{n}$. Then, $A_T$ is a near $\frac{1}{5}$-expander in $H$.
    \end{lemma}
    \begin{proof}
        Recall (Lemma \ref{lemma:our_basic_properties}) that $\F{T}$ is symmetric and $\dG$-stochastic. Let $k = \vol(A_T)$.
        Let $S\subseteq A_T$ be a cut, and denote $\dG{S}\in\RR^n$ to be the vector where 
        $\dG{S}{u} = \left\{\begin{array}{cl}
            \dG{}{u} & \mbox{if $u\in S$,} \\
            0 &\mbox{otherwise.}
            \end{array}\right.$ 
            Additionally, denote $\ell = \vol(S) \le \frac{1}{2}k$. Note that $\norm{\sqrt{\dG{S}}}_2^2 = \ell$. 

        Denote by $\bar{\lambda} \ge 0$ the largest singular value of $\X{T} \defeq \Pmat{T} \Dminushalf \F{T} \Dminushalf \Pmat{T}$ (square root of the largest eigenvalue of $(\Pmat{T} \Dminushalf \F{T} \Dminushalf \Pmat{T})^2$).
        Because $\tr(\X{T}^{2 \delta}) = \psi(T) \le \frac{1}{n}$, we have in particular that the largest eigenvalue of $\X{T}^{2 \delta}$ is at most $\frac{1}{n}$, so we have $\bar{\lambda} \le \frac{1}{n^{\frac{1}{2\delta}}}$.
        We choose $\delta  = \Theta(\log n)$ such that $\frac{1}{n^{\frac{1}{2\delta}}}\le \frac{1}{20}$, so $\bar{\lambda}\le\frac{1}{20}$.

        In order to prove near-expansion we need to lower bound $|\bar{E}_{\F{T}}(S,V\setminus {S})|$. We do so by upper bounding $|\bar{E}_{\F{T}}(S,S)| =  \1_S' \F{T} \1_S$. Note that because $S\subseteq A_T$, $\1_S' \F{T} \1_S = \1_S' (\I{T} \F{T} \I{T}) \1_S$. Observe the following relation between $\X{T}$ and $\I{T} \F{T} \I{T}$: 
        \begin{align*}
            D^{\frac{1}{2}}\X{T}D^{\frac{1}{2}} &= D^{\frac{1}{2}}(\Pmat{T} \Dminushalf \F{T} \Dminushalf \Pmat{T}) D^{\frac{1}{2}} 
            \\
            &=
            D^{\frac{1}{2}}(\I{T} - \frac{1}{k}\sqrt{\dG{T}}\sqrt{\dG{T}'}) \Dminushalf \F{T} \Dminushalf 
            (\I{T} - \frac{1}{k}\sqrt{\dG{T}}\sqrt{\dG{T}'}) D^{\frac{1}{2}} 
            \\
            &= (\I{T} - \frac{1}{k}\dG{T}\1_T') \F{T} (\I{T} - \frac{1}{k}\1_T\dG{T}') 
            \\
            &=
            \I{T} \F{T} \I{T} - \frac{1}{k}\dG{T}\1_T' \F{T} \I{T} - \frac{1}{k} \I{T} \F{T} \1_T\dG{T}' + 
            \frac{1}{k^2} \dG{T}\1_T' \F{T} \1_T \dG{T}'.
        \end{align*}

        Rearranging the terms, we get
        \begin{align*}
            \I{T} \F{T} \I{T} = 
            D^{\frac{1}{2}} \X{T} D^{\frac{1}{2}} 
            +\frac{1}{k}\dG{T}\1_T' \F{T} \I{T}
            +\frac{1}{k} \I{T} \F{T} \1_T\dG{T}'
            -\frac{1}{k^2} \dG{T}\1_T' \F{T} \1_T \dG{T}'  \ .
        \end{align*}

        Therefore
         \begin{align*}
            |\bar{E}_{\F{T}}(S,S)| = & \1_S' \F{T} \1_S =
            \1_S' \left(
            D^{\frac{1}{2}}\X{T} D^{\frac{1}{2}} +
            \frac{1}{k}\dG{T}\1_T' \F{T} \I{T} + 
            \frac{1}{k} \I{T} \F{T} \1_T\dG{T}' - 
            \frac{1}{k^2} \dG{T}\1_T' \F{T} \1_T \dG{T}'
            \right) \1_S .
        \end{align*}

        We analyze the summands separately. The first summand can be bounded using $\bar{\lambda}$, the largest singular value of $\X{T}$:
         \begin{align*}
            \1_S' D^{\frac{1}{2}}\X{T}D^{\frac{1}{2}} \1_S = \sqrt{\dG{S}'} X \sqrt{\dG{S}} =
            \left\langle \sqrt{\dG{S}}, X \sqrt{\dG{S}} \right\rangle \le
            \norm{\sqrt{\dG{S}}}_2 \norm{\X{T}\sqrt{\dG{S}}}_2 \le \norm{\sqrt{\dG{S}}}_2^2 \bar{\lambda} \le \frac{\ell}{20},
        \end{align*}
          where the first inequality is the Cauchy-Schwartz inequality. Observe that the second and third summands are equal:
        \begin{align*}
            \frac{1}{k} \1_S' \dG{T}\1_T' \F{T} \I{T} \1_S =
            \frac{\ell}{k} \1_T' \F{T} \1_S = 
            \frac{\ell}{k} \1_S' \F{T} \1_T = 
            \frac{1}{k} \1_S' \I{T} \F{T} \1_T \dG{T}' \1_S,
        \end{align*}

        where the second equality follows by transposing and since $\F{T}$ is symmetric. We now bound the sum of the second, third and fourth summands:
        \begin{align*}
            &\1_S' \left(
            \frac{1}{k}\dG{T}\1_T' \F{T} \I{T} + 
            \frac{1}{k} \I{T} \F{T} \1_T\dG{T}' - 
            \frac{1}{k^2} \dG{T}\1_T' \F{T} \1_T \dG{T}'\right) \1_S =
            \frac{2\ell}{k} \1_T' \F{T} \1_S - \frac{\ell^2}{k^2}\1_T' \F{T} \1_T \\&\le
            \left( \frac{2\ell}{k} - \frac{\ell^2}{k^2} \right)\1_T' \F{T} \1_S \le
            \left( \frac{2\ell}{k} - \frac{\ell^2}{k^2} \right)\1' \F{T} \1_S = 
            \left( \frac{2\ell}{k} - \frac{\ell^2}{k^2} \right)\dG' \1_S = 
            \frac{\ell}{k}\left( 2 - \frac{\ell}{k} \right) \ell,
        \end{align*}

        where the first inequality follows since $S \subseteq \A{t}$. Note that $\frac{\ell}{k}\in [0,\frac{1}{2}]$. The last inequality is true because for $\frac{\ell}{k}$ in this range, $\left(\frac{2\ell}{k}-\frac{\ell^2}{k^2}\right)\ge 0$. Moreover, because $\frac{\ell}{k}\in \left[0,\frac{1}{2}\right]$, we have $\frac{\ell}{k}\left(2-\frac{\ell}{k}\right) \le \frac{3}{4}$. Therefore, $|\bar{E}_{\F{T}}(S,S)| \le \frac{1}{20}\ell + \frac{3}{4}\ell = \frac{4}{5}\ell$, and
        \begin{align*}
            |\bar{E}(S,V\setminus {S})| &= \sum_{u\in S}{\sum_{v\in V\setminus {S}}{\F{T}{u,v}}} = \sum_{u\in S}{\sum_{v\in V}{\F{T}{u,v}}} - \sum_{u\in S}{\sum_{v\in S}{\F{T}{u,v}}}
            \\
            &= \sum_{u\in S}{\dG{}{u}} - \sum_{u\in S}{\sum_{v\in S}{\F{T}{u,v}}} \ge \ell - \frac{4}{5}\ell = \frac{\ell}{5} \ .
        \end{align*}
     So $\Phi_G(S, V\setminus {S}) = \frac{|\bar{E}(S,V\setminus {S})|}{\vol(S)} \ge \frac{1}{5}$, and this is true for all cuts $S\subseteq A$ with $\frac{\vol(S)}{\vol(A)}\le\frac{1}{2}$.

\end{proof}

    \begin{corollary}\label{cor:our_G_expander}
        If we reach round $T$, then with high probability, $A_T$ is a near $\Omega(\phi)$-expander in $G$.
    \end{corollary}
    \begin{proof}
        Assume we reach round $T$. By Corollary \ref{cor:our_total_potential} and Lemma \ref{lemma:our_F_expander}, with high probability, $A_T$ is a near $\Omega(1)$-expander in $\F{T}$. By Corollary~\ref{lemma:our_F_t_embeddable_in_G_t}, $\F{T}$ is embeddable in $G_T$ with congestion $O(\frac{1}{\delta})$. Note that $G_T$ is a union of $T$ $\dGG$-matchings $\{\M{t}\}_{t=1}^T$, each having $\dG{\M{t}} = \dGG = \dG{\F{T}}$. Therefore, $\dG{G_T} = T\cdot \dG{\F{T}}$. 
        So by Lemma \ref{lemma:near_expansion_and_embedding}, $A_T$ is a near $\Omega(\frac{\delta}{T})$-expander in $G_T$. By Lemma~\ref{lemma:our_G_t_embeddable_in_G}, $G_T$ is embeddable in $G$ with congestion $cT$. Together with  the fact that $\dGG = \frac{1}{T} \cdot \dG{G_T}$, we get by Lemma \ref{lemma:near_expansion_and_embedding} again, that $A$ is a near  $\Omega(\frac{\delta}{cT})$-expander in $G$. Recall that $c=O\left(\frac{1}{\phi\log n}\right)$, $\delta = \Theta(\log n)$, and $T=O(\log^2 n)$. Therefore, $A$ is an near $\Omega(\phi)$-expander in $G$.
    \end{proof}

    \subsection{Proof of Theorem \ref{theorem:our}}
    \label{section:our_theorem_proof}

We are now ready to prove Theorem~\ref{theorem:our}. 
\begin{proof} [Proof of Theorem \ref{theorem:our}]
    Recall that $S_t$ denotes the cut returned by Lemma \ref{lemma:unit_flow} at iteration $t$, so that $A_{t+1} = A_{t}\setminus S_{t}$. 
    Observe first that in any round $t$, we have $\Phi_{G}(A_t, R_t)\le\frac{7}{c}=O(\phi\log n)$. 
	    This is because $R_t=\bigcup_{0\le t' < t}{S_{t'}}$ and by Lemma \ref{lemma:unit_flow}, for each $t'$, $\Phi_{G[A_{t'}]}(S_{t'}, V\setminus {S_{t'}})\le\frac{7}{c}=O(\phi\log n)$. 

    Assume Algorithm \ref{algo:cut_matching} terminates because $\vol(R_t)>\frac{m\cdot c\cdot \phi}{70}=\Omega(\frac{m}{\log n})$. We also have, by Lemma \ref{lemma:unit_flow}, that $\vol(A_t)= \Omega(m)= \Omega(\frac{m}{\log n})$. Then $(A_t,R_t)$ is a balanced cut where $\Phi_{G}(A_t, R_t)=O(\phi\log n)$. We end in Case (2) of Theorem \ref{theorem:our}. 

	    Otherwise, the algorithm reached round $T$ and we apply Corollary \ref{cor:our_G_expander}. If $R=\emptyset$, then we obtain the first case of Theorem \ref{theorem:our} because the whole vertex set $V$ is, with high probability, a near $\Omega(\phi)$-expander, which means that $G$ is an $\Omega(\phi)$-expander. Otherwise, we write $c = \frac{c_1}{\phi\log n}$ for some constant $c_1$, and let $c_0 \defeq \frac{7}{c_1}$. We have $\Phi_{G}(A_T, R_T)\le\frac{7}{c}=\frac{7}{c_1} \phi \log n = c_0 \phi\log n$.
    Additionally, $\vol(R_T)\le\frac{m\cdot c\cdot \phi}{70}=\frac{m\cdot c_1}{70\log n}=\frac{m}{10c_0\log n}$, and, with high probability, $A_T$ is a near $\Omega(\phi)$-expander in $G$, which means we obtain the third case of Theorem \ref{theorem:our}.

    Finally, the running time of Algorithm  \ref{algo:cut_matching} is $O(m\log^5 n + \frac{m\log^2 n}{\phi})$:
    The time of iteration $t$ is the sum of the running times of the following steps:
    \begin{enumerate}
        \item Sample a random unit vector $r\in \RR^n$.
        \item Compute the projections vector $u = \Dminushalf \W{t} \cdot r$. This takes $O(t\cdot\delta\cdot m)=O(m\cdot t\cdot\log n)$ time since $\W{t}$ is a multiplication of $O(t \cdot \delta)$ matrices, where each matrix either has $O(m)$ non-zero entries or is a projection matrix $\Pmat{t}$.
        \item Computing $A^l_t$ and $A^r_t$ in time $O(n\log n)$ (Lemma \ref{lemma:RST}).
        \item Computing the cut $S_t$ and the flow $f$ on $G - S_t$ in time $O(\frac{m}{\phi})$ (Lemma \ref{lemma:unit_flow}).
        \item Moving $S_t$ from $A_t$ to $R_{t+1}$ in time $O(m)$.
        \item Constructing $M_{t}$ in time $O(m\log n)$ (using dynamic trees \cite{ST83}).
    \end{enumerate}
    This gives a total running time of $O(mt\log n + \frac{m}{\phi})$ for iteration $t$. As $t$ ranges from $1$ to $T=\Theta(\log^2 n)$, the total time will be $O(m\log^5 n + \frac{m\log^2 n}{\phi})$.
    This completes the proof of Theorem \ref{theorem:our}.

\end{proof}

\bibliography{refs}{}
\bibliographystyle{alpha}

\newpage

\appendix

 \section{KRV's Cut-Matching Game for Conductance}
	\label{section:krv}
 	\label{appendix:krv}
    In this section we show how to modify the cut-matching framework of \cite{khandekar2009graph}
    to bound the conductance of the graph 
 (rather than its expansion). It expands on the overview in subsection \ref{section:krv-short}. We prove the following theorem.
    \begin{theorem}[Theorem \ref{theorem:krv}]
        Given a graph $G$ and a parameter $\phi > 0$, there exists a randomized algorithm, whose running time is dominated by computing a polylogarithmic number of max flow problems, that either 

        \begin{enumerate}
            \item Certifies that $\Phi(G)= \Omega(\phi)$ with high probability; or
            \item Finds a cut $(S, \bar{S})$ in $G$ whose conductance is $\Phi_G(S,\bar{S}) = O(\phi\log^2 n)$.
        \end{enumerate}    
    \end{theorem}

    The algorithm is based on the cut-matching game, defined in Section \ref{section:cut-matching-conductance}.
    Section \ref{section:krv-short} presented a general strategy for the cut player, which achieves 
     \emph{quality} of
    $r(n) \coloneqq \Phi(G_T) =\Omega(1/\log^2 n)$. It also showed a strategy for the matching player, tailored for the graph $G$, that either finds a sparse cut in $G$ or finds a $d$-matching that can be embedded in $G$ with low congestion.

    \subsection{The Algorithm  }\label{appendix:krv-algorithm}

KRV's algorithm is shown in Algorithm \ref{algo:krv_cut_matching}. 
It has
 $T=\Theta(\log^2 n)$ rounds.
   We maintain a sparse representation of a $\dGG$-stochastic flow matrix $\F{t}\in \RR^{n\times n}$. For the purpose of analysis, we maintain a series of graphs $G_t$ ($t$ is the index of the round). The update of the matrix $F_t$, as detailed in Section \ref{section:krv-short}, is summarized in Algorithm \ref{algo:krv_update}. We show that $\F{t}$ is embeddable in $G_t$ with congestion $O(1)$, and that $G_t$ is embeddable in $G$ with congestion $ct$, where $c=\Theta\left(\frac{1}{\phi\log^2 n}\right)$ is an integer.
   At the beginning, $\F{0} = \D = \diag(\dG)$, and $G_0$ is the empty graph on $V = [n]$. 
    In each iteration, the cut player computes a cut according to $F_t$ and the matching player either discovers a low conductance cut or returns a $d$-matching which is embeddable in $G$ with constant congestion (and the cut player updates $F_t$ accordingly). In the former case, we stop and return the cut (Lemma \ref{lemma:krv_3.7_small_flow} shows that it is indeed a sparse cut). If we did not stop after $T=\Theta(\log^2 n)$ rounds, $\F{T}$ will have constant conductance with high probability, which implies that $\Phi(G)=\Omega(\phi)$, so we can terminate with Case (1) of Theorem \ref{theorem:krv}. 
\begin{algorithm}[hbt!]
        \caption{KRV Cut-Matching for Conductance \cite{khandekar2009graph}}
        \label{algo:krv_cut_matching}
        \begin{algorithmic}[1]
            \Function{KRV-Cut-Matching}{$G, \phi$}
                \State $T\gets\Theta(\log^2 n)$.
                \State Set $F\gets D$.
                \For{$t = 1, 2, \ldots, T$} 
                    \State $F \gets $ \Call{KRV-Update-F}{$G, \phi, F$}. \Comment{See Algorithm \ref{algo:krv_update} and Section \ref{section:krv-short}.}
                    \If {the update returned a cut $(S, \bar{S})$}
                        \State \Return{$(S, \bar{S})$}. \Comment{Case (2) of Theorem \ref{theorem:krv}.}
                    \EndIf
                \EndFor
                \State Certify that $\Phi(G)=\Omega(\phi)$. \Comment{Case (1) of Theorem \ref{theorem:krv}.}
            \EndFunction
        \end{algorithmic}
    \end{algorithm}
    \begin{algorithm}[hbt!]
        \caption{KRV Round Update \cite{khandekar2009graph}. See Section \ref{section:krv-short} for details.}
        \label{algo:krv_update}
        \begin{algorithmic}[1]
            \Function{KRV-Update-F}{$G, \phi, F$}
                \State $c\gets\Theta\left(\frac{1}{\phi \log^2 n}\right)$. \Comment{$c$ is an integer.}
                \smallskip
                \State \underline{\textbf{Cut Player:}}
                \smallskip
                \State $r\gets$ Random unit vector of $\RR^{n}$ orthogonal to $\sqrt{\dGG}$.
                \State $u \gets D^{-1} \F\Dminushalf \cdot r$.
                \State Sort the entries of $u$ as $u_{i_1} \le \cdots \le u_{i_n}$.
                \State $Q \gets (i_1, i_1, \ldots, i_1, i_2, i_2, \ldots, i_2, \ldots, i_n, \ldots, i_n)$. \Comment{Each $i_j$ appears $\dGG{i_j}$ times.}
                \State $L \gets (Q_1, \ldots, Q_m)$, $R \gets (Q_{m+1}, \ldots, Q_{2m})$.
                \smallskip
                \State \underline{\textbf{Matching Player:}}
                \smallskip
                \State $G' \gets (V', E')$, where $V' \gets V\cup \{s, t\}$, $E' \gets E \cup (\{s\} \times L) \cup (R \times \{t\})$.
                \State Let $m_v$, $\bar{m_v}$ be the number of times $v$ appears in $L$, $R$, respectively.
                \State Set the capacity of $e\in E$ to $c$, of $(s, v) \in \{s\}\times L$ to $m_v$ and of $(v, t)\in R\times \{t\}$ to $\bar{m_v}$.

                \State Compute a maximum flow $g$ from $s$ to $t$ in $G'$.
                \If{$|g| < m$} \Comment{$|g|$ is the value of the flow $g$.}
                    \State Find a minimum $(s,t)$-cut $(S,\bar{S})$ in $G'$.
                    \State \Return{$(S\cap G, \bar{S}\cap G)$}.
                \Else
                    \State Decompose $g$ into a set of paths $\{u_j \to v_j\}_j$, where $u_j \in L$ and $v_j \in R$.
                    \State $M \gets \{\{u_j, v_j\}\}_{j=1}^m$. 

                    \\\Comment{$M$ is a symmetric $n\times n$ matrix, $M(u,v)$ is the number of paths between $u$ and $v$.}
                    \smallskip
                    \State \underline{\textbf{Update of $F$:}}
                    \smallskip
                    \State $F_{\text{new}} \gets \frac{1}{2}\left(I + \M\cdot D^{-1}\right)\F$.
                    \State \Return{$F_{\text{new}}$}.
                \EndIf
            \EndFunction
        \end{algorithmic} 

    \end{algorithm}

    \begin{lemma}[Similar to Lemma 3.7 of~\cite{khandekar2009graph}]
    \label{lemma:krv_3.7_small_flow}
        If the auxiliary maximum flow problem in $G'$ has value smaller than $m$ then $\Phi(G) < \frac{1}{c}$.
    \end{lemma}

    \begin{proof}
        Consider a minimum $(s,t)$-cut, $(S,\bar{S})$, in $G' = G \cup \{s, t\}$. Note that because $m > |E_{G'}(S,\bar{S})|$, $(S,\bar{S})$ is nontrivial in $G$. I.e., it is not $(\{s\}, G'\setminus \{s\})$ or $(G'\setminus \{t\}, \{t\})$. We get 
        \[m > |E_{G'}(S,\bar{S})| = \vol(L \cap \bar{S}) + \vol(R \cap S) + c\cdot |E_G(S,\bar{S})| \ ,\]
        where the first inequality follows since the maximum flow has value smaller than $m$.
        By rewriting the terms, it follows $|E_G(S,\bar{S})| < \frac{1}{c} \cdot (m - \vol(L \cap \bar{S}) - \vol(R \cap S))$. Therefore, 
\begin{align*}
            \Phi_G(S, \bar{S}) = \frac{|E_G(S,\bar{S})|}{\min (\vol(S),\vol(\bar{S}))} \le \frac{|E_G(S,\bar{S})|}{\min (m - \vol(L \cap \bar{S}), m - \vol(R \cap S))} < \frac{1}{c} \ ,
        \end{align*}

        where the first inequality follows since $\vol(S) + \vol(L \cap \bar{S}) \ge \vol(L) = m$ and $\vol(\bar{S}) + \vol(R\cap S) \ge \vol(R) = m$.
    \end{proof}

    The rest of this section is organized as follows.  Subsection \ref{section:krv_F_embeddable_G} shows that $\F{t}$ is embeddable in $G_t$ with congestion $1$ and that $G_t$ is embeddable in $G$ with congestion $c\cdot t$. Subsection \ref{section:krv_F_t_expander} shows that if we reach round $T$, then with high probability, $\Phi(G)=\Omega(\phi)$. Finally, in Subsection \ref{section:krv_theorem_proof} we prove Theorem \ref{theorem:krv}.

    \subsection{$\F{t}$ is embeddable in $G$}
    \label{section:krv_F_embeddable_G}
\begin{lemma}
    \label{lemma:krv_F_t_d_stochastic}
        For all rounds $t$, $\F{t}$ is a $\dG$-stochastic matrix.
    \end{lemma}
    \begin{proof}
        By induction on $t$. $\F{0} = D = \diag(\dG)$ is clearly $\dG$-stochastic. 
        After step $t$, 
        \begin{align*}
            \F{{t+1}}\1_n &= \frac{1}{2}\left(I + \M{t}\cdot D^{-1}\right)\F{t}\cdot\1_n = \frac{1}{2}\left(I + \M{t}\cdot D^{-1}\right)\dG 
            \\
            &= \frac{1}{2}\left(\dG + \M{t}\cdot \1_n\right) = \frac{1}{2}\left(\dG + \dG\right) = \dG \ .
        \end{align*}
        Similarly,
        \begin{align*}
            \1'_n \cdot \F{{t+1}} &= \1'_n\frac{1}{2}\left(I + \M{t}\cdot D^{-1}\right)\F{t} = \frac{1}{2}\left(\1'_n + \1'_n\cdot \M{t}\cdot D^{-1}\right)\F{t} 
            \\
            &= \frac{1}{2}\left(\1'_n + \dG'\cdot D^{-1}\right)\F{t} = \frac{1}{2}\left(\1'_n + \1'_n\right)\F{t} = \1'_n \F{t} = \dG' \ .
        \end{align*}
    \end{proof}
\begin{lemma}
    \label{lemma:krv_F_t_embeddable_in_G_t}
        For all rounds $t$, $\F{t}$ is embeddable in $G_t$ with congestion $1$.
    \end{lemma}
    \begin{proof}

        We proceed by induction. Initially, $\F{0} = I$ is embeddable in $G_0=\emptyset$ with congestion $0$. Next, we assume that $\F{t}$ is embeddable in $G_t$ with congestion $1$, and show that $\F{{t+1}} = \frac{1}{2}(I + \M{t}\cdot D^{-1}) \F{t}$ is embeddable in $G_{t+1} = G_t \cup \M{t}$ with congestion $1$.

        Let $P:E\times V \to \RR_{\ge 0}$ be a routing of $\F{t}$ in $G_t$ with congestion $1$ (where $P((u,v), w)$ indicates how much of $w$'s commodity goes through the edge $(u,v)\in E$). Recall that $\F{{t+1}}$ is the flow matrix obtained by performing a weighted average on the rows of $\F{t}$ described by $\M{t}$. Explicitly, for every $v\in V, a \in V$, we have $\F{t+1}{v, a} = \Matrix{\frac{1}{2}(I + \M{t}\cdot D^{-1}) \F{t}}{v,a} = \frac{1}{2}\F{t}{v,a} + \frac{1}{2}\sum_{u : \{u,v\}\in \M{t}}{\frac{1}{\dG{}{u}}\F{t}{u,a}}$.

	    We define a routing $P'$ of $\F{{t+1}}$ in $G_{t+1}$ using the following operations:

        \begin{algorithm}[H]

            \begin{algorithmic}[]
                \State $P' \leftarrow \frac{1}{2}\cdot P$.
                \For{$\{a,b\}\in \M{t}$}
                    \State $P'((a,b),a) \leftarrow P'((a, b), a) + \frac{1}{2}$.
                    \State $P'((b,a),b) \leftarrow P'((b, a), b) + \frac{1}{2}$.
                \EndFor
                \For{$\{a,b\}\in \M{t}$}
                    \For{$(u,v)\in G_t$}\Comment{The edge $\{u,v\}\in G_t$ is scanned in both directions}
                        \State $P'((u,v),a) \leftarrow P'((u,v),a) + \frac{1}{2\dG{}{b}}P((u,v),b)$.
                        \State $P'((u,v),b) \leftarrow P'((u,v),b) + \frac{1}{2\dG{}{a}}P((u,v),a)$.
                    \EndFor
                \EndFor
            \end{algorithmic}
        \end{algorithm}

        Intuitively, we think of $P'$ in  stages. In the first stage, we scale $P$ by $\frac{1}{2}$. This routes $\frac{1}{2}\F{t}{v,a}$ units of flow (of $v$'s commodity) from $v$ to $a$ for every $v$ and $a$. After this stage, each vertex $v\in V$ currently sends $\frac{\dG{}{v}}{2}$ units of its commodity. In the next stage we wish to route an additional $\frac{1}{2}\sum_{u : \{u,v\}\in \M{t}}{\frac{1}{\dG{}{u}}\F{t}{u,a}}$ units from $v$ to $a$. To this end, we first move $\frac{1}{2}$ units from $v$'s commodity to each $u$ with $\{u,v\}\in \M{t}$ (this flow is sent through the edge $\{u,v\}\in G_{t+1}$). Now, each vertex $u\in V$ ``mixes'' the commodities it got from its neighbors and routes the $\frac{\dG{}{u}}{2}$ ``new'' units according to $P$, as if they were of its own commodity. 
Thus, $\frac{\F{t}{u,a}}{2}$ is the total flow sent from $u$ to $a$ when it routes the $\frac{\dG{}{u}}{2}$ ``new'' units and the ``share'' of $v$'s commodity from this flow is $\frac{1}{\dG{}{u}}$.
It follows that out of the $\frac{1}{2}$ unit $v$ sent to a $u$, exactly $\frac{1}{\dG{}{u}}\frac{\F{t}{u,a}}{2}$ units go to $a$. 

	    We now argue that the congestion of $P'$ is $1$.\footnote{A stronger claim holds: The congestion of every directed edge is $\frac{1}{2}$.} The edges of $\M{{t}}$ have congestion $1$ (because we route through them $\frac{1}{2}$ unit of flow in each direction), and the congestion on an edge $\{u,v\}\in G_t$ consists of at most $\frac{1}{2}$ from the routing of $\frac{1}{2}P$, plus an additional $\frac{1}{2\dG{}{b}}P(\{u,v\},b)+\frac{1}{2\dG{}{a}}P(\{u,v\},a)$ for each $\{a,b\}\in \M{t}$. In total, the congestion on $\{u,v\}$ is at most
	    \begin{align*}
	        \frac{1}{2} &+ \sum_{\{a,b\}\in \M{t}}{\left(\frac{1}{2\dG{}{b}}P(\{u,v\},b)+\frac{1}{2\dG{}{a}}P(\{u,v\},a)\right)} \numeq{1} \frac{1}{2} + \sum_{a\in V}{\dG{}{a}\cdot\frac{1}{2\dG{}{a}}P(\{u,v\},a)} 
	        \\
	        &= \frac{1}{2} + \frac{1}{2}\sum_{a\in V}{P(\{u,v\},a)} \numle{2} \frac{1}{2} + \frac{1}{2}\cdot 1 = 1 \ ,
	    \end{align*}
	    where Equality $(1)$ follows since every vertex $a\in V$ appears $d(u)$ times in the matching $\M{t}$ and Inequality $(2)$ holds since $\sum_{a\in V}{P(\{u,v\},a)}$ is the congestion of $P$ on the edge $(u,v)$, which is at most $1$ by the induction hypothesis. Note that in fact, in our construction, the congestion on an edge $\{u,v\}\in G_t$  
         did not change. Only the mix of commodities that flow on the edge was modified.

    \end{proof}
\begin{lemma}
    \label{lemma:krv_G_t_embeddable_in_G}
        For all rounds $t$, $G_t$ is embeddable in $G$ with congestion $c\cdot t$.
    \end{lemma}
    \begin{proof}
        For every $t$, by the definition of the flow problem at round $t$, $\M{t}$ is embeddable in $G$ with congestion $c$. Summing these routings gives a routing of $G_t = \bigcup_{i=1}^t{\M{i}}$ in $G$ with congestion $c\cdot t$.
    \end{proof}

    \subsection{$\F{T}$ is an expander}
    \label{section:krv_F_t_expander}
    In this section we
     prove that after $T=\Theta(\log^2 n)$ rounds, with high probability, $\Phi(\F{T}) = \Omega(1)$, which implies $\Phi(G) = \Omega(\phi)$.
    Consider the potential:
    \[\psi(t) = \sum_{i\in V}{\sum_{j\in V}{\frac{1}{\dG{}{i}\cdot \dG{}{j}}\left(\F{t}{i,j}-\frac{\dG{}{i} \dG{}{j}}{2m}\right)^2}}\]

    This potential represents the distance between the normalized adjacency matrix $\Dminushalf \F{t} \Dminushalf$ and the matrix $\sqrt{\dGG} \sqrt{\dGG}'/2m$. 
    Note that if the graph is regular (\ie, $d(i)$ are all equal), then $\sqrt{\dGG} \sqrt{\dGG}'/2m$ is the uniform matrix $\frac{\1\cdot\1'}{n}$.

    Note that 
    \begin{align}
    \label{eq:krv_psi_0}
        \psi(0) &= \sum_{i\in V}{\sum_{j\in V}{\frac{1}{\dG{}{i}\cdot \dG{}{j}}\left(\D{}{i,j}-\frac{\dG{}{i} \dG{}{j}}{2m}\right)^2}} 
        \notag \\
        &= \sum_{i\in V}{\sum_{j\in V\setminus\{i\}}{\frac{1}{\dG{}{i}\cdot \dG{}{j}}\left(-\frac{\dG{}{i} \dG{}{j}}{2m}\right)^2}} + \sum_{i\in V}{\frac{1}{\dG{}{i}^2}\left(\dG{}{i}-\frac{\dG{}{i}^2}{2m}\right)^2} 
        \notag \\
        &= \sum_{i\in V}{\sum_{j\in V\setminus\{i\}}{\frac{\dG{}{i} \dG{}{j}}{4m^2}}} + \sum_{i\in V}{\left(1-\frac{\dG{}{i}}{2m}\right)^2} = \sum_{i\in V}{\frac{\dG{}{i}}{2m}\left(1 - \frac{\dG{}{i}}{2m}\right)} + \sum_{i\in V}{\left(1-\frac{\dG{}{i}}{2m}\right)^2}
        \notag \\
        &= \sum_{i\in V}{\left(1 - \frac{\dG{}{i}}{2m}\right)} = n-1    
    \end{align}
\begin{lemma}
    \label{lemma:krv_potential_step_1}
        For each round $t$, $\psi(t) - \psi(t+1) \ge \frac{1}{2}\sum_{\{i,k\} \in \M{t}}{\sum_{j\in V}{\left(\frac{\F{}{i,j}}{\dG{}{i}\sqrt{\dG{}{j}}} - \frac{\F{}{k,j}}{\dG{}{k}\sqrt{\dG{}{j}}}\right)^2}}$.
    \end{lemma}
\begin{proof}
        Rewrite the potential as 
        \begin{equation} \label{eq:psi}
        \psi(t) = \sum_{i\in V}{\sum_{\ell = 1}^{\dG{}{i}}{\sum_{j\in V}{\left(\frac{\F{t}{i,j}}{\dG{}{i}\sqrt{\dG{}{j}}} - \frac{\sqrt{\dG{}{j}}}{2m}\right)^2}}} = \sum_{(i,k)\in \M{t}}{\sum_{j\in V}{\left(\frac{\F{t}{i,j}}{\dG{}{i}\sqrt{\dG{}{j}}} - \frac{\sqrt{\dG{}{j}}}{2m}\right)^2}} \ .
        \end{equation}
        Note that each edge $\{i,k\}\in \M{t}$ appears in the sum twice, once as $(i,k)$ and once as $(k,i)$, and therefore the $i$'th row of $\F{t}$ appears in $\dG{}{i}$ summands. Thus, we think of the $i$'th row of $\F{t}$ as the sum of $\dG{}{i}$ equal ``fractional rows'', each corresponding to a different edge $(i, k) \in \M{t}$ incident to $i$. 

        At a high level, in order to bound the potential decrease, we view the decrease in two stages. In the first stage we average the fractional rows according to $\M{t}$: For each edge $\{i, k\}\in \M{t}$ we average the fractional row of $i$ corresponding to $(i, k)$ with a fractional row of $k$ corresponding to $(k, i)$.
        Note that after this step, the fractional rows corresponding to a vertex $i$ might be different, as each of them was averaged with a fractional row corresponding to different vertex. The potential after this averaging is denoted by $\psi'(t)$. In the second stage, we average the fractional parts of each row, and show that we get $\psi(t+1)$. We show that both stages reduce the potential.
        Formally, define 
        \begin{align}
            \psi'(t) &= \sum_{(i,k) \in \M{t}}{\sum_{j\in V}{\left(\frac{\frac{\F{t}{i,j}}{\dG{}{i}\sqrt{\dG{}{j}}} + \frac{\F{t}{k,j}}{\dG{}{k}\sqrt{\dG{}{j}}}}{2} - \frac{\sqrt{\dG{}{j}}}{2m}\right)^2}} 
            \notag \\
            &= 2\sum_{\{i,k\} \in \M{t}}{\sum_{j\in V}{\left(\frac{\frac{\F{t}{i,j}}{\dG{}{i}\sqrt{\dG{}{j}}} + \frac{\F{t}{k,j}}{\dG{}{k}\sqrt{\dG{}{j}}}}{2} - \frac{\sqrt{\dG{}{j}}}{2m}\right)^2}}\ . \label{eq:psi'}
        \end{align}
        Using (\ref{eq:psi}) and (\ref{eq:psi'}), we can write $\psi(t)-\psi'(t)$ as follows: 
        \begin{align*}
            \psi(t)-\psi'(t) = \sum_{\{i,k\} \in \M{t}} \sum_{j\in V}
            \Bigg[\Bigg. & \left(\frac{\F{t}{i,j}}{\dG{}{i}\sqrt{\dG{}{j}}} - \frac{\sqrt{\dG{}{j}}}{2m}\right)^2 
             + \left(\frac{\F{t}{k,j}}{\dG{}{k}\sqrt{\dG{}{j}}} - \frac{\sqrt{\dG{}{j}}}{2m}\right)^2 
            \\
            & - 2\left(\frac{\frac{\F{t}{i,j}}{\dG{}{i}\sqrt{\dG{}{j}}} + \frac{\F{t}{k,j}}{\dG{}{k}\sqrt{\dG{}{j}}}}{2} - \frac{\sqrt{\dG{}{j}}}{2m}\right)^2 \Bigg.\Bigg] \ .
        \end{align*}
        Note that for two $n$-dimensional vectors $w,v$, $\norm{w}_2^2+\norm{v}_2^2-2\norm{\frac{w+v}{2}}_2^2=\frac{1}{2}\norm{w-v}_2^2$. So we get that
        \begin{align*}
            \psi(t)-\psi'(t) = \frac{1}{2}\sum_{\{i,k\} \in \M{t}}{\sum_{j\in V}{\left(\frac{\F{t}{i,j}}{\dG{}{i}\sqrt{\dG{}{j}}} - \frac{\F{t}{k,j}}{\dG{}{k}\sqrt{\dG{}{j}}}\right)^2}} \ .
        \end{align*}

        It follows that to finish the proof, we only have to show that $\psi(t+1)\le\psi'(t)$.
For this define the vectors $a_{i,k}\in \RR^n$ for $(i,k)\in \M{t}$ as follows: $\MatrixSimple{a_{i,k}}{j} = \frac{\frac{\F{t}{i,j}}{\dG{}{i}\sqrt{\dG{}{j}}} + \frac{\F{t}{k,j}}{\dG{}{k}\sqrt{\dG{}{j}}}}{2} - \frac{\sqrt{\dG{}{j}}}{2m}$.
Using these vectors we can write
       \begin{align} \label{eq:psi-alternate}
            \psi'(t) &= \sum_{(i,k) \in \M{t}}{\sum_{j\in V}{\left(\frac{\frac{\F{t}{i,j}}{\dG{}{i}\sqrt{\dG{}{j}}} + \frac{\F{t}{k,j}}{\dG{}{k}\sqrt{\dG{}{j}}}}{2} - \frac{\sqrt{\dG{}{j}}}{2m}\right)^2}} = \sum_{i\in V}{\sum_{k : (i,k) \in \M{t}}{\norm{a_{i,k}}_2^2}} \ .
        \end{align}
        By rewriting the potential as before we get that 
        \begin{align*}
            \psi(t+1) = \sum_{i\in V}{\sum_{j\in V}{\dG{}{i}\left(\frac{\F{t+1}{i,j}}{\dG{}{i}\sqrt{\dG{}{j}}} - \frac{\sqrt{\dG{}{j}}}{2m}\right)^2}} \ .
        \end{align*}
To present  $\psi(t+1)$ also in terms of the vectors 
$a_{i,k}$ we recall that $\F{t+1}{i,j} = \frac{1}{2}\F{t}{i,j} + \sum_{(i, k) \in \M{t}}{\frac{1}{2\dG{}{k}}\F{t}{k,j}}$. Therefore,
for all $i\in V$, $\sum_{k : (i,k)\in \M{t}}{\MatrixSimple{a_{i,k}}{j}} = \frac{\F{t+1}{i,j}}{\sqrt{\dG{}{j}}} - \frac{\dG{}{i} \sqrt{\dG{}{j}}}{2m}$, so $\sum_{k : (i,k)\in \M{t}}{\frac{\MatrixSimple{a_{i,k}}{j}}{\dG{}{i}}} = \frac{\F{t+1}{i,j}}{\dG{}{i} \sqrt{\dG{}{j}}} - \frac{\sqrt{\dG{}{j}}}{2m}$. This means that 
      \begin{align} \label{eq:psit+1}
            \psi(t+1) &= \sum_{i\in V}{\sum_{j\in V}{\dG{}{i}\left(\frac{\F{t+1}{i,j}}{\dG{}{i}\sqrt{\dG{}{j}}} - \frac{\sqrt{\dG{}{j}}}{2m}\right)^2}} = \sum_{i\in V}{\dG{}{i}\norm{\sum_{k : (i,k)\in \M{t}}{\frac{a_{i,k}}{\dG{}{i}}}}_2^2} \ .
        \end{align}
Comparing (\ref{eq:psi-alternate})
and (\ref{eq:psit+1}) we get that
$\psi(t+1) \le \psi'(t)$ follows from Lemma \ref{lemma:cs_mean} in Appendix \ref{appendix:matrix_inequalities}.
    \end{proof}

    We provide an alternative proof of Lemma \ref{lemma:krv_potential_step_1}, using matrix formulations of the potential and updates. This serves as a warmup for the proofs in Appendix \ref{section:osvv}.
    \begin{proof} [Alternative proof of Lemma \ref{lemma:krv_potential_step_1}]
    Let $\bar{F_t} = \Dminushalf \F{t} \Dminushalf$ and let $P = I - \frac{1}{2m}\sqrt{\dG}\sqrt{\dG'}$ be the projection matrix on the subspace orthogonal to $\sqrt{\dG}$. We first rewrite the potential $\psi(t)$ of KRV in matrix form as
    \begin{align*}
        \psi(t) &= \norm{D^{-\frac{1}{2}} (\F{t} - \frac{1}{2m}\dG \dG') D^{-\frac{1}{2}}}^2_F = \norm{D^{-\frac{1}{2}} \F{t} D^{-\frac{1}{2}}-\frac{1}{2m}\sqrt{\dG}\sqrt{\dG'}}^2_F  = 
        \norm{D^{-\frac{1}{2}} \F{t} D^{-\frac{1}{2}}P}^2_F \\ &=
        \norm{\bar{\F{t}} P}^2_F = 
        \tr \left((\bar{\F{t}}  P)(\bar{\F{t}} P)' \right) 
        = \tr( \bar{\F{t}} P^2 \bar{\F{t}'}) = \tr(  P \bar{\F{t}'} \bar{\F{t}}),
    \end{align*}

     where the third equality holds by Lemma~\ref{lemma:normalized_commutes_with_P} and the last equality is due to Fact~\ref{fact:book_trace_identities} (and since $P^2=P$ as a projection matrix).

     Define $\N{t} \defeq \frac{1}{2}D + \frac{1}{2}\M{t}$, so $\F{{t+1}} = \N{t} D^{-1} \F{t} $. Note that $\N{t} D^{-1}$ is the lazy  uniform random walk defined by $\M{t}$. I.e.,\ with probability $1/2$ do not move and with probability $1/2$ we move on an edge of $M_t$ picked uniformly among the edges incident to the current vertex.
     Denote $\bar{\N{t}} = \Dminushalf \N{t} \Dminushalf$. Observe that $\N{t},\M{t},\F{t}$  are $\dG$-stochastic matrices (moreover, $\N{t},\M{t}$ are symmetric).
    Additionally, 
    \begin{align*}
        \psi(t+1) &= \tr(  P \bar{F}_{t+1}' \bar{F}_{t+1}) = 
        \tr(  P (\Dminushalf \F{t+1} \Dminushalf)' (\Dminushalf \F{t+1} \Dminushalf)) 
        \\ &=
        \tr(  P (\Dminushalf \N{t} \D^{-1} F_t \Dminushalf)' (\Dminushalf \N{t} \D^{-1} F_t \Dminushalf)) 
        \\ &=
        \tr(  P (\bar{\N{t}} \bar{F_t} )' (\bar{\N{t}} \bar{F_t} )) = 
        \tr(  P  \bar{F_t'} \bar{\N{t}}^2 \bar{F_t} ) \numeq{6}
        \tr(  \bar{\N{t}}^2 \bar{F_t} P  \bar{F_t'}) =
        \tr(  \bar{\N{t}}^2 P \bar{F_t}  \bar{F_t'}),
    \end{align*}

    where equality $(6)$ holds by Fact~\ref{fact:book_trace_identities} and the last equality holds by Lemma~\ref{lemma:normalized_commutes_with_P}. 

    The following claim shows the spectral relation between $\bar{\N{t}}$ and the normalized Laplacian $\mathcal{N}(\M{t})$.
\begin{claim}\label{claim:krv-matrices-N_t-relations}
    The following holds:
    \begin{enumerate}
        \item $\bar{\N{t}} = I - \frac{1}{2} \mathcal{N}(\M{t})$.
        \item $I - \bar{\N{t}}^2 \succeq \frac{1}{2} \mathcal{N}(\M{t})$. That is, $I - \bar{\N{t}}^2 - \frac{1}{2} \mathcal{N}(\M{t})$ is \emph{PSD} (positive semi definite).

    \end{enumerate}
    \end{claim}
\begin{proof}
    \begin{enumerate}
        \item Indeed,
        \[
        \normalized{\left(\frac{1}{2}D + \frac{1}{2}\M{t}\right)} = 
            \frac{1}{2}I + \frac{1}{2}\normalized{\M{t}} = 
            I - \frac{1}{2}(I - \normalized{\M{t}}).
        \]
        \item Following from the previous property, $I-\bar{\N{t}}^2$ and  $\mathcal{N}(\M{t})$ have the same eigenvectors with the following relation: if $\mathcal{N}(\M{t})v = \lambda \cdot v$, then $(I-\bar{\N{t}}^2)v = (1- (1-\frac{1}{2}\lambda)^2)\cdot v$. By Lemma~\ref{lemma:normalized-laplacian-eigenvalues}, all eigenvalues of $\mathcal{N}(\M{t})$ are in $[0,2]$ and therefore $1-(1-\frac{1}{2}\lambda)^2 \ge \frac{1}{2}\lambda$.
    \end{enumerate}
    \end{proof}

    Using Claim~\ref{claim:krv-matrices-N_t-relations}, we can bound the potential reduction as follows
    \begin{align*}
        \psi(t) - \psi(t+1) &= \tr((I - \bar{\N{t}}^2) P \bar{F_t}  \bar{F_t'} ) \numeq{2} 
        \tr((I - \bar{\N{t}}^2) P \bar{F_t} P \bar{F_t'} ) \\ &\numeq{3}
        \tr( P \bar{F_t'} (I - \bar{\N{t}}^2) P \bar{F_t} ) \numeq{4} 
        \tr( (P \bar{F_t})' (I - \bar{\N{t}}^2) (P \bar{F_t}) ) \\ &\numge{5}
        \frac{1}{2} \tr\left( (P \bar{F_t})' \mathcal{N}\left(\M{t}\right) (P \bar{F_t}) \right) \\ &=
        \frac{1}{2} \tr\left( (P \bar{F_t})' \normalized{\mathcal{L}\left(\M{t}\right)} (P \bar{F_t}) \right) 
        \\ &=
        \frac{1}{2}  \tr\left( \left(\D^{-1} F_t  \Dminushalf - \frac{1}{2m}\1 \sqrt{\dG}\right)' \mathcal{L}\left(\M{t}\right) 
        \left( \D^{-1} F_t  \Dminushalf - \frac{1}{2m}\1 \sqrt{\dG} \right) \right) \\&=
        \frac{1}{2}\sum_{\{i,k\} \in \M{t}}{\sum_{j\in V}{\left(\frac{\F{}{i,j}}{\dG{}{i}\sqrt{\dG{}{j}}} - \frac{\F{}{k,j}}{\dG{}{k}\sqrt{\dG{}{j}}}\right)^2}},
    \end{align*}
    where the Equalities~$(2)$ and~$(4)$ hold by Lemma~\ref{lemma:normalized_commutes_with_P} (and the fact that $P^2=P$), Equality~$(3)$ holds by Fact~\ref{fact:book_trace_identities}(1), Inequality~$(5)$ holds by Claim~\ref{claim:krv-matrices-N_t-relations}(2)\footnote{For any matrix $A$ and any PSD matrix $B$ it holds that $\tr(A'BA) = \sum_{i=0}^n{A(,i)' B A(,i)}\ge 0$.} and the last equality holds by Lemma~\ref{lemma:matching_laplacian}. This completes the proof of Lemma \ref{lemma:krv_potential_step_1}.
    \end{proof}
\begin{lemma}
    \label{lemma:krv_potential_step_2}
        For each round $t$, $$\EE\left[\frac{1}{2}\sum_{\{i,k\} \in \M{t}}{\sum_{j\in V}{\left(\frac{\F{t}{i,j}}{\dG{}{i}\sqrt{\dG{}{j}}} - \frac{\F{t}{k,j}}{\dG{}{k}\sqrt{\dG{}{j}}}\right)^2}}\right] \ge \frac{1}{2\alpha \log n}\psi(t) - \frac{1}{n^{\alpha/16}},$$ for all $\alpha> 48$, where the expectation is over the random unit vector $r\in \RR^n$ orthogonal to $\sqrt{\dG}$ and conditioned on $\psi(t)$.
    \end{lemma}

 \begin{proof}
        Recall that $u_i = \frac{1}{\dG{}{i}} \langle D^{-\frac{1}{2}}\F{t}{i}, r \rangle$. Notice that $\Matrix{\frac{1}{\dG{}{i}}D^{-\frac{1}{2}}\F{t}{i}}{j} = \frac{\F{t}{i,j}}{\dG{}{i}\sqrt{\dG{}{j}}}$.
        Apply Lemma \ref{lemma:projection_krv} from Appendix \ref{appendix:projection} to the set of vectors $\left\{\frac{1}{\dG{}{i}}D^{-\frac{1}{2}}\F{t}{i} \mid i\in V\right\}\cup\left\{\frac{\sqrt{\dG}}{2m}\right\}$, 
        and $x=\sqrt{\dG}$. The conditions of the lemma are indeed satisfied (for $c = 1$) since
        \begin{align*}
            \forall i\in V: \left\langle \frac{1}{\dG{}{i}}D^{-\frac{1}{2}}\F{t}{i}, \sqrt{\dG} \right\rangle &= \sum_{j\in V}{\frac{\F{t}{i,j}}{\dG{}{i}\sqrt{\dG{}{j}}}\cdot\sqrt{\dG{}{j}}} = \sum_{j\in V}{\frac{\F{t}{i,j}}{\dG{}{i}}} = 1
            \\
            \left\langle\frac{\sqrt{\dG}}{2m}, \sqrt{\dG}\right\rangle &= \frac{1}{2m}\sum_{j\in V}{\dG{}{j}} = 1
        \end{align*}
        (we used the fact that $\F{t}$ is $\dG$-stochastic, which follows from Lemma \ref{lemma:krv_F_t_d_stochastic}). Therefore, by Lemma \ref{lemma:projection_krv}(2),  with probability of at least $1-\frac{1}{n^{\alpha/8}}$:

        \begin{align*}
            \forall i, k\in V: \sum_{j\in V}{\left(\frac{\F{t}{i,j}}{\dG{}{i}\sqrt{\dG{}{j}}} - \frac{\F{t}{k,j}}{\dG{}{k}\sqrt{\dG{}{j}}}\right)^2} &= \norm{\frac{1}{\dG{}{i}}D^{-\frac{1}{2}}\F{t}{i} - \frac{1}{\dG{}{k}}D^{-\frac{1}{2}}\F{t}{k}}_2^2 
            \\
            &\ge \frac{n-1}{\alpha \log n}\cdot(u_i - u_k)^2
        \end{align*}
        We need to replace the high probability inequalities with an inequality of expectations. Similarly to  KRV we use the following argument for this: Introduce a random variable $z$ that is non-zero only when this inequality fails to hold, such that  
        \begin{align*}
		        \forall i, k\in V: \norm{\frac{1}{\dG{}{i}}D^{-\frac{1}{2}}\F{t}{i} - \frac{1}{\dG{}{k}}D^{-\frac{1}{2}}\F{t}{k}}_2^2 \ge \frac{n-1}{\alpha \log n}\cdot(u_i - u_k)^2 - z
        \end{align*}
        holds with probability $1$. \emph{I.e.}, we define
        \begin{align*}
            \mathcal{B} = \{0\} &\cup \left\{\frac{n-1}{\alpha\log n}(u_i - u_k)^2-\norm{\frac{1}{\dG{}{i}}D^{-\frac{1}{2}}\F{t}{i} - \frac{1}{\dG{}{k}}D^{-\frac{1}{2}}\F{t}{k}}_2^2 : (i, k)\in V\times V\right\}
        \end{align*}
        and $z = \max(\mathcal{B})$. We get that
        \begin{align*}
            \frac{1}{2}\sum_{\{i,k\} \in \M{t}}{\sum_{j\in V}{\left(\frac{\F{t}{i,j}}{\dG{}{i}\sqrt{\dG{}{j}}} - \frac{\F{t}{k,j}}{\dG{}{k}\sqrt{\dG{}{j}}}\right)^2}} \ge \frac{n-1}{2\alpha\log n}\sum_{\{i,k\} \in \M{t}}{\left(u_i - u_k\right)^2} - m\cdot z
        \end{align*}
        Any $\{i,k\} \in \M{t}$ satisfies $u_i \le \eta \le u_k$ or $u_i \ge \eta \ge u_k$. Therefore, $(u_i-u_k)^2 \ge (u_i-\eta)^2+(u_k-\eta)^2$. We get:
        \begin{align*}
            &\frac{1}{2}\sum_{\{i,k\} \in \M{t}}{\sum_{j\in V}{\left(\frac{\F{t}{i,j}}{\dG{}{i}\sqrt{\dG{}{j}}} - \frac{\F{t}{k,j}}{\dG{}{k}\sqrt{\dG{}{j}}}\right)^2}} \ge \frac{n-1}{2\alpha\log n}\sum_{i \in V}{\dG{}{i}\left(u_i - \eta\right)^2} - m\cdot z
            \\
            &= \frac{n-1}{2\alpha\log n}\left(\sum_{i \in V}{\dG{}{i} \cdot u_i^2} -2\eta\sum_{i \in V}{\dG{}{i} \cdot u_i} + 2m\eta^2 \right)  - m\cdot z 
            \\
            &= \frac{n-1}{2\alpha\log n}\left(\sum_{i \in V}{\dG{}{i} \cdot u_i^2} + 2m\eta^2 \right)  - m\cdot z
            \ge \frac{n-1}{2\alpha\log n}\sum_{i \in V}{\dG{}{i} \cdot u_i^2}  - m\cdot z\ .
        \end{align*}
  The last equality holds because $\F{t}$ is $\dG$-stochastic, indeed:
        \begin{align*}
            \sum_{i \in V}{\dG{}{i} \cdot u_i} = \sum_{i \in V}{\dG{}{i} \cdot\left\langle \frac{1}{\dG{}{i}}D^{-\frac{1}{2}}\F{t}{i}, r\right\rangle} = \left\langle D^{-\frac{1}{2}}\sum_{i \in V}{\F{t}{i}}, r\right\rangle = \left\langle D^{-\frac{1}{2}}\dG, r\right\rangle = \left\langle \sqrt{\dG}, r\right\rangle = 0,
        \end{align*}
        since $r\perp \sqrt{\dG}$. 
        Furthermore, since
    $\left\langle \frac{\sqrt{\dG}}{2m}, r\right\rangle = 0$
        we have that
        \begin{align*}
            \frac{1}{2}\sum_{\{i,k\} \in \M{t}}{\sum_{j\in V}{\left(\frac{\F{t}{i,j}}{\dG{}{i}\sqrt{\dG{}{j}}} - \frac{\F{t}{k,j}}{\dG{}{k}\sqrt{\dG{}{j}}}\right)^2}} &\ge \frac{n-1}{2\alpha\log n}\sum_{i \in V}{\dG{}{i} \cdot \left(u_i - \left\langle\frac{\sqrt{\dG}}{2m}, r \right\rangle\right)^2}  - m\cdot z\ .
        \end{align*}
        Therefore, in expectation:
        \begin{align*}
        \EE\left[\frac{1}{2}\sum_{\{i,k\} \in \M{t}}{\sum_{j\in V}{\left(\frac{\F{t}{i,j}}{\dG{}{i}\sqrt{\dG{}{j}}} - \frac{\F{t}{k,j}}{\dG{}{k}\sqrt{\dG{}{j}}}\right)^2}}\right] \ge &\frac{n-1}{2\alpha\log n}\sum_{i\in V}{\dG{}{i} \cdot\EE\left[\left(u_i - \left\langle\frac{\sqrt{\dG}}{2m}, r \right\rangle\right)^2\right]}\\& - m\cdot\EE[z]\ .
        \end{align*}
        By Lemma \ref{lemma:projection_krv}(1),

        \begin{align*}
            &\EE\left[\frac{1}{2}\sum_{\{i,k\} \in \M{t}}{\sum_{j\in V}{\left(\frac{\F{t}{i,j}}{\dG{}{i}\sqrt{\dG{}{j}}} - \frac{\F{t}{k,j}}{\dG{}{k}\sqrt{\dG{}{j}}}\right)^2}}\right] \\ 
            & \;\;\; \ge  \frac{1}{2\alpha\log n}\sum_{i\in V}{\dG{}{i} \cdot \norm{\frac{1}{\dG{}{i}}D^{-\frac{1}{2}}\F{t}{i} - \frac{\sqrt{\dG}}{2m}}_2^2} - m\cdot\EE[z]
            \\
            & \;\;\; = \frac{1}{2\alpha\log n}\sum_{i\in V}{\sum_{j\in V}{\dG{}{i} \cdot \left(\frac{\F{t}{i,j}}{\dG{}{i}\sqrt{\dG{}{j}}} - \frac{\sqrt{\dG{}{j}}}{2m}\right)^2}} - m\cdot\EE[z]
            \\
            & \;\;\; = \frac{1}{2\alpha\log n}\sum_{i\in V}{\sum_{j\in V}{\frac{1}{\dG{}{i}\cdot \dG{}{j}} \cdot \left(\F{t}{i,j} - \frac{\dG{}{i} \dG{}{j}}{2m}\right)^2}} - m\cdot\EE[z]= \frac{1}{2\alpha\log n}\psi(t) - m\cdot\EE[z]
        \end{align*}
        Since $r$ is a unit vector and $\F{t}$ is $\dG$-stochastic, we get by Cauchy Schwartz inequality
        \[
        u_i = \left\langle \frac{1}{\dG{}{i}}D^{-\frac{1}{2}}\F{t}{i}, r\right\rangle \le
        \norm{\frac{1}{\dG{}{i}}D^{-\frac{1}{2}}\F{t}{i}}_2 \le
        \norm{\frac{1}{\dG{}{i}}\F{t}{i}}_2 \le 
        \norm{\frac{1}{\dG{}{i}}\F{t}{i}}_1
        =1.
        \]
        Therefore $z \le \frac{n}{\alpha \log n} \le n$ and non-zero with probability at most $\frac{1}{n^{\alpha/8}}$ (Lemma~\ref{lemma:projection_krv}). Thus $m \EE[z] \le \frac{n^3}{n^{\alpha/8}} =  n^{-(\alpha-24)/8} \le n^{-\alpha/16}$ ($\alpha > 48$),  completing the proof. 

    \end{proof}
\begin{corollary}
    \label{cor:krv_potential_reduction}
        For each round $t$, $\EE[\psi(t+1)]\le\left(1-\frac{1}{2\alpha\log n}\right)\psi(t) + \frac{1}{n^{\alpha/16}}$, for every $\alpha > 48$, where the expectation is over the unit vector $r\in \RR^n$  and conditioned on $\psi(t)$.

    \end{corollary}
    \begin{proof}
        By Lemmas \ref{lemma:krv_potential_step_1} and \ref{lemma:krv_potential_step_2}, 
        \begin{align*}
            \psi(t)-\EE[\psi(t+1)\mid \psi(t)]&\ge\EE\left[\frac{1}{2}\sum_{\{i,k\} \in \M{t}}{\sum_{j\in V}{\left(\frac{\F{t}{i,j}}{\dG{}{i}\sqrt{\dG{}{j}}} - \frac{\F{t}{k,j}}{\dG{}{k}\sqrt{\dG{}{j}}}\right)^2}}\right]
            \\
            &=\frac{1}{2\alpha\log n}\psi(t) - \frac{1}{n^{\alpha/16}}
        \end{align*}
        Rearranging the above inequality, we get $\EE[\psi(t+1)] \le \left(1-\frac{1}{2\alpha\log n}\right)\psi(t) + \frac{1}{n^{\alpha/16}}$.
    \end{proof}

    \begin{corollary} [Decrease in Potential]
    \label{cor:krv_total_potential}
        With high probability over the choices of $r$, 
        \begin{align*}
            \psi(T) \le \frac{1}{n^{5}}
        \end{align*}
    \end{corollary} 
    \begin{proof}
        By induction on the steps, using the law of total expectation and Corollary \ref{cor:krv_potential_reduction}, we get that: 
        \begin{align*}
            \EE[\psi(T)] \le \left(1-\frac{1}{2\alpha\log n}\right)^T \psi(0) + \frac{T}{n^{\alpha/16}} \le \left(1-\frac{1}{2\alpha\log n}\right)^T\cdot n + \frac{T}{n^{\alpha/16}},
        \end{align*}

        where the second inequality holds due to Equation (\ref{eq:krv_psi_0}).
        We can therefore choose $\alpha > 48$ 
        and $T = \Theta(\log^2 n)$ such that $\EE[\psi(T)] \le \frac{1}{n^{\beta}}$ (for any constant $\beta$).
        An application of Markov's inequality gives the result:
        Indeed, say we want $\psi(T)\le\frac{1}{n^\ell}$
        to hold with probability $1-\frac{1}{n^p}$ for some constants
        $\ell$ and $p$. 
        Then, we choose $\beta = \ell + p$ in the equation above, such that we have $\EE[\psi(T)]\le\frac{1}{n^{\ell + p}}$. Then, by Markov's inequality, because the potential is non-negative, we get that
        \begin{align*}
            \PP[\psi(T)\ge n^{-\ell}] \le \frac{\EE[\psi(T)]}{n^{-\ell}} \le \frac{n^{-\ell -p}}{n^{-\ell}} = \frac{1}{n^p}
        \end{align*}
        The result follows by setting $\ell = 5$.
    \end{proof}

    \begin{lemma}
    \label{lemma:krv_F_expander}
        If $\psi(T)\le\frac{1}{16m^2}$, then $\F{T}$ is a $\frac{1}{4}$-expander.

    \end{lemma}
    \begin{proof}
        Let $S\subseteq V, \bar{S}=V\setminus S$ be a cut, and assume that $\vol_{\F{T}}(S)\le \vol_{\F{T}}(\bar{S})$. By Lemma \ref{lemma:krv_F_t_d_stochastic}, $\F{T}$ is $\dG$-stochastic. Therefore, $\vol_{\F{T}}(S) = \vol_G(S)$ and  $\vol_{\F{T}}(\bar{S}) = \vol_G(\bar{S})$. 
        Because $\psi(T) = \sum_{i\in V}{\sum_{j\in V}{\frac{1}{\dG{}{i}\cdot \dG{}{j}}\left(\F{T}{i,j}-\frac{\dG{}{i} \dG{}{j}}{2m}\right)^2}} \le \frac{1}{16m^2}$, in particular for all $i,j\in V$,
        $$\frac{1}{\dG{}{i}\cdot \dG{}{j}}\left(\F{T}{i,j}-\frac{\dG{}{i} \dG{}{j}}{2m}\right)^2\le\frac{1}{16m^2}.$$

        Note that $\F{T}{i,j} \ge \frac{\dG{}{i} \dG{}{j}}{4m}$ since otherwise,
        $$\frac{1}{\dG{}{i}\cdot \dG{}{j}}\left(\F{T}{i,j}-\frac{\dG{}{i} \dG{}{j}}{2m}\right)^2
        >
        \frac{1}{\dG{}{i}\cdot \dG{}{j}}\left(\frac{\dG{}{i} \dG{}{j}}{4m}\right)^2
        \ge
        \frac{1}{16m^2},$$

        a contradiction. We get that
        \begin{align*}
            |E_{\F{T}}(S, \bar{S})| &\ge \sum_{i\in S}{\sum_{j\in\bar{S}}{\F{T}{i,j}}} \ge \sum_{i\in S}{\sum_{j\in\bar{S}}{\frac{\dG{}{i} \dG{}{j}}{4m}}} = \frac{\vol_G(S)\vol_G(\bar{S})}{4m} \\&= \frac{\vol_{\F{T}}(S)\vol_{\F{T}}(\bar{S})}{4m} \ge \frac{\vol_{\F{T}}(S)}{4} \ ,
        \end{align*}
        where the first inequality is because we only sum edges of $\F{T}$ going out of $S$, and do not sum incoming edges. 
        The last inequality is because $\vol_{\F{T}}(S)\le\vol_{\F{T}}(\bar{S})\Rightarrow\vol_{\F{T}}(\bar{S})\ge m$.
        In total we get that
        \begin{align*}
            \Phi_{\F{T}}(S,\bar{S}) = \frac{|E_{\F{T}}(S, \bar{S})|}{\vol_{\F{T}}(S)} \ge \frac{1}{4}\ .
        \end{align*}
    \end{proof}

    We are now ready to show that after $T$ rounds $\Phi(G) = \Omega(\phi)$, so that in this case the algorithm can certify that case (1) of Theorem \ref{theorem:krv} holds.
    \begin{corollary}
    \label{cor:krv_G_expander}
        If we reach round $T$, then with high probability, $G$ is a $\Omega(\phi)$-expander.
    \end{corollary}
    \begin{proof}
        Assume we reach round $T$. By Corollary \ref{cor:krv_total_potential}, with high probability, $\psi(T)\le \frac{1}{n^5} \le \frac{1}{16m^2}$. 
        This implies by Lemma \ref{lemma:krv_F_expander}, that with high probability $\F{T}$ is a $\frac{1}{4}$-expander. By Lemma \ref{lemma:krv_F_t_embeddable_in_G_t}, $\F{T}$ is embeddable in $G_T$ with congestion $1$. Note that $G_T$ is a union of $T$ $\dG$-matchings $\{\M{t}\}_{t=1}^T$, each having $\dG{\M{t}} = \dGG = \dG{\F{T}}$. Therefore, $\dG{G_T} = T\cdot \dG{\F{T}}$. So by Corollary \ref{cor:expansion_and_embedding}, $G_T$ is a $\frac{1}{4T}$-expander. By Lemma \ref{lemma:krv_G_t_embeddable_in_G}, $G_T$ is embeddable in $G$ with congestion $cT$. Together with  the fact that $\dGG = \frac{1}{T} \cdot \dG{G_T}$, we get by Corollary \ref{cor:expansion_and_embedding} again, that $G$ is a $\frac{1}{4cT}$-expander. Recall that $c=O\left(\frac{1}{\phi\log^2 n}\right)$ and $T=O(\log^2 n)$. Therefore, $G$ is an $\Omega(\phi)$-expander.
    \end{proof}

    \subsection{Proof of Theorem \ref{theorem:krv}}
    \label{section:krv_theorem_proof}

    We are now ready to prove Theorem~\ref{theorem:krv}.
\begin{proof} [Proof of Theorem \ref{theorem:krv}]
    If one of the iterations of Algorithm \ref{algo:krv_cut_matching} returns a cut $(S, \bar{S})$, then by Lemma \ref{lemma:krv_3.7_small_flow} we have $\Phi_G(S, \bar{S}) < \frac{1}{c} = O(\phi \log^2 n)$ and we end in Case (2) of Theorem \ref{theorem:krv}. If the algorithm completes $T$ iterations, then by Corollary \ref{cor:krv_G_expander} we have $\Phi(G) = \Omega(\phi)$ with high probability, so we end in Case (1) of Theorem \ref{theorem:krv}.

    Finally, the running time of each iteration is the sum of the running times of the following steps:
    \begin{enumerate}
        \item Sample a random unit vector $r\in \RR^n$ orthogonal to $\sqrt{\dG}$ in $O(n)$ time.
        \item Compute the projections vector $u = D^{-1} \F{t} \Dminushalf \cdot r$. This takes $O(t\cdot m)$ time as $\F{t} = \left(\frac{I + \M{t-1}\cdot D^{-1}}{2}\right)\cdot \left(\frac{I + \M{t-2}\cdot D^{-1}}{2}\right)\cdot\cdots\cdot\left(\frac{I + \M{0}\cdot D^{-1}}{2}\right)\cdot D$, and each $\dG$-matching $\M{i}$ has $O(m)$ non-zero entries.
        \item Compute the max flow on the auxiliary flow problem in time $T_\text{flow}(n, m)$, where $T_\text{flow}(n,m)$ is the time for computing exact max flow on a graph with $n$  vertices and $m$ edges.
        \item Decompose the flow into paths. This can be done in $O(m\log n)$ time using dynamic trees \cite{ST83}.
    \end{enumerate}
    In total, iteration $t$ (for $t\in \{1,\ldots,T\}$) takes $O\left(m\cdot t + T_\text{flow}(n, m)\right)$ time. Recall that $T=\Theta(\log^2 n)$, so the total running time is $O\left(m\log^4 n + \log^2 n \cdot T_\text{flow}(n, m)\right)$.

    \end{proof}

    \section{OSVV's Cut-Matching Game for Conductance}
    \label{section:osvv}

    In this section we modify the framework by \cite{orecchia2008partitioning} in order to bound the expansion of the graph. Our description here expands the overview in subsection \ref{section:osvv-short}. We prove the following theorem.
    \begin{theorem}[Theorem \ref{theorem:osvv}]
        Given a graph $G$ and a parameter $\phi > 0$, there exists a randomized algorithm, whose running time is dominated by computing a polylogarithmic number of max flow problems, that either
        \begin{enumerate}
            \item Certifies that $\Phi(G)= \Omega(\phi)$ with high probability; or
            \item Finds a cut $(S, \bar{S})$ in $G$ whose conductance is $\Phi_G(S,\bar{S}) = O(\phi\log n)$.
        \end{enumerate}    

    \end{theorem}

    The algorithm for Theorem \ref{theorem:osvv} is the same as Algorithm \ref{algo:krv_cut_matching}, except that we use a stronger cut player in order to update $\F{t}$, which is described in Subsection \ref{section:osvv-short}. The matching player remains the same except for a minor difference: the capacities in the flow problem are now $c = \Theta(\frac{1}{\phi \log n})$. 
    The new update of the flow matrix $F$ is summarized in Algorithm \ref{algo:osvv_update}.

    \begin{algorithm}[hbt!]
        \caption{OSVV Round Update \cite{orecchia2008partitioning}. See Section \ref{section:osvv-short} for details.}
        \label{algo:osvv_update}
        \begin{algorithmic}[1]
            \Function{OSVV-Update-F}{$G, \phi, F$}
                \State $c\gets\Theta\left(\frac{1}{\phi \log n}\right)$. \Comment{$c$ is an integer.}
                \State $P \gets I -  \frac{1}{2m}\sqrt{\dG} \sqrt{\dG'}$.
                \State $\delta \gets \Theta(\log n)$. \Comment{$\delta$ is a power of $2$.}
                \smallskip
                \State \underline{\textbf{Cut Player:}}
		\smallskip
                \State $\W \gets (P\normalized{\F}P)^{\delta}$
                \State $r\gets$ Random unit vector of $\RR^{n}$.
                \State $u \gets \Dminushalf \W \cdot r$. 

                \State Sort the entries of $u$ as $u_{i_1} \le \cdots \le u_{i_n}$.
                \State $Q \gets (i_1, i_1, \ldots, i_1, i_2, i_2, \ldots, i_2, \ldots, i_n, \ldots, i_n)$. \Comment{Each $i_j$ appears $\dGG{i_j}$ times.}
                \State $L \gets (Q_1, \ldots, Q_m)$, $R \gets (Q_{m+1}, \ldots, Q_{2m})$.
                \smallskip
                \State \underline{\textbf{Matching Player:}}
                \smallskip
                \State $G' \gets (V', E')$, where $V' \gets V\cup \{s, t\}$, $E' \gets E \cup (\{s\} \times L) \cup (R \times \{t\})$.
                \State Let $m_v$, $\bar{m_v}$ be the number of times $v$ appears in $L$, $R$, respectively.
                \State Set the capacity of $e\in E$ to $c$, of $(s, v) \in \{s\}\times L$ to $m_v$ and of $(v, t)\in R\times \{t\}$ to $\bar{m_v}$.

                \State Compute a maximum flow $g$ from $s$ to $t$ in $G'$.
                \If{$|g| < m$} \Comment{$|g|$ is the value of the flow $g$.}
                    \State Find a minimum $(s,t)$-cut $(S,\bar{S})$ in $G'$.
                    \State \Return{$(S\cap G, \bar{S}\cap G)$}.
                \Else
                    \State Decompose $g$ into a set of paths $\{u_j \to v_j\}_j$, where $u_j \in L$ and $v_j \in R$.
                    \State $M \gets \{\{u_j, v_j\}\}_{j=1}^m$. 

                    \\\Comment{$M$ is a symmetric $n\times n$ matrix, $M(u,v)$ is the number of paths between $u$ and $v$.}
                    \smallskip
                    \State \underline{\textbf{Update of $F$:}}
		    \smallskip
                    \State $\N \gets \frac{\delta - 1}{\delta}D + \frac{1}{\delta}\M$.
                    \State $F_{\text{new}} \gets \N D^{-1} \F D^{-1} \N$.
		    \State \Return{$F_{\text{new}}$}.
                \EndIf
            \EndFunction
        \end{algorithmic}
    \end{algorithm}

    The rest of this section is organized as follows. Subsection \ref{section:osvv_F_embeddable_G} shows that $\F{t}$ is embeddable in $G_t$ (the graph induced by the union of all the matchings that we have computed so far) with congestion $\frac{4}{\delta}$, where $\delta = \Theta(\log n)$ is a power of $2$ and is set in the proof of Lemma~\ref{lemma:osvv_F_expander}, and that $G_t$ is embeddable in $G$ with congestion $c\cdot t$. Subsection \ref{section:osvv_F_t_expander} shows that if we reach round $T$, then with high probability, $\Phi(G)=\Omega(\phi)$. Finally, in Subsection \ref{section:osvv_theorem_proof} we prove Theorem \ref{theorem:osvv}.

\subsection{$\F{t}$ is embeddable in $G$}
    \label{section:osvv_F_embeddable_G}

    The following lemma states basic properties regarding the matrices $M_t,F_t,N_t$ and $W_t$ (recall their definition from the cut player strategy presented in Subsection \ref{section:osvv-short}).
\begin{lemma}
    \label{lemma:osvv_basic_properties}
        The following holds for all $t$:
        \begin{enumerate}
            \item $\M{t}, \N{t}, \F{t}$ and $\W{t}$ are symmetric.
            \item $\M{t}, \N{t}$ and $\F{t}$ are $\dG$-stochastic.
        \end{enumerate}
    \end{lemma}
\begin{proof}
        \begin{enumerate}
            \item This is clear from the definitions.
            \item For $\M{t}$ this is clear. For $\N{t}$, note that
            \begin{align*}
                \N{t} \1_n = \left(\frac{\delta - 1}{\delta}D + \frac{1}{\delta}\M{t}\right) \1_n = \frac{\delta - 1}{\delta}\dG + \frac{1}{\delta}\dG = \dG \ ,
            \end{align*}
            and $\1_n' \N{t} = \dG'$ follows since $N_t$ is symmetric.

            For $\F{t}$, we use induction  on $t$. $\F{0} = D = \diag(\dG)$ is clearly $\dG$-stochastic. After step $t$, 
            \begin{align*}
                \F{{t+1}}\1_n &= \N{t} D^{-1} \F{t} D^{-1} \N{t} \cdot\1_n = \N{t} D^{-1} \F{t} D^{-1} \cdot \dG \\ &= \N{t} D^{-1} \F{t}\cdot \1_n = \N{t} D^{-1} \cdot \dG = \N{t} \cdot \1_n = \dG \ ,
            \end{align*}
            and $\1_n' \F{{t+1}} = \dG'$ follows since $F_t$ is symmetric.
        \end{enumerate}
    \end{proof}

The proofs of the following two lemmas are 
the same as the proofs of 
Corollary~\ref{lemma:our_F_t_embeddable_in_G_t} and Lemma~\ref{lemma:our_G_t_embeddable_in_G} in our analysis of the spectral cut player in Section~\ref{section:our_F_embeddable_G}. 

\begin{lemma}
\label{cor:osvv_F_t_embeddable_in_G_t}
    For all rounds $t$, $\F{t}$ is embeddable in $G_t$ with congestion $\frac{4}{\delta}$.
\end{lemma}

\begin{lemma}
        \label{lemma:osvv_G_t_embeddable_in_G}
            For all rounds $t$, $G_t$ is embeddable in $G$ with congestion $ct$.
        \end{lemma}

    \subsection{$\F{T}$ is an expander}
    \label{section:osvv_F_t_expander}
    In this section we  prove that after $T=\Theta(\log^2 n)$ rounds, with high probability, $\Phi(\F{T}) = \Omega(1)$, which implies that $\Phi(G) = \Omega(\phi)$. Define the following potential function, analogous to KRV:

\begin{align*}
    \psi(t) &= \norm{(\normalized{\F{t}})^\delta - \frac{1}{2m}\sqrt{\dG} \sqrt{\dG'}}^2_F.
    \end{align*}

We use the second identity in the following lemma to rewrite the potential. Let $P = I - \frac{1}{2m}\sqrt{\dG}\sqrt{\dG'}$ be the projection matrix on the subspace orthogonal to $\sqrt{\dG}$.
\begin{lemma}
    \label{lemma:normalized_commutes_with_P}

    Let $A$ be a $\dG$-stochastic matrix, then: 

    \begin{enumerate}
        \item $\normalized{A} \sqrt{\dG} = \sqrt{\dG}$, and $\sqrt{\dG'}\normalized{A}  = \sqrt{\dG'}$.
        \item $\normalized{A} P = P \normalized{A}$.
    \end{enumerate}
    \end{lemma}
\begin{proof}
        \begin{enumerate}
            \item Indeed,
        \[
        \normalized{A} \sqrt{\dG} = \Dminushalf A \1_n = \Dminushalf \dG = \sqrt{\dG},
        \]
        \[
        \sqrt{\dG'} \normalized{A}  = \1_n'  A \Dminushalf =  \dG' \Dminushalf = \sqrt{\dG'}
        \]

            \item This follows from the previous property:
            \begin{align*}
                \normalized{A} P &= \normalized{A}\cdot (I - \frac{1}{2m}\sqrt{\dG} \sqrt{\dG'}) = \normalized{A} - \frac{1}{2m}\cdot \normalized{A}\sqrt{\dG}\sqrt{\dG'} 
                \\
                &= \normalized{A} - \frac{1}{2m}\sqrt{\dG}\sqrt{\dG'} = \normalized{A} - \frac{1}{2m}\sqrt{\dG}\sqrt{\dG'}\normalized{A}
                \\
                &=  (I - \frac{1}{2m}\sqrt{\dG} \sqrt{\dG'})\cdot\normalized{A} = P \normalized{A}
            \end{align*}

        \end{enumerate}
    \end{proof}

    Recall that $\W{t} = (P\normalized{\F{t}}P)^{\delta}$. Since $\F{t}$ is $\dG$-stochastic (Lemma \ref{lemma:osvv_basic_properties}(1)) and  $P^2 = P$, it follows by Lemma~\ref{lemma:normalized_commutes_with_P} that  
    $$\W{t} = P(\normalized{\F{t}})^{\delta}= (\normalized{\F{t}})^{\delta} P = 
    P(\normalized{\F{t}})^{\delta}P.$$ 
    Therefore, we can rewrite the potential function as
    \begin{align*}
        \psi(t) &= \norm{(\normalized{\F{t}})^\delta - \frac{1}{2m}\sqrt{\dG} \sqrt{\dG'}}^2_F
        = \norm{(\normalized{\F{t}})^\delta P}^2_F \\ &=
        \tr(P (\normalized{\F{t}})^{2\delta} P) \numeq{1} \tr( (P\normalized{\F{t}}P)^{2\delta}) = \tr(W^2_t),
    \end{align*}
    where Equation~$(1)$ follows from Lemma~\ref{lemma:normalized_commutes_with_P}$(2)$. In particular,

    \begin{align*}
        \psi(0) = \norm{I - \frac{1}{2m}\sqrt{\dG} \sqrt{\dG'}}^2_F
        = \norm{P}^2_F = \tr(P^2)=\tr(P)=n-1.
    \end{align*}

    Additionally, denote $\X{t} = \normalized{\N{t}}, \Y{t} = \normalized{\F{t}}$.  
    Recall the definition of a normalized Laplacian: $\mathcal{N}(A) \defeq \normalized{\mathcal{L}(A)} = I - \normalized{A}$, where $A$ is a symmetric $\dG$-stochastic matrix. 

    The following lemma shows the relation between $\X{t}$ and $\mathcal{N}(\M{t})$. 

\begin{lemma}
    \label{lemma:osvv_X_as_normalized_laplacian}
        The following relations hold:
        \begin{enumerate}
            \item $\X{t} = I - \frac{1}{\delta} \mathcal{N}(\M{t})$.
            \item $I-\X{t}^{4\delta} \succeq \frac{1}{3} \cdot \mathcal{N}(\M{t})$. That is, for any vector $v\in \RR^n$, $v'(I-\X{t}^{4\delta})v \ge v'(\frac{1}{3} \cdot \mathcal{N}(\M{t}))v$.
            \item For any matrix $A\in \RR^{n\times n}$, $\tr(A'(I-\X{t}^{4\delta})A)\ge \frac{1}{3}\tr(A' \mathcal{N}(\M{t})A)$.
        \end{enumerate}
    \end{lemma}
\begin{proof}
        \begin{enumerate}
            \item Indeed, 
            \[
            \X{t} = \normalized{\left(\frac{\delta-1}{\delta}D + \frac{1}{\delta}\M{t}\right)} = 
            \frac{\delta-1}{\delta}I + \frac{1}{\delta}\normalized{\M{t}} = 
            I - \frac{1}{\delta}(I - \normalized{\M{t}}).
            \]
            \item Observe that $\mathcal{N}(\M{t})$ and $I-\X{t}^{4\delta}$ have the same eigenvectors: consider an eigenvector $v\in \RR^n$ of $\mathcal{N}(\M{t})$, with an eigenvalue of $\lambda$. Then, by part (1) 
            \begin{align*}
                \X{t} v = \left(I-\frac{1}{\delta}\mathcal{N}(\M{t})\right)v = v-\frac{\lambda}{\delta}v = \left(1-\frac{\lambda}{\delta}\right)v \ .
            \end{align*}

            Therefore, $(I-\X{t}^{4\delta}) v = \left(1-\left(1-\frac{\lambda}{\delta}\right)^{4\delta}\right)v$. So we
            can see that an eigenvector $v\in \RR^n$ of $\mathcal{N}(\M{t})$, with eigenvalue $\lambda$, is an eigenvector of $I-\X{t}^{4\delta}$ with eigenvalue $1-(1-\frac{\lambda}{\delta})^{4\delta}$. A known property of normalized Laplacians (see Lemma~\ref{lemma:normalized-laplacian-eigenvalues}) 
            is that all of their eigenvalues are in the interval $[0,2]$, so $\lambda \in [0,2]$. For any such $\lambda$ it holds that $1-(1-\frac{\lambda}{\delta})^{4\delta}\ge1-\frac{1}{e^{4\lambda}}$. 
            Indeed, for $\lambda=0$ an equality is achieved, and for $\lambda\in (0,2]$, note that $\delta = \Theta(\log n)$, so $\frac{\delta}{\lambda} \ge 1$ and we get
            \[
            1-\left(1-\frac{\lambda}{\delta}\right)^{4\delta} = 1-\left(\left(1-\frac{\lambda}{\delta}\right)^\frac{\delta}{\lambda}\right)^{4\lambda} \ge
            1- \frac{1}{e^{4\lambda},}
            \]
            where the last inequality holds since $(1-1/x)^x < 1/e$  for all $x\ge 1$.
            Simple calculus shows that $1-\frac{1}{e^{4x}} \ge \frac{1}{3}x$ for any $x\in [0,2]$. 

            Therefore, we get that any eigenvector $v\in \RR^n$ of $\mathcal{N}(\M{t})$, with eigenvalue $\lambda$, is also an eigenvector of $I-\X{t}^{4\delta}$ with eigenvalue at least $\frac{1}{3}\lambda$. The result follows since both matrices are symmetric and therefore have a spanning basis of eigenvectors.

            \item This follows immediately from (2) since  $I-\X{t}^{4\delta} - \frac{1}{3} \cdot \mathcal{N}(\M{t}) \succeq 0$ is PSD. Indeed, for any matrix $A$ and any PSD matrix $B$ it holds that $\tr(A'BA) = \sum_{i=0}^n{A(,i)' B A(,i)}\ge 0$.

        \end{enumerate}
    \end{proof}

    We now bound the 
    potential decrease in round $t$.
    \begin{lemma}
    \label{lemma:osvv_potential_step_1}
        For each round $t$, $\psi(t) - \psi(t+1) \ge \frac{1}{3}\sum_{\{i,k\} \in \M{t}}{\sum_{j\in V}{\left(\frac{\W{}{i,j}}{\sqrt{\dG{}{i}}} - \frac{\W{}{k,j}}{\sqrt{\dG{}{k}}}\right)^2}}$.
    \end{lemma}
    \begin{proof}
        Recall that $\X{t} = \normalized{\N{t}}, \Y{t} = \normalized{\F{t}}$. We have
        \begin{align*}
            \psi(t+1) &= \tr(P (\normalized{\F{{t+1}}})^{2\delta} P) = 
            \tr(P (\normalized{(\N{t} D^{-1} \F{t} D^{-1} \N{t})})^{2\delta} P) \\ &=
            \tr(P ((\normalized{\N{t}}) (\normalized{\F{t}}) (\normalized{\N{t}}))^{2\delta} P)
            \\
            & = \tr(P ( \X{t} \Y{t} \X{t})^{2\delta} P) \numeq{5} \tr(  (\X{t} (P \Y{t} P) \X{t})^{2\delta}) \\ &\le
            \tr ( \X{t}^{2\delta} (P \Y{t} P)^{2\delta} \X{t}^{2\delta}) = \tr(\X{t}^{2\delta} W^2_t \X{t}^{2\delta}) \\&=
            \tr(\X{t}^{4\delta} W^2_t) \ .
        \end{align*}

     Equality  $(5)$ follows from the fact that $P^2 = P$, Lemma \ref{lemma:normalized_commutes_with_P}(2) and from the fact that $\N{t}$ and $\F{t}$ are $\dG$-stochastic matrices (Lemma \ref{lemma:osvv_basic_properties}). 
     The inequality follows from Theorem \ref{theorem:symmetric_rearrangement} in Appendix \ref{appendix:matrix_inequalities}. The last equality follows from Fact \ref{fact:book_trace_identities}(1) in Appendix \ref{appendix:matrix_inequalities}.
        We are now ready to bound the decrease in potential: 
        \begin{align*}
            \psi(t)- \psi(t+1) &\ge \tr((I-\X{t}^{4\delta}) W^2_t) \numeq{1} \tr(\W{t} (I-\X{t}^{4\delta}) \W{t}) 
            \\
            &\ge \frac{1}{3} \tr(\W{t} \mathcal{N}(\M{t}) \W{t}) = \frac{1}{3} \tr(\W{t} \normalized{\mathcal{L}(\M{t})} \W{t}) 
            \\
            &\numeq{3} \frac{1}{3} \tr((\Dminushalf \W{t})' \mathcal{L}(\M{t}) (\Dminushalf \W{t})) 
            \\
            &\numeq{4} \frac{1}{3} \sum_{\{i,k\}\in \M{t}} \norm{\Matrix{\Dminushalf \W{t}}{i} - \Matrix{\Dminushalf \W{t}}{k}}^2_2 
            \\
            &= \frac{1}{3} \sum_{\{i,k\}\in \M{t}} \sum_{j\in V}{\left(\frac{\W{}{i,j}}{\sqrt{\dG{}{i}}} - \frac{\W{}{k,j}}{\sqrt{\dG{}{k}}}\right)^2}~,
        \end{align*}

        where Equality $(1)$ follows from Fact \ref{fact:book_trace_identities}(1) in Appendix \ref{appendix:matrix_inequalities}. The first inequality follows from Lemma \ref{lemma:osvv_X_as_normalized_laplacian}(3). Equality $(3)$ is true because $\W{t}$ and $\Dminushalf$ are symmetric. Equality $(4)$ follows from lemma \ref{lemma:matching_laplacian} in Appendix \ref{appendix:matrix_inequalities}.
    \end{proof}

    The proof of the following lemma is very similar to Lemma \ref{lemma:krv_potential_step_2}, except for the slightly simpler setting and the use of Lemma \ref{lemma:projection_pairs} instead of Lemma \ref{lemma:projection_krv} from Appendix \ref{appendix:projection}. For completeness, we give the full proof.
    \begin{lemma}
    \label{lemma:osvv_potential_step_2}
        For each round $t$, $\EE\left[\frac{1}{3}\sum_{\{i,k\} \in \M{t}}{\sum_{j\in V}{\left(\frac{\W{}{i,j}}{\sqrt{\dG{}{i}}} - \frac{\W{}{k,j}}{\sqrt{\dG{}{k}}}\right)^2}}\right]
        \ge
        \frac{1}{3\alpha \log n}\psi(t) - \frac{1}{n^{\alpha/16}}$, for every $\alpha > 48$, where the expectation is over the unit vector $r\in \RR^n$ and conditioned on $\psi(t)$.
    \end{lemma}
\begin{proof}
        Recall that $u_i = \frac{1}{\sqrt{\dG{}{i}}} \langle \W{t}{i}, r \rangle$. Notice that $\Matrix{\frac{1}{\sqrt{\dG{}{i}}} \W{t}{i}}{j} = \frac{\W{t}{i,j}}{\sqrt{\dG{}{i}}}$. 
        Apply Lemma \ref{lemma:projection_pairs} from Appendix \ref{appendix:projection} to the set of vectors $\left\{\frac{1}{\sqrt{\dG{}{i}}}\W{t}{i} \mid i\in V\right\}$. 
        By Lemma \ref{lemma:projection_pairs}(2), we have with high probability:
        \begin{align*}
            \forall i, k\in V: \sum_{j\in V}{\left(\frac{\W{t}{i,j}}{\sqrt{\dG{}{i}}} - \frac{\W{t}{k,j}}{\sqrt{\dG{}{k}}}\right)^2} = \norm{\frac{1}{\sqrt{\dG{}{i}}}\W{t}{i} - \frac{1}{\sqrt{\dG{}{k}}}\W{t}{k}}_2^2 \ge \frac{n}{\alpha\log n}\cdot(u_i-u_k)^2
        \end{align*}
        for every constant $\alpha \ge 16$.
        Like in Section \ref{section:krv_F_t_expander}, we introduce a random variable $z$ that is non-zero only when this inequality fails to hold, such that  
        \begin{align*}
		        \forall i, k\in V: \norm{\frac{1}{\sqrt{\dG{}{i}}}\W{t}{i} - \frac{1}{\sqrt{\dG{}{k}}}\W{t}{k}}_2^2 \ge \frac{n}{\alpha \log n}\cdot(u_i - u_k)^2 - z
        \end{align*}
        holds with probability $1$. \emph{I.e.}, we define
        \begin{align*}
            \mathcal{B} = \{0\} &\cup \left\{\frac{n}{\alpha\log n}(u_i - u_k)^2-\norm{\frac{1}{\sqrt{\dG{}{i}}}\W{t}{i} - \frac{1}{\sqrt{\dG{}{k}}}\W{t}{k}}_2^2 : (i, k)\in V\times V\right\}
        \end{align*}
        and $z = \max(\mathcal{B})$. We get that
        \begin{align*}
            \frac{1}{3}\sum_{\{i,k\} \in \M{t}}{\sum_{j\in V}{\left(\frac{\W{t}{i,j}}{\sqrt{\dG{}{i}}} - \frac{\W{t}{k,j}}{\sqrt{\dG{}{k}}}\right)^2}} \ge \frac{n}{3 \cdot \alpha\log n}\sum_{\{i,k\} \in \M{t}}{\left(u_i - u_k\right)^2} - m\cdot z
        \end{align*}
        Any $\{i,k\}\in \M{t}$ satisfies $u_i \le \eta \le u_k$ or $u_i \ge \eta \ge u_k$. Therefore, $(u_i-u_k)^2 \ge (u_i-\eta)^2+(u_k-\eta)^2$. We get:
        \begin{align*}
            &\frac{1}{3}\sum_{\{i,k\} \in \M{t}}{\sum_{j\in V}{\left(\frac{\W{t}{i,j}}{\sqrt{\dG{}{i}}} - \frac{\W{t}{k,j}}{\sqrt{\dG{}{k}}}\right)^2}} \ge \frac{n}{3\alpha\log n}\sum_{i \in V}{\dG{}{i}\left(u_i - \eta\right)^2} - m\cdot z
            \\
            &= \frac{n}{3\alpha\log n}\left(\sum_{i \in V}{\dG{}{i} \cdot u_i^2} -2\eta\sum_{i \in V}{\dG{}{i} \cdot u_i} + 2m\eta^2 \right) - m\cdot z 
            \\
            &= \frac{n}{3\alpha\log n}\left(\sum_{i \in V}{\dG{}{i} \cdot u_i^2} + 2m\eta^2 \right) - m\cdot z \ge \frac{n}{3\alpha\log n}\sum_{i \in V}{\dG{}{i} \cdot u_i^2} - m\cdot z \ .
        \end{align*}
        The last equality holds because $\W{t}\sqrt{\dG} = \left( P \Dminushalf \F{t} \Dminushalf P \right)^\delta \sqrt{\dG} = 0$ (since $P \sqrt{\dG} = 0$), and:
        \begin{align*}
            \sum_{i \in V}{\dG{}{i} \cdot u_i} = \sum_{i \in V}{\dG{}{i} \cdot\left\langle \frac{1}{\sqrt{\dG{}{i}}}\W{t}{i}, r\right\rangle} = \left\langle \left(\sqrt{\dG}'\W{t}\right)', r\right\rangle = \left\langle \W{t}\sqrt{\dG}, r\right\rangle = \left\langle 0, r\right\rangle = 0
        \end{align*}
        Therefore, in expectation:
        \begin{align*}
            \EE\left[\frac{1}{3}\sum_{\{i,k\} \in \M{t}}{\sum_{j\in V}{\left(\frac{\W{t}{i,j}}{\sqrt{\dG{}{i}}} - \frac{\W{t}{k,j}}{\sqrt{\dG{}{k}}}\right)^2}}\right] &\ge \frac{n}{3\alpha\log n}\sum_{i\in V}{\dG{}{i} \cdot\EE\left[u_i^2\right]} - m\cdot \EE[z] \ .
        \end{align*}
        By Lemma \ref{lemma:projection_pairs}(1),
        \begin{align*}
            &\EE\left[\frac{1}{3}\sum_{\{i,k\} \in \M{t}}{\sum_{j\in V}{\left(\frac{\W{t}{i,j}}{\sqrt{\dG{}{i}}} - \frac{\W{t}{k,j}}{\sqrt{\dG{}{k}}}\right)^2}}\right] \ge \frac{1}{3\alpha\log n}\sum_{i\in V}{\dG{}{i} \cdot \norm{\frac{1}{\sqrt{\dG{}{i}}}\W{t}{i}}_2^2} - m\cdot \EE[z] 
            \\
            &= \frac{1}{3\alpha\log n}\sum_{i\in V}{\sum_{j\in V}{\dG{}{i} \cdot \left(\frac{\W{t}{i,j}}{\sqrt{\dG{}{i}}}\right)^2}} - m\cdot \EE[z] = \frac{1}{3\alpha\log n}\sum_{i\in V}{\sum_{j\in V}{\W{t}{i,j}^2}} - m\cdot \EE[z] 
            \\
            &= \frac{1}{3\alpha\log n}\psi(t) - m\cdot \EE[z]
        \end{align*}

        We note that $z \le  \frac{4n}{\alpha \log n}\le n$. Indeed, for every $i \in \A{t}$
        \[
        u_i = \left\langle \frac{\W{t}{i}}{\sqrt{\dG{}{i}}}, r \right\rangle \le
        \norm{\frac{\W{t}{i}}{\sqrt{\dG{}{i}}}}_2 \le 
        \norm{\W{t}{i}}_2 \le 1,
        \]
        where the first inequality holds due to Cauchy-Schwartz since $r$ is a unit vector and the last inequality holds since all eigenvalues of $\W{t}$ are in $[0,1]$:\footnote{The $i$'th row is obtained by multiplying the matrix by a unit vector and since all eigenvalues are in $[0,1]$, $\norm{\W{t}{i}}_2 = \norm{(e_i' \W{t})'}_2 = \norm{\W{t} e_i}_2\le \norm{e_i}_2 = 1$.} 
        By Lemma~\ref{lemma:normalized-laplacian-eigenvalues} we get that $\Dminushalf \F{t} \Dminushalf = I - \mathcal{N}(\F{t})$ has eigenvalues in $[-1,1]$. Since $\Pmat$ is a projection matrix, it has eigenvalues in $[0,1]$. Therefore, $\Pmat \Dminushalf \F{t} \Dminushalf \Pmat$ has eigenvalues in $[-1,1]$. Finally, $\W{t}= (\Pmat \Dminushalf \F{t} \Dminushalf \Pmat)^\delta$ is PSD since $\delta$ is a power of $2$, so its eigenvalues are in $[0,1]$.

        By  Lemma~\ref{lemma:projection_pairs}, $z$ is non-zero with probability at most $\frac{1}{n^{\alpha/8}}$. Because $z \le n$, we get  
        $m \EE[z] \le \frac{mn}{n^{\alpha/8}} \le \frac{1}{n^{(\alpha-24)/8}} \le \frac{1}{n^{\alpha/16}}$, where the last inequality holds since $\alpha > 48$. 
    \end{proof}

    The following two corollaries follow by Lemmas \ref{lemma:osvv_potential_step_1} and \ref{lemma:osvv_potential_step_2}, similarly to Corollaries~\ref{cor:krv_potential_reduction} and \ref{cor:krv_total_potential}.

\begin{corollary}
    \label{cor:osvv_potential_reduction}
        For each round $t$, $\EE[\psi(t+1)]\le\left(1-\frac{1}{3\alpha \log n}\right)\psi(t)+\frac{1}{n^{\alpha/16}}$, where the expectation is over the unit vector $r\in \RR^n$ and conditioned on $\psi(t)$.
    \end{corollary}

\begin{corollary} [Decrease in Potential]
    \label{cor:osvv_total_potential}
        With high probability over the choices of $r$, 
        \begin{align*}
            \psi(T) \le \frac{1}{n}
        \end{align*}
    \end{corollary} 

    Differently from Lemma \ref{lemma:krv_F_expander}, in the following lemma we have to use the spectral properties of $\F{T}$ and Cheeger's inequality (Lemma \ref{lemma:normalized_laplacian_expansion}).
\begin{lemma}
    \label{lemma:osvv_F_expander}
        If $\psi(T)\le\frac{1}{n}$, then $\F{T}$ is a $\Omega(1)$-expander.
    \end{lemma}
    \begin{proof}
        Because $\W{T} = (P\normalized{\F{T}}P)^{\delta}=(\normalized{\F{T}})^{\delta}P$, $\W{T}$ has the same eigenvalues as $(\normalized{\F{T}})^{\delta}$, apart from the one corresponding to the eigenvector $\sqrt{\dG}$, whose eigenvalue in $\W{T}$ is $0$. Because $\psi(T) = \tr(\W{T}^2)\le\frac{1}{n}$, all the eigenvalues of $\W{T}^2$ are at most $\frac{1}{n}$. In particular, this means that the second-largest eigenvalue of $\normalized{\F{T}}$ is at most $\frac{1}{n^{\frac{1}{2\delta}}}$. We pick $\delta = \Theta(\log n)$, such that this is at most $\frac{1}{2}$. This means that the second-smallest eigenvalue of $\mathcal{N}(\F{T}) = I-\normalized{\F{T}}$ is at least $1-\frac{1}{2} = \frac{1}{2}$. It follows by Lemma \ref{lemma:normalized_laplacian_expansion} that $\Phi(\F{T}) = \Omega(1)$.

    \end{proof}

    The following corollary follows from Corollary~\ref{cor:osvv_total_potential} and Lemma \ref{lemma:osvv_F_expander} the same way as Corollary~\ref{cor:krv_G_expander} followed from Corollary~\ref{cor:krv_total_potential} and Lemma~\ref{lemma:krv_F_expander}.

    \begin{corollary}
    \label{cor:osvv_G_expander}
        If we reach round $T$, then with high probability, $G$ is a $\Omega(\phi)$-expander.
    \end{corollary}
    \begin{proof}
        Assume we reach round $T$. By Corollary \ref{cor:osvv_total_potential} and Lemma \ref{lemma:osvv_F_expander}, with high probability, $\F{T}$ is a $\Omega(1)$-expander. By Lemma \ref{cor:osvv_F_t_embeddable_in_G_t}, $\F{T}$ is embeddable in $G_T$ with congestion $O(\frac{1}{\delta})$. Note that $G_T$ is a union of $T$ $\dG$-matchings $\{M_{t}\}_{t=1}^T$, each having $\dG{\M{t}} = \dGG = \dG{\F{T}}$. Therefore, $\dG{G_T} = T\cdot \dG{\F{T}}$. 
        So by Corollary \ref{cor:expansion_and_embedding}, $G_T$ is a $\Omega(\frac{\delta}{T})$-expander. By Lemma \ref{lemma:osvv_G_t_embeddable_in_G}, $G_T$ is embeddable in $G$ with congestion $cT$. Together with  the fact that $\dGG = \frac{1}{T} \cdot \dG{G_T}$, we get by Corollary \ref{cor:expansion_and_embedding} again, that $G$ is a $\Omega(\frac{\delta}{cT})$-expander. Recall that $c=O\left(\frac{1}{\phi\log n}\right)$, $\delta = \Theta(\log n)$, and $T=O(\log^2 n)$. Therefore, $G$ is an $\Omega(\phi)$-expander.
    \end{proof}

    \subsection{Proof of Theorem \ref{theorem:osvv}}
    \label{section:osvv_theorem_proof}

    We are now ready to prove Theorem~\ref{theorem:osvv}. The proof is very similar to the proof of Theorem \ref{theorem:krv} in Section \ref{section:krv_theorem_proof}. We include the full details here for completeness.
\begin{proof} [Proof of Theorem \ref{theorem:osvv}]
    If one of the iterations of Algorithm \ref{algo:krv_cut_matching} returns a cut $(S, \bar{S})$, then by Lemma \ref{lemma:krv_3.7_small_flow} we have $\Phi_G(S, \bar{S}) < \frac{1}{c} = O(\phi \log n)$ and we end in Case (2) of Theorem \ref{theorem:osvv}. If the algorithm completes $T$ iterations, then by Corollary \ref{cor:osvv_G_expander} we have $\Phi(G) = \Omega(\phi)$ with high probability, so we end in Case (1) of Theorem \ref{theorem:osvv}.

    Finally, the running time of each iteration is the sum of the running times of the following steps:
    \begin{enumerate}
        \item Sample a random unit vector $r\in \RR^n$.
        \item Compute the projections vector $u = \Dminushalf \W{t} \cdot r$. This takes $O(t\cdot\delta\cdot m)$ time since $\W{t}$ is a multiplication of $O(t \cdot \delta)$ matrices, where each matrix either has $O(m)$ non-zero entries or is a projection matrix $\Pmat$.
        \item Compute the max flow on the auxiliary flow problem in time $T_\text{flow}(n, m)$, where $T_\text{flow}(n,m)$ is the time for computing exact max flow on a graph with $n$  vertices and $m$ edges.
        \item Decompose the flow into paths. This can be done in $O(m\log n)$ time using dynamic trees \cite{ST83}.
    \end{enumerate}
    In total, iteration $t$ (for $t\in \{1,\ldots,T\}$) takes $O\left(m\cdot\delta\cdot t + T_\text{flow}(n, m)\right)$ time. Recall that $T=\Theta(\log^2 n)$ and $\delta = \Theta(\log n)$, so the total running time is $O\left(m\log^5 n + \log^2 n \cdot T_\text{flow}(n, m)\right)$. 

    \end{proof}

    \section{Proof of Theorem \ref{theorem:our_expander_decomposition}}
	\label{appendix:unit_flow_trimming}
	We use Theorem 2.1 from~\cite{saranurak2019expander}:
	\begin{theorem} [Trimming; Theorem 2.1 of \cite{saranurak2019expander}]
	\label{theorem:SW19_2.1}
	    Given a graph $G=(V,E)$ and $A\subseteq V, \bar{A}=V\setminus A$ such that $A$ is a near $\phi$-expander in $G$, $|E(A,\bar{A})|\le \phi\vol(A)/10$, there exists an algorithm (``the trimming step''), that finds $A'\subseteq A$ in time $O\left(\frac{|E(A,\bar{A})|\log n}{\phi^2}\right)$ such that $\Phi_{G\{A'\}}\ge \phi/6$. Moreover, $\vol(A')\ge \vol(A) - \frac{4}{\phi}|E(A,\bar{A})|$, and $|E(A',\bar{A'})|\le 2|E(A,\bar{A})|$.
	\end{theorem}

	Using this theorem and Theorem \ref{theorem:our_cut_matching}, we can prove Theorem \ref{theorem:our_expander_decomposition} in the same way~SW proved 
 Theorem \ref{theorem:SW-main} (which is
 Theorem 1.2 in their paper). Recall their Expander Decomposition algorithm [Algorithm 1 of~\cite{saranurak2019expander}] given as Algorithm \ref{algo:expander_decomposition} here.
\begin{algorithm}[hbt!]
        \caption{Expander Decomposition~\cite{saranurak2019expander}}
        \label{algo:expander_decomposition}
        \begin{algorithmic}[1]
            \Function{Decomp}{$G$, $\phi$}
                \State Call Cut-Matching($G$, $\phi$) \Comment{See Algorithm \ref{algo:cut_matching}}
                \If {we certify that $\Phi_G\ge \phi$}  
                    \Return $G$
                \ElsIf {we find a relatively balanced cut $(A,R)$}
                    \State \Return Decomp($G\{A\}$, $\phi$) $\cup$ Decomp($G\{R\}$, $\phi$)
                \Else {~we find a very unbalanced cut $(A,R)$}
                    \State $A' \gets $Trimming($G$, $A$, $\phi$) \Comment{Algorithm from Theorem \ref{theorem:SW19_2.1}.}
                    \State \Return $A' \cup$ Decomp($G\{V\setminus A'\}$, $\phi$)
                \EndIf
            \EndFunction
        \end{algorithmic}
    \end{algorithm}

	\begin{proof} [Proof of Theorem \ref{theorem:our_expander_decomposition}]

	    First, note that the conditions of Theorem \ref{theorem:SW19_2.1} hold when we call Trimming($G$, $A$, $\phi$). Indeed, by Case (3) of Theorem \ref{theorem:our_cut_matching}, we have $\vol(\bar{A})\le \frac{m}{10c_0\log n}$ and $\Phi_G(A, \bar{A})\le c_0\phi\log n$. This means that $|E(A,\bar{A})|\le \phi m/10$. 
	    Therefore, the trimming step in $O\left(\frac{m \log n}{\phi}\right)$ time gives a $\phi/6$ expander $G\{A'\}$ where $\vol(A') \ge \Omega(m)$ (the total volume is $2m$, and the volume of $\bar{A}$ is $O\left(\frac{m}{\log n}\right)$).

	    Theorem \ref{theorem:SW19_2.1} also indicates that the conductance of $(A', V \setminus A')$ is at most twice the conductance of $(A, V \setminus A)$, so $(A', V \setminus A')$ has conductance $O(\phi \log n)$.
	    Since Algorithm \ref{algo:expander_decomposition} only stops working on a component when it certifies that the induced sub-graph has conductance at least $\phi/6$, the leaves of the recursion tree give an expander decomposition with expansion $\Omega(\phi)$.

	    Note that whenever we run on a smaller sub-graph, we always add self-loops such that the degree of a node remains the same as its degree in the original graph, so we get the stronger expansion guarantee with respect to the volume in the original graph.
	    To bound the running time, note that if we get Case (2) of Theorem \ref{theorem:our_cut_matching} then both sides of the cut have volume at most $\left(1-\Omega\left(\frac{1}{\log n}\right)\right)\cdot 2m$. If we get Case (3), then the volume of $A'$ is $\Omega(m)$, so the volume left is reduced by a constant factor (\ie, $\vol(V\setminus A')\le \left(1-\Omega(1)\right)\cdot 2m$). 

     In any case, the volume of the largest component drops by a factor of at least $1-\Omega\left(\frac{1}{\log n}\right)$ across each level of the recursion, so the recursion depth is $O(\log^2 n)$. 

	    Since the components on one level of the recursion are all disjoint, the total running time on all components of one level of the recursion is 
        \[
            O\left(\frac{m \log n}{\phi}\right) + O\left(m\log^5 n + \frac{m\log^2 n}{\phi}\right) = O\left(m\log^5 n + \frac{m\log^2 n}{\phi}\right).
        \]
        The first term comes from the applications of trimming step (Theorem \ref{theorem:SW19_2.1}), and the second comes from the applications of the cut-matching step (Theorem \ref{theorem:our_cut_matching}). Hence, the total running time is $O\left(m\log^7 n + \frac{m\log^4 n}{\phi}\right)$.

	    To bound the number of edges between expander clusters, observe that in both Case (2) and Case (3) of Theorem \ref{theorem:our_cut_matching}, we always cut a component along a cut of conductance $O(\phi \log n)$, and in Case (3) this is true after the trimming step as well. Thus, we can charge the edges on the cut to the edges in the smaller side of the cut, so each edge is charged $O(\phi \log n)$. An edge can be on the smaller side of a cut at most $\log n$ times, so we can charge each edge at most $O(\phi\log^2 n)$ to pay for all the edges we leave between the final clusters. This bounds the total number of edges between the expanders to be at most $O(m\phi\log^2 n)$.
	\end{proof}
    \begin{remark}
        In \cite[Section 8]{LNPSsoda13} it is shown that one can implement Theorem \ref{theorem:SW19_2.1} with a running time of $\tilde{O}(m)$. As mentioned in Remark \ref{remark:fair_cut_global}, Theorem \ref{theorem:our_cut_matching} can also be implemented in $\tilde{O}(m)$ time. It then follows using a similar proof that Theorem \ref{theorem:our_expander_decomposition} can be implemented in $\tilde{O}(m)$ time. 
    \end{remark}

    \section{Algebraic Tools}
	\label{appendix:matrix_inequalities}
\begin{fact}[Exercise \uppercase\expandafter{\romannumeral 9\relax}.3.3~\cite{bhatia2013matrix}]\label{fact:book_trace_identities}
	Let $X,Y,A \in \RR^{n \times n}, m \in \NN$, then 
	\begin{enumerate}
	    \item $\tr(XY) = \tr(YX)$.
	    \item $\tr(A^{2m}) \le \tr((A \cdot A')^m)$.
	\end{enumerate}
	\end{fact}
\begin{proof} 
    \begin{enumerate}
	        \item Follows since 
            \[
                \tr(XY) = \sum_{i=1}^n {\sum_{j=1}^n{\X{}{i,j}\Y{}{j,i}}}
    	        = \sum_{j=1}^n {\sum_{i=1}^n{\Y{}{j,i}\X{}{i,j}}}
    	        = \tr(YX).
            \]
	        \item Using Schur decomposition (Chapter I.2 of~\cite{bhatia2013matrix}), we decompose $A = V U V^{*}$, where $V \in \CC^{n\times n}$ is a unitary matrix, 
            $U\in \CC^{n\times n}$ is an upper triangular matrix, and $U^\star$ is the \emph{conjugate transpose} of $U$ (that is, $U^{\star}_{i,j} = \overline{U_{j,i}}$). Observe that
	        \[\tr(A^{2m}) = \tr((V U V^{*})^{2m}) = \tr(V U^{2m} V^{*}) \numeq{3} \tr(V^{*} V U^{2m} ) = \tr(U^{2m}) = \sum_{i=1}^n{\U{}{i,i}^{2m}}, \]

            where Equality $(3)$ follows from the first part of the lemma.

	        Let $S = U U^*$, then $\Smat{}{i,i} = \sum_{j=1}^n{|\U{}{i,j}|^2} \ge |\U{}{i,i}|^2$, for every $i$. Moreover, $S$ is PSD Hermitian  
            matrix\footnote{A matrix $S \in \CC^{n\times n}$ is \emph{Positive Semi-Definite} (or PSD) if for every $0\neq v \in \CC^n$ it holds that $v^\star S v \ge 0$. A matrix $S \in \CC^{n\times n}$ is \emph{Hermitian} if $S^\star = S$.} so $S$ can be decomposed to $S = R D R^*$, where $R$ is a unitary matrix and $D$ is diagonal with non-negative real entries (this is essentially the spectral theorem, see Chapter I.2 of~\cite{bhatia2013matrix}).
            Since $R$ is unitary, it holds that $\sum_{j=1}^n|\R{}{i,j}|^2 = 1$, for every $i$. Therefore, we get that
	        \begin{align*}
	            \Matrix{S^m}{i,i} &= \Matrix{R D^m R^*}{i,i} = \sum_{j=1}^n{|\R{}{i,j}|^2 \D{}{j,j}^m} \numge{1}
    	        \left( \sum_{j=1}^n{|\R{}{i,j}|^2 \D{}{j,j}} \right)^m \\&=
    	        (\Smat{}{i,i})^m \ge |\U{}{i,i}|^{2m},
	        \end{align*}
            where Inequality $(1)$ holds since $f(x) = x^m$ is convex for $x\ge 0$.
	        We conclude that 
	        \begin{align*}
	             \tr((A \cdot A')^m) &= \tr((V U U^* V^*)^m) = \tr(V (U U^*)^m V^*) \numeq{3} \tr((U U^*)^m)\\
	             &= \tr(S^m) = \sum_{i=1}^n \Matrix{S^m}{i,i} \ge \sum_{i=1}^n (\U{}{i,i})^{2m} 
	             = \tr(A^{2m}),
	        \end{align*}
	    \end{enumerate}
     where Equality $(3)$ follows by the first part of the lemma.
	\end{proof}
\begin{lemma}[In the proof of Theorem \uppercase\expandafter{\romannumeral 9\relax}.3.5~\cite{bhatia2013matrix}]\label{lemma:bhatia_inequality}
	Let $X,Y \in \RR^{n \times n}$ be symmetric matrices. Then for any positive integer $k$,  $\tr((XY)^{2^k}) \le \tr(X^{2^k}Y^{2^k})$.
	\end{lemma}
	\begin{proof}
	   \begin{align*}
	       \tr((XY)^{2^k}) \numle{1} \tr(((XY)(XY)')^{2^{k-1}}) = 
	       \tr((XY^2 X)^{2^{k-1}}) \numeq{2} \tr((X^2 Y^2)^{2^{k-1}}).
	   \end{align*}
    Inequality (1) follows from Fact~\ref{fact:book_trace_identities}(2) and Equality (2) follows from Fact~\ref{fact:book_trace_identities}(1). 
	   The result follows by induction since $X^2$ and $Y^2$ are symmetric.
	\end{proof}
\begin{theorem} [Symmetric Rearrangement; Theorem A.2~\cite{orecchia2008partitioning}]
	\label{theorem:symmetric_rearrangement}
	    Let $X, Y\in \RR^{n\times n}$ be symmetric matrices. Then for any positive integer $k$,
	    \begin{align*}
	        \tr\left((XYX)^{2^k}\right) \le \tr\left(X^{2^k}Y^{2^k}X^{2^k}  \right)\ .
	    \end{align*}
	\end{theorem}
	\begin{proof}
	    \begin{align*}
	        \tr\left((XYX)^{2^k}\right) = \tr((X^2Y)^{2^k}) \le 
	        \tr(X^{2^{k+1}}Y^{2^k}) = \tr(X^{2^k}Y^{2^k}X^{2^k})
	    \end{align*}
     where the equalities follow from Fact~\ref{fact:book_trace_identities}(1) and the inequality follows from Lemma~\ref{lemma:bhatia_inequality}.
	\end{proof}

	The following lemma is standard.
	\begin{lemma}
	\label{lemma:laplacian_mult}
	    For every weighted graph $G = (V,E,w)$ and $v \in \mathbb{R}^{n}$, where $n=|V|$, it holds that $v^t \mathcal{L}(G) v = \sum_{\{s,t\}\in E}{w(\{s,t\})\cdot (v_s - v_t)^2}$.

	\end{lemma}

    \begin{lemma}
    \label{lemma:matching_laplacian}
	    Let $M$ be a $\dGG$-matching on a graph with vertices $V=[n]$. For every $A \in \RR^{n\times n}$, it holds that $\tr(A' \mathcal{L}(M) A) =\sum_{\{i,j\}\in M}{\norm{\A{}{i}-\A{}{j}}_2^2}$, where $\A{}{i}\in \RR^n$ is row $i$ of $A$. Note that some pairs may appear in the sum multiple times.

	\end{lemma}
\begin{proof}
	    The $i$'th column of $A$ is $A'(i)\in\RR^n$.
        By Lemma~\ref{lemma:laplacian_mult},
	    \begin{align*}
	        \tr(A' \mathcal{L}(M) A) &= \sum_{k=1}^{n}{\left(A'(k)\right)' \mathcal{L}(M) A'(k)} 
    	    = \sum_{k=1}^{n}{\sum_{\{i,j\}\in M}{(\A{}{i,k}-\A{}{j,k})^2}}\\
    	    &=  \sum_{\{i,j\}\in M}{\sum_{k=1}^{n}{(\A{}{i,k}-\A{}{j,k})^2}}
    	    = \sum_{\{i,j\}\in M}{\norm{\A{}{i}-\A{}{j}}_2^2}
        \end{align*}

	\end{proof}
\begin{lemma}\label{lemma:normalized-laplacian-eigenvalues}
	Let $A \in \RR^{n \times n}$ be a symmetric $\dGG$-stochastic matrix with non-negative entries, 
    then the eigenvalues of the normalized Laplacian $\mathcal{N}(A) =I - \Dminushalf A \Dminushalf$ are in the range $[0,2]$. 
	\end{lemma}
\begin{proof}
	    Since $\mathcal{L}(A) \succeq 0$, we get that $\mathcal{N}(A) = \Dminushalf \mathcal{L}(A) \Dminushalf \succeq 0$.
        Therefore, it is left to show that the eigenvalues of $\Dminushalf A \Dminushalf$ are at least $-1$. Let $X = D + A$. Note that $X \succeq 0$. Indeed, similarly to Lemma~\ref{lemma:laplacian_mult}, for any $v\in \RR^n$,
	    \[v'Xv = \sum_{i=1}^n\sum_{j=i}^n \A{}{i,j} (v_i+v_j)^2 \ge 0.\]

        Hence, $I + \Dminushalf A \Dminushalf = \Dminushalf X \Dminushalf \succeq 0$ which means that the eigenvalues of $\Dminushalf A \Dminushalf$ are at least $-1$.
	\end{proof}
\begin{lemma}
    \label{lemma:cs_mean}
        Let $\{v_i\}_{i=1}^k$ be a set of $k$ vectors in $\RR^n$. Then,
        \begin{align*}
            k\norm{\frac{\sum_{i=1}^k{v_i}}{k}}_2^2 \le \sum_{i=1}^k{\norm{v_i}_2^2}
        \end{align*}
    \end{lemma}
    \begin{proof}
        By the Cauchy-Schwartz inequality, for all $j\in [n]$, 
        \begin{align*}
        k\left({\frac{\sum_{i=1}^k{v_{i,j}}}{k}}\right)^2 = \frac{1}{k}\left({\sum_{i=1}^k{v_{i,j}}}\right)^2
        \le \frac{1}{k}\cdot k\cdot \sum_{i=1}^k{v_{i,j}^2} = \sum_{i=1}^k{v_{i,j}^2},
        \end{align*}
        Therefore,
        \begin{align*}
        k\norm{\frac{\sum_{i=1}^k{v_i}}{k}}_2^2 = k\sum_{j=1}^n{\left({\frac{\sum_{i=1}^k{v_{i,j}}}{k}}\right)^2}=
        \sum_{i=1}^n
        \frac{1}{k}\left({\sum_{i=1}^k{v_{i,j}}}\right)^2\le \sum_{j=1}^n{\sum_{i=1}^k{v_{i,j}^2}} = \sum_{i=1}^k{\sum_{j=1}^n{v_{i,j}^2}} = \sum_{i=1}^k{\norm{v_i}_2^2}
        \end{align*}
    \end{proof}

    \section{Projection Lemmas}
	\label{appendix:projection}
	The following fact is quite standard in the analysis of algorithms in the cut-matching framework.

\begin{lemma} [Gaussian Behavior of Projections; Lemma 3.5 of~\cite{khandekar2009graph}]
	\label{lemma:projection_original}
	    Let $v\in \RR^n$ and let $r\in \mathbb{S}^{n-1}\subseteq \RR^n$ be a random unit vector. Let $u = \langle v,r\rangle$ be the projection of $v$ onto $r$. Then:
	    \begin{enumerate}
	        \item $\EE[u^2]=\frac{1}{n}\norm{v}_2^2$.
	        \item 
	        $ u^2 \le \frac{\alpha\log n}{n}\norm{v}_2^2
            $ holds with probability of at least $1-n^{-\alpha/4}$, for every constant $\alpha > 0 $ and large enough $n$.\footnote{The exact condition on $n$ is $\frac{n}{\log n} \ge 16\alpha$.}
	    \end{enumerate}
	\end{lemma}

\begin{lemma} 
	\label{lemma:projection_simple}
	    Let $\{v_i\}_{i=1}^k$ be a set of $k\le n$ vectors in $\RR^n$. For $i\in [k]$, let $u_i=\langle v_i,r\rangle$ be the projection of $v_i$ onto a random unit vector $r\in \mathbb{S}^{n-1}\subseteq \RR^n$. Then:

	    \begin{enumerate}
	        \item $\EE[u_i^2]=\frac{1}{n}\norm{v_i}_2^2$ for all $i$.
	        \item 
	        \[
    		        u_i^2 \le \frac{\alpha\log n}{n}\norm{v_i}_2^2
            \]
    	    holds for all $i$ with probability of at least $1-n^{-\alpha/8}$, for every constant $\alpha\ge8$ and large enough~$n$.
	    \end{enumerate}
	\end{lemma}

 \begin{proof}
     Apply Lemma~\ref{lemma:projection_original} on the set of vectors $\{v_i\}_{i=1}^k$. Condition $(1)$ follows immediately. Moreover, by Lemma~\ref{lemma:projection_original} and union bound, Condition $(2)$ holds for all $i$ with probability at least $1-k\cdot n^{-\alpha/8} \ge 1- n^{1-\alpha/4}$. Since $\alpha \ge 8$, it holds that $1- n^{1-\alpha/4} \ge 1- n^{-\alpha/8}$. 
 \end{proof}

\begin{lemma}
	\label{lemma:projection_pairs}
	    Let $\{v_i\}_{i=1}^k$ be a set of $k\le n$ vectors in $\RR^n$. For $i\in [k]$, let $u_i=\langle v_i,r\rangle$ be the projection of $v_i$ onto a random unit vector $r\in \mathbb{S}^{n-1}\subseteq \RR^n$. Then:

	    \begin{enumerate}
	        \item $\EE[u_i^2]=\frac{1}{n}\norm{v_i}_2^2$ for all $i$, and $\EE[(u_i - u_j)^2]=\frac{1}{n}\norm{v_i - v_j}_2^2$ for all pairs $(i,j)$.
	        \item 
	        \begin{align*}
	            u_i^2 &\le \frac{\alpha\log n}{n}\norm{v_i}_2^2
    		    \\
    		    (u_i - u_j)^2 &\le \frac{\alpha\log n}{n}\norm{v_i - v_j}_2^2
	        \end{align*}
    	    holds for all indices $i$ and pairs $(i,j)$ with probability of at least $1-n^{-\alpha/8}$, for every constant $\alpha\ge16$ and large enough $n$.
	    \end{enumerate}
	\end{lemma}
	\begin{proof}
	    Similarly to the proof of Lemma~\ref{lemma:projection_simple}, apply Lemma \ref{lemma:projection_original} on the set of vectors
	    \begin{align*}
	        \{v_i : i\in [k]\}\cup \{v_i - v_j : i,j\in[k]\}
	    \end{align*}
	    By the union bound all the inequalities in Condition $(2)$ are satisfied simultaneously with probability at least $1-\frac{k\cdot (k+1)}{2}\cdot n^{-\alpha/4}\ge1-k^2\cdot n^{-\alpha/4}\ge1- n^{2-\alpha/4}$. Since $\alpha \ge 16$, it holds that $1- n^{2-\alpha/4} \ge 1- n^{-\alpha/8}$.
	\end{proof}

\begin{lemma}
	\label{lemma:projection_krv}
	    Let $\{v_i\}_{i=1}^k$ be a set of $k\le n+1$ vectors in $\RR^n$. Let $x\in \RR^n$ be a vector, and let $c\in \RR$. Assume that $\langle v_i, x \rangle = c$ for all $i\in [k]$. For $i\in [k]$, let $u_i=\langle v_i,r\rangle$ be the projection of $v_i$ onto a random unit vector $r\in \mathbb{S}^{n-1}\subseteq \RR^n$ orthogonal to $x$. Then:

	    \begin{enumerate}
	        \item $\EE[(u_i - u_j)^2]=\frac{1}{n-1}\norm{v_i-v_j}_2^2$ for all pairs $(i,j)$.
	        \item 
	        \[
    		        (u_i - u_j)^2 \le \frac{\alpha\log n}{n-1}\norm{v_i-v_j}_2^2
            \]
    	    holds for all pairs $(i,j)$ with probability of at least $1-n^{\alpha/8}$, for every constant $\alpha\ge 16$ and large enough $n$.

	    \end{enumerate}
	\end{lemma}
	\begin{proof}
	    Note that the vectors $\{v_i - v_j : i,j\in[k]\}$ are all in the $n-1$ dimensional space orthogonal to $x$ (as $\langle v_i - v_j, x\rangle = c - c = 0$). We can therefore use Lemma \ref{lemma:projection_original} which proves the first item. Given that $\alpha\ge 16$,  Lemma \ref{lemma:projection_original} and the union bound imply that all $\binom{k}{2} \le \binom{n+1}{2}\le  n^2$ inequalities in the second item of the lemma are satisfied with probability at least $1- n^{2-\alpha/4} \ge 1- n^{-\alpha/8}$.

	\end{proof}

\section{\unitFlow}
\label{appendix:unit_flow_matching}

	The goal of this section is to prove Lemma \ref{lemma:unit_flow} using a variation of the push-relabel algorithm by \cite{gt87} called \emph{\unitFlow}~  (\cite{henzinger2017flow, saranurak2019expander} used the name \emph{Unit-Flow}). The following non-vanilla definitions regarding flow will be useful for us throughout this appendix. 

	\begin{definition} [Flow Problem]
	\label{def:flow_problem}
	    A \emph{flow problem} $\Pi$ on a graph $G=(V,E)$ is specified by a source function $\Delta:V\to \RR_{\ge 0}$, a sink function $T:V\to\RR_{\ge 0}$, and edge capacities $c:E\to\RR_{\ge 0}$. We use \emph{mass} to refer to the substance being routed. 

	    For a node $v$, $\Delta(v)$ specifies the amount of mass initially placed on $v$ (source mass at $v$), and $T(v)$ specifies the capacity of $v$ as a sink. For an edge $e=\{u,v\}$, $c(\{u,v\})$ bounds how much mass can be routed along the edge (in either direction).  

	\end{definition}
	\begin{definition} [Flow]
	\label{def:flow}
 Since the flow along an edge $\{u,v\}$ is directed we associate two \emph{arcs} $(u,v)$ and $(v,u)$ with each edge $\{u,v\}$.
	    A \emph{flow} in the graph $G$ is a function $f:V\times V\to \RR$ that satisfies $f(u,v)=-f(v,u)$ for all $u,v\in V$ and $f(u,v)=0$ for all $\{u,v\}\notin E$. $f(u,v)>0$ means that mass is routed in the direction from $u$ to $v$ (along the arc $(u,v)$). If $f(u,v)=c(\{u,v\})$, then we say that the arc $(u,v)$ is \emph{saturated} or that \emph{$\{u,v\}$ is saturated in the direction from $u$ to $v$}. 
	    Given $\Delta$, we also extend $f$ to a function on the vertices, where $f(v)=\Delta(v)+\sum_{u\in V}{f(u,v)}$ is the amount of mass ending at $v$ after the routing according to $f$. If $f(v)\ge T(v)$, then we say $v$'s sink is \emph{saturated}.
	\end{definition}

	\begin{definition} [Feasible Flow]
	\label{def:feasible_flow}
	    We say that $f$ is a \emph{feasible} flow for a flow problem $\Pi$ if:
	    \begin{enumerate}
	        \item $\forall e=\{u,v\}\in E, |f(u,v)|\le c(\{u,v\})$ (\ie, $f$ obeys edge capacities).
	        \item $\forall v\in V, \sum_{u\in V}{f(v,u)}\le \Delta(v)$ (\ie, the net amount of mass routed away from a node can be at most the amount of its initial mass). In other words, $f(v) \ge 0$.
	        \item $\forall v\in V, f(v)\le T(v)$.
	    \end{enumerate}
	\end{definition}

	The goal of this section is to prove Lemma \ref{lemma:unit_flow}. We repeat it here for convenience, using the new definitions. 

	\begin{lemma} [Lemma \ref{lemma:unit_flow}]

	    Let $G=(V,E)$ be a graph with $n$ vertices and $m$ edges, let $A^l, A^r \subseteq V$ be multisets such that $|A^r| \ge \frac{1}{2}m, |A^l| \le \frac{1}{8}m$, and let $0<\phi<\frac{1}{\log n}$ be a parameter. For a vertex $v\in V$, denote by $m_v$ the number times $v$ appears in $A^l$, and by $\bar{m}_v$ the number of times $v$ appears in $A^r$. Assume that $m_v + \bar{m}_v\le \dG{}{v}$ for all $v\in V$. We define the flow problem $\Pi(G)$, as the problem in which each vertex $v\in A^l$ is a source of $m_v$ units of mass and each vertex $v\in A^r$ is a sink with capacity of $\bar{m}_v$ units. All edges have the same capacity $c=\Theta\left(\frac{1}{\phi\log n}\right)$, which is an integer. 
	    Then, in time $O(\frac{m}{\phi})$, we can find:
	    \begin{enumerate}
	        \item A feasible flow $f$ for $\Pi(G)$ (\ie,\ a flow that ships out $m_v$ units from every source $v\in A^l$); or
	        \item A cut $S$ where $\Phi_{G}(S, V\setminus S)\le\frac{7}{c}=O(\phi\log n)$, $\vol(V\setminus S) \ge \frac{1}{3}m$ and a feasible flow for the problem $\Pi(G-S)$, where we only consider the sub-graph $G[V\setminus S]$, vertices $v\in A^l\setminus S$ are sources of $m_v$ units, and vertices $v\in A^r\setminus S$ are sinks of $\bar{m}_v$ units. 
	    \end{enumerate}
	\end{lemma}

     The proof of this lemma 
     is by an algorithm similar
     to the one
      described in Section 3 and in Appendix A of~\cite{saranurak2019expander}, which is  a dynamic version of the algorithm in \cite{henzinger2017flow}. The fact that we work directly on the original graph (and not the subdivision graph) allows us to avoid many technicalities regarding the returned cut. Specifically, we don't have to argue how to modify the algorithm such that the returned cut in the subdivision graph will be easy to transform into an appropriate cut in the original graph.

	Since \emph{\unitFlow} is a variation of the push-relabel algorithm \cite{gt87}, we give the following standard push-relabel definitions.
\begin{definition} [Preflow]
	\label{def:preflow}
	    We say that $f$ is a \emph{preflow} for a flow problem $\Pi$ if it satisfies Conditions (1) and (2) in Definition \ref{def:feasible_flow}.
	    For a node $v\in V$, we define the \emph{excess mass} at $v$ as $ex(v)=\max\{f(v)-T(v), 0\}$. We say that there is excess mass at $v$ if $ex(v)>0$.

	    The residual capacity of the arc $(u,v)\in V\times V$ with respect to the preflow $f$ is $r_f(u,v) = c(\{u,v\}) - f(u,v)$.
	\end{definition}

Our algorithm maintains a labeling of the nodes $\ell: V\rightarrow \{0,1,\ldots,h \}$ for a parameter $h$ that we set as $h \defeq 1000c\log m=\Theta(1/\phi)$.
\begin{definition} [valid-state, valid-solution~\cite{saranurak2019expander}]
	\label{def:G-valid}
	    We say that a tuple $(\Delta, f, l)$ is a \emph{valid-state in $G$} if it satisfies:
	    \begin{enumerate}
	        \item If $l(u)>l(v)+1$ and $\{u,v\}$ is an edge, then $(u,v)$ is saturated, \ie, $f(u,v)=c$.
	        \item If $l(u)\ge1$, then $u$'s sink is saturated, \ie, $f(u)\ge T(u)$.
	    \end{enumerate}
	    We say that $(\Delta, f, l)$ is a \emph{valid-solution in $G$} if it additionally satisfies
        \begin{enumerate}
            \setcounter{enumi}{2}
            \item If $l(u)<h$ then there is no excess mass at $u$, \ie, $f(u)\le T(u)$. Together with (2), this means that if $l(u)\in \{1, \ldots, h-1\}$ (\ie,  $l(u)\notin \{0,h\}$), then $f(u) = T(u)$.
        \end{enumerate} 
	\end{definition}

	In this appendix, we only consider flow problems where $c(\{u,v\})=c=\Theta(\frac{1}{\phi\log n})$ for all edges $\{u,v\}\in E$, and $T(v)=\bar{m}_v$ for all vertices $v\in V$. As specified in Lemma \ref{lemma:unit_flow}, $c$ is an integer.

	The algorithm (see Algorithm \ref{algo:unit_flow_matching}) runs in  iterations. In iteration $t$, we use Algorithm \ref{algo:unit_flow} to improve the current flow. At the end of the iteration, Algorithm \ref{algo:unit_flow} specifies a set of nodes $S_t$ to be removed from the graph. The next iteration runs on the \emph{surviving} vertices. At the end, the returned cut $S$ (for Lemma \ref{lemma:unit_flow}) is the union of these sets $S_t$.
\begin{remark}
	    Throughout this appendix we use $S_t$ to denote the cut computed in iteration $t$ of \emph{\unitFlow}. It should not be confused with the $S_t$ in section \ref{section:our-algorithm}, which denotes the cut at iteration $t$ of \emph{Cut-Matching}.
	\end{remark}
\begin{algorithm}[hbt!]
        \caption{\unitFlow~\cite{henzinger2017flow, saranurak2019expander} }
        \label{algo:unit_flow}
        \begin{algorithmic}
            \Function{\unitFlow}{$G, h, (\Delta, f, l)$}\\
                \textbf{Require:} {$(\Delta, f, l)$ is valid-state.}
                \While{$\exists v$ where $l(v) < h$ and $ex(v)>0$}
                    \State let $v$ be such a vertex with smallest $l(v)$.
                    \State Push/Relabel($v$).
                \EndWhile \\
                \State \underline{\textbf{Compute a sparse cut $S$:}}
                \State $S\gets \emptyset$.
                \If{$B_h \neq \emptyset $}
                    \Comment{$B_i \defeq \{v \mid l(v)=i \}$ and $B_{\ge i} \defeq \{v \mid l(v) \ge i \}$.}
                    \State $i\gets h$.
                    \While{$|\{ (u,v) \in E(B_{\ge i}, V\setminus B_{\ge i}) \mid -c\le f(u,v)<c\}| > \frac{5\log (2m)}{h}\cdot \vol(B_{\ge i})$}
                    \State $i \gets i-1$.
                    \EndWhile
                    \Comment{The while loop ends with $i>0$ by Lemma~\ref{lemma:unit_flow_cut}.}
                    \State $S \gets B_{\ge i}$.
                \EndIf

                \State \Return{$(\Delta, f, l), S$}.
            \EndFunction
            \\
            \Function{Push/Relabel}{$v$}
                \If{$\exists$ edge $(v,u)$ such that $r_f(v,u)>0$ and $l(v)=l(u)+1$}
                    \State Push($v,u$).
                \Else
                    \State Relabel($v$).
                \EndIf
            \EndFunction
            \\
            \Function{Push}{$v,u$}\\
                \textbf{Require:} {$r_f(v,u)>0$, and $ex(u)=0$. }

                \Comment{$v$ has the smallest label among $\{v \in V \mid ex(v) > 0\}$.}
                \State $f(v,u)\gets f(v,u) + 1, f(u,v) \gets f(u,v) - 1$.
            \EndFunction
            \\
            \Function{Relabel}{v}\\
                \textbf{Require:} {$ex(v)>0, l(v)<h$, and $\forall u\in V,$ $r_f(v,u)>0\Longrightarrow l(v)\le l(u)$. }
                \State $l(v) \gets l(v) + 1$.
            \EndFunction
        \end{algorithmic}
    \end{algorithm}

    The procedure to route flow in one iteration is called \emph{\unitFlow}, and is described in Algorithm \ref{algo:unit_flow} [Algorithm 3.1 of~\cite{henzinger2017flow}, Algorithm 3 of ~\cite{saranurak2019expander}]. This algorithm gets a valid-state $(\Delta, f, l)$, and returns a valid-solution $(\Delta, f', l')$ and a cut $S$.  It
    maintains a preflow $f$ and labels $l:V\to\{0,...,h\}$ such that $(\Delta, f, l)$ is always valid. The cut returned by \emph{\unitFlow} is a \emph{level cut} $B_{\ge i^\star} \coloneqq \{v \mid l'(v) \ge i^\star \}$ for some $0 < i^\star \le h$.

    At a high level, the goal of the \emph{\unitFlow} algorithm is to dissipate the input mass $\Delta$ to the sinks. However, nodes that reach label $l(v)=h$ are considered \emph{inactive} and we do not try to push their excess mass away. Instead, we remove them from the surviving vertex set
    in the next iteration of Algorithm \ref{algo:unit_flow_matching} [Algorithm at Section 3.4.2 of~\cite{saranurak2019expander}]. In other words, they do not participate in the following iterations of \emph{\unitFlow}. 

  Lemmas \ref{lemma:SW19_A.1} and \ref{lemma:unit_flow_cut} detail the properties of Algorithm \ref{algo:unit_flow}. The proof of Lemma \ref{lemma:SW19_A.1} is by standard arguments for push/relabel algorithms~\cite{gt87}. 

\begin{lemma}[\cite{saranurak2019expander}]
    \label{lemma:SW19_A.1}
        During its execution, \emph{\unitFlow} maintains a preflow $f$ and labels $l$ so that $(\Delta, f, l)$ is valid-state, given that it was valid initially.
        When it terminates, $(\Delta, f, l)$ is a valid-solution.

    \end{lemma}

    Lemma \ref{lemma:unit_flow_cut} gives the key properties of the cut $S$ that \emph{\unitFlow} returns at each iteration. It is based on Lemma 3.4 in~\cite{saranurak2019expander}. 
    Intuitively, if we find a valid-solution which is not a feasible flow, the cut $S$ should satisfy the following properties:
    \begin{itemize}
        \item All vertices with label $h$ should be in $S$. Since they are not considered active, we ``give up'' on these vertices and do not try to dissipate their excess mass. They should therefore be removed in the following iteration.
        \item The cut should be ``close'' to minimal. This means that most of the edges crossing from $S$ to $V\setminus S$ should be saturated by the preflow. Explicitly, we upper-bound the number of non-saturated edges going out from $S$. This allows us to bound the total number of edges crossing $(S, V\setminus S)$ (both saturated and unsaturated), in the proof of Lemma \ref{lemma:unit_flow}.
        \item Finally, to bound the number of vertices in $S\cap A^r$ (see Lemma \ref{lemma:S_t_bound}), we also require that any vertex $v \in S$ has a label of at least $1$ and hence $v$'s sink is saturated.
    \end{itemize}

    We prove in Lemma~\ref{lemma:unit_flow_cut} that such a set $S$ exists and moreover, it is achieved by a level cut, which is defined as follows.

    \begin{definition}[level cut]
        Let $(\Delta, f', l')$ be a valid-solution and let $0\le i \le h$. The $i$'th level is defined as $B_i = \{v \mid l'(v) = i \}$. The $i$'th level cut is defined as $(B_{\ge i}, V\setminus B_{\ge i})$, where $B_{\ge i} = \{v \mid l'(v) \ge i \}$.
    \end{definition}

    \begin{lemma}
    \label{lemma:unit_flow_cut}
	    Given a graph $G=(V,E)$ where $n=|V|, m=|E|$, a parameter $h$ and a valid-state $(\Delta, f, l)$, \emph{\unitFlow} outputs a valid-solution $(\Delta, f', l')$ and a set $S\subseteq V$ where we have:
	    \begin{enumerate}
	        \item $S=\emptyset$ and $f'$ is a feasible flow for $\Delta$; or:
	        \item $S\neq\emptyset$ is a cut, such that $\{u\in V \mid l'(u) = h\}\subseteq S$, $\{u\in V \mid l'(u) = 0\}\cap S = \emptyset$. Moreover, the number of edges $\{u,v\}$ such that $u\in S, v\in V\setminus S$, and $-c\le f'(u, v) < c$ is at most $\frac{5\log (2m)}{h}\cdot \vol(S)\le \frac{6}{1000c}\cdot \vol(S)$.\footnote{$5\log (2m) \le 6\log m$ holds for sufficiently large $m$.} Note that these are all the edges crossing the cut except those sending exactly $c$ units of mass from $S$ to $V\setminus S$. 
	    \end{enumerate}
     The computation of the set $S$ takes $O(\vol(S))$ time.
    \end{lemma}
    The proof of this lemma follows~\cite{saranurak2019expander}. We give the details for completeness.
    \begin{proof}
    \label{proof:lemma:unit_flow_cut}
        If $B_h=\emptyset$, no vertex has positive excess, and we end up with Case (1) of the lemma. 

        Assume that $B_h \neq \emptyset$. We prove that there is a level cut $S = B_{\ge i^\star}$, for some $0< i \le h$, that satisfies the requirements of Case (2).
        For any $i\ge 1$, an edge $\{u,v\}$ crossing the cut $B_{\ge i}$, with $u\in B_{\ge i}$ and $v\in V\setminus B_{\ge i}$, must be of one of two types:
        \begin{enumerate}
            \item $l'(u) = i, l'(v)=i-1$
            \item $l'(u)-l'(v)>1$. Note that, by Lemma \ref{lemma:SW19_A.1}, all edges of this type must be saturated, \ie, $f'(u,v)=c$.
        \end{enumerate}
      Let $z_1(i)$ denote the number of edges of the first type. We sweep from $i=h$ to $1$, and argue below that there must be some $i^\star\in [1, h]$ such that 
        \begin{equation} \label{eq:istar}
            z_1(i^\star)\le \frac{5\log (2m)}{h}\cdot \vol(S_{i^\star})\ .
        \end{equation}
        We let $S$ be $B_{\ge i^\star}$. Clearly $S$ satisfies all properties of Case (2) of the lemma. 
        It remains to argue that an index $i^\star\in [1,h]$ that satisfies Equation (\ref{eq:istar}) exists.
        Assume towards a contradiction that $z_1(i)\ge \frac{5\log (2m)}{h}\cdot \vol(B_{\ge i})$ for all $i\in [1, h]$. This implies that $\vol(B_{\ge i-1})\ge (1+\frac{5\log (2m)}{h}) \vol(B_{\ge i})$ for all $i\in [1, h]$. Since $\vol(B_{\ge h})\ge 1$ (because $B_{\ge h} = B_h\neq\emptyset$) and since $h\ge 1000 \log m$, we have $\vol(B_{\ge 0})\ge(1+\frac{5\log (2m)}{h})^{h} > 2m$  which is a contradiction.

        We can compute $S$ in $O(\vol(S))$ time by maintaining the sizes of the cuts as we sweep from $i=h$ to $1$, and stopping as soon as we find a level cut satisfying Equation (\ref{eq:istar}). 

    \end{proof}

\begin{algorithm}[hbt!]
        \caption{Algorithm for Lemma \ref{lemma:unit_flow}}
        \label{algo:unit_flow_matching}
        \begin{algorithmic}[1]
            \Function{Route-Flow}{$G=(V,E), A^l, A^r$} \\
            \textbf{Require:} {$|A^l|\le \frac{1}{8}m$, $|A^r|\ge \frac{1}{2}m$.} 
                \State $c \gets \Theta\left(\frac{1}{\phi\log n}\right)$. \Comment{$c$ is an integer.}
                \State $h \gets 1000c\log m=\Theta(1/\phi)$.
                \State $t \gets 1$.
                \State $A_1\gets V$.
                \State $S \gets \emptyset$.
                \State $T(u)\gets \bar{m}_u$, $\forall u\in V$.
                \State $\Delta_1(u)\gets m_u$, $\forall u\in V$.
                \State $f_1(u,v)\gets 0$, $\forall (u,v)\in E$.
                \State $l_1(v) \gets 0$, $\forall v\in V$.
                \While{Not found a feasible flow}
                    \State $(\Delta_t, f'_t, l'_t), S_t \gets$ \unitFlow($G[A_t], h, (\Delta_t, f_t, l_t)$).
                    \If{$S_t = \emptyset$} \Comment{Found a feasible flow.}

                        \State\Return{$(S, f'_t)$}. 
                    \Else
                        \State $A_{t+1}\gets A_t \setminus S_t$.
                        \State $S \gets S \cup S_t$.
                        \State $f_{t+1} \gets f'_t$ restricted to $A_{t+1}\times A_{t+1}$.
                        \State $l_{t+1} \gets l'_t$ restricted to $A_{t+1}$.
                        \State For each $u\in A_{t+1}$,
                        \begin{align*}
                            \Delta_{t+1}(u)\gets \Delta_t(u) + \max\left(0, \sum_{v\in S_t}{f'_t(v,u)}\right)
                        \end{align*}
                        \State $t \gets t+1$.
                    \EndIf
                \EndWhile
            \EndFunction
        \end{algorithmic}
    \end{algorithm}

    The algorithm proving Lemma \ref{lemma:unit_flow} consists of a number of iterations of \emph{\unitFlow} and is shown in Algorithm \ref{algo:unit_flow_matching}. 
    It is based on the dynamic adaptation of \emph{Unit-Flow} presented in Section 3.4.2 of~\cite{saranurak2019expander}. 
    In iteration $t$, we run \emph{\unitFlow} on the current set $A_t$ of surviving vertices.

    \emph{\unitFlow} moves the mass around and possibly also returns a cut $S_t$ (the cut $S$ at iteration $t$). If it did return such a cut, we remove it from the current set of surviving vertices and set $A_{t+1} \defeq A_t \setminus S_t$. This removal might reduce the amount of mass entering a vertex $u\in A_{t+1}$. We compensate for this by increasing the source mass of $u$. That is, we increase $\Delta_{t+1}(u)$ by $\sum_{v\in S_t}{f'_t(v,u)}$ (if it is positive), where $f'_t$ is the preflow at the end of the $t$'th iteration of Algorithm~\ref{algo:unit_flow_matching}.

    \begin{remark}
        Note that when $\Delta_{t+1}(u)$ increases, no new mass was introduced to the system. This is an artificial way to ensure that $(\Delta_{t+1}, f_{t+1}, l_{t+1})$ is a valid state in $G[A_{t+1}]$, as stated in Lemma~\ref{lemma:SW19_3.5}.

        On the other hand, if $\sum_{v\in S_t}{f'_t(v,u)}$ is negative (\ie, $u$ sends mass to $S_t$), then the deletion of $S_t$ causes an increase in $f_{t+1}(u)$ (the mass at $u$, with respect to the edges in $A_{t+1}$) since $u$ keeps the outbound mass. This does add new mass to the system (even though $\Delta_{t+1}(u) = \Delta_t(u)$)
    \end{remark}

\begin{lemma}[\cite{saranurak2019expander}]
    \label{lemma:SW19_3.5}
        For each round $t$, $(\Delta_t, f_t, l_t)$ is a valid-state in $G[A_t]$.
    \end{lemma}

\begin{proof}

    By induction on $t$. The first condition of Definition \ref{def:G-valid} is clear. It suffices to prove that $f_{t}(v) \ge T(v)$ for every $v \in A_t$ that satisfies $l_{t}(v) \ge 1$. Indeed, Algorithm~\ref{algo:unit_flow_matching} ensures the slightly stronger condition, that $f_{t}(v)\ge f'_{t-1}(v)\ge T(v)$. 
\end{proof}

    Note that $\Delta_{t+1}(v)\ge\Delta_t(v)$ for all nodes $v\in A_{t+1}$.
    We define the total increase in $f(\cdot)$ (\ie, the total mass in the system) throughout the execution of the dynamic version of \emph{\unitFlow}. This value is crucial to bound the running time of the algorithm and the sizes of the cuts that it produces.
    \begin{definition} [Total Balance~\cite{saranurak2019expander}]
    \label{def:f_total}
        We denote by ${f}_\text{total}$ the total increase in $f$ that we introduce at various vertices over all rounds. More formally, for each round $t>1$, let us denote the \emph{change in $f$ introduced at vertex $u$ at the end of round $t-1$} as ${\delta f}_t$. That is, for a vertex $u\in A_t$,
        \begin{align*}
            {\delta f}_t(u) = f_t(u) - f'_{t-1}(u) 
            &= 
            \left(\Delta_t(u) +  \sum_{v\in A_{t}}{f_{t}(v,u)}\right) - \left(\Delta_{t-1}(u) + \sum_{v\in A_{t-1}}{f'_{t-1}(v,u)}\right)
            \\ & =\Delta_t(u) - \Delta_{t-1}(u) - \sum_{v\in S_{t-1}}{f'_{t-1}(v,u)} \\ & = 
            \max\left(0, \sum_{v\in S_{t-1}}{f'_{t-1}(v,u)}\right) - \sum_{v\in S_{t-1}}{f'_{t-1}(v,u)} \\
         & =  \max\left(0 , -\sum_{v\in S_{t-1}}{f'_{t-1}(v,u)}\right)\ .        
        \end{align*}
        In particular, $f_t(u) \ge f'_{t-1}(u)$.
        In other words, if at the end of iteration $t-1$,  $u$ sent mass to $S_{t-1}$, then at iteration $t$ this mass remains at $u$ (and consequently $f(u)$ increases) and if, overall, $u$ received mass from $S_{t-1}$ at iteration $t-1$, then at the next iteration $u$ gets this mass as initial mass in $\Delta_t(u)$ (so $f_{t}(u) = f'_{t-1}(u)$). 

        Let ${\delta f}_t(A_t) = \sum_{u\in A_t}{{\delta f}_t(u)}$ be the net flow sent from $A_t$ to $S_{t-1}$ by $f'_{t-1}$ (that is, the sum over net flows of vertices $u\in A_t$ that send more flow to $S_{t-1}$ than they receive).
        We also denote the total change in $f$ at the beginning of round $1$ (\ie, the total initial mass) as:
        \begin{align*}
            \Delta_1(V) \defeq \sum_{v\in V}{\Delta_1(v)} = \sum_{v\in V}{m_v} = |A^l|\ .
        \end{align*}

        Then, ${f}_\text{total}$, the total mass introduced to the network during Algorithm~\ref{algo:unit_flow_matching}, includes the \emph{initial} mass at the beginning of round $1$ and the \emph{increase} in $f$ over all rounds $t>1$:
        \begin{align*}
            {f}_\text{total} \defeq \Delta_1(V)+\sum_{t>1}{{\delta f}_t(A_t)}\ .
        \end{align*}
    \end{definition}

    Let $S=\bigcup_{t\ge 1}{S_t}, A'=V\setminus S$. Suppose the algorithm terminates at round $\kappa$, \ie, $A'=A_\kappa$, $S_\kappa =\emptyset$.
\begin{lemma}
    \label{lemma:f_total_is_sum_of_f}
        ${f}_\text{total}$ is the sum, over all vertices $v\in V$, of their ``final'' $f'_t(v)$ (when they were removed). More formally,
        $f_\text{total} = \sum_{v\in A'}{f'_\kappa(v)} + \sum_{t\ge 1}{\sum_{v\in S_t}{f'_t(v)}}$.
    \end{lemma}
    \begin{proof}
        For all $t\ge 1$, recall that $f_t$ is the preflow at the beginning of the $t$'th iteration of Algorithm~\ref{algo:unit_flow_matching} and $f'_t$ is the preflow (that is a part of a valid-solution) at the end of the iteration.  Therefore, 
        \begin{align*}
            f_\text{total} &= \Delta_1(V)+\sum_{t>1}{{\delta f}_t(A_t)} = \sum_{v\in A_1}{f_1(v)}+\sum_{t>1}{{\delta f}_t(A_t)} 
            \\
            &= \sum_{v\in A_1}{f_1(v)}+\sum_{t>1}{\sum_{v\in A_t}{(f_t(v) - f'_{t-1}(v))}} \numeq{1} \sum_{v\in A_1}{f'_1(v)}+\sum_{t>1}{\sum_{v\in A_t}{(f'_t(v) - f'_{t-1}(v)})} 
            \\
            &= \sum_{v\in A'}{f'_\kappa(v)} + \sum_{t\ge 1}{\sum_{v\in S_t}{f'_t(v)}} \ ,
        \end{align*}
        where Equation~$(1)$ holds since $f_t$ and $f'_t$ are preflows on $G[A_t]$ with source function $\Delta_t$, so $\sum_{v\in A_t}{f_t(v)} = \sum_{v\in A_t}{\Delta_t(v)} = \sum_{v\in A_t}{f'_t(v)}$.

    \end{proof}

    Given a preflow $f$ on a graph $G=(V,E)$, and two disjoint sets $A,B\subseteq V$, denote by $f(A,B)$ the amount of flow that passes from $A$ to $B$. 
	That is, $f(A,B) = \sum_{u\in A}{\sum_{v\in B}{\max\{0,f(u,v)\}}}$.
    
 Since most edges between
 $S_{t-1}$ to $A_t$ are saturated in the direction from $S_{t-1}$ to $A_t$, the following lemma bounds the amount of flow in the opposite direction: from $A_t$ to 
 $S_{t-1}$. From this we derive a bound on 
 $f_{total}$, the total mass introduced throughout the execution.
\begin{lemma} \label{lemma:delta_f_bound}
        For all $1 < t \le \kappa$, it holds that:
        \begin{enumerate}
            \item $f'_{t-1}(A',S_{t-1})\le f'_{t-1}(A_{t}, S_{t-1}) \le \frac{6}{1000}\vol(S_{t-1})$.
            \item ${\delta f}_{t}(A_{t})\le\frac{6}{1000}\vol(S_{t-1})$. In particular, ${f}_\text{total}\le \frac{1}{6}m$.
        \end{enumerate}
    \end{lemma}
    \begin{proof}
        The inequality $f'_{t-1}(A',S_{t-1})\le f'_{t-1}(A_{t}, S_{t-1})$ holds since $A' \subseteq A_t$. By Lemma \ref{lemma:unit_flow_cut}, the number of edges $(v,u)$ where $v\in S_{t-1}, u\in A_{t}$ and $f'_{t-1}(v,u)<0<c$ is at most $\frac{6}{1000c}\cdot\vol(S_{t-1})$. Each of these edges can carry into $S_{t-1}$  at most $c$ units of mass. So together they carry at most $\frac{6}{1000}\vol(S_{t-1})$ units. This proves (1).

        For each such edge, the increase in ${\delta f}_t(u)=f_t(u) - f'_{t-1}(u)$ is at most $-f'_{t-1}(v,u)\le c$. So, ${\delta f}_t(A_t)\le\frac{6}{1000c}\cdot\vol(S_{t-1})\cdot c=\frac{6}{1000}\vol(S_{t-1})$. Initially, $\Delta_1(V)=|A^l| \le \frac{m}{8}$. By the definition of $f_\text{total}$, we get $${f}_\text{total}\le \frac{m}{8}+\sum_{t > 1}{\frac{6}{1000}\vol(S_{t-1})}\le \left(\frac{1}{8}+\frac{12}{1000}\right)m \le \frac{1}{6} m$$.
    \end{proof}
\begin{remark}
        By Lemma \ref{lemma:delta_f_bound}, ${f}_\text{total}\le \frac{1}{6}m$. So we can only have at most $\frac{1}{6}m$ units of mass entering the sinks, across all rounds. 
        In particular we never saturate all the vertices of $A^r$, since the sum of sink capacities of these vertices is $|A^r| \ge \frac{1}{2}m > \frac{1}{6}m$. Since $S=\bigcup_{t \ge 1}{S_t}$ contains only saturated vertices, it cannot be all of $V$, so if $S$ is not empty then it is a nontrivial cut. 
        Moreover, we prove in Lemma~\ref{lem:f(S_t, A_t)} that $\vol(V\setminus S) = \Omega(m)$.
    \end{remark}

\subsection{Bounding $\Phi_{G}(S, V\setminus S)$}

The main claim in Lemma~\ref{lemma:unit_flow} is the bound on the conductance $\Phi_{G}(S, V\setminus S) = \frac{|E(S,V\setminus S)|}{ \min \{\vol(S), \vol(V \setminus S) \}}$ of the cut $S=\bigcup_{t \ge 1}{S_t}$ (see Lemma~\ref{cor:unit_flow_correctness}). We bound the denominator (Lemmas~\ref{lemma:S_t_bound}-\ref{cor:vol_S_bar}) and nominator (Lemmas~\ref{lem:f(S_t, A_t)}-\ref{cor:E(S,A)}) separately. 

The bound we show on the denominator is $\min \{\vol(S), \vol(V \setminus S) \} = \Omega(\vol(S))$ (see Corollary~\ref{cor:vol_S_bar}) and we prove it by showing that  $\vol(V \setminus S) = \Omega(m)$. Intuitively, this is true because all the sinks in $S$ are saturated (the vertices of $S$ have positive labels) so the total sink capacity of nodes in $S$ 
is at most $f_{total}$ (which is at most $\frac{m}{6}$). On the other hand, the total sink capacity of all nodes is $|A^r|\ge \frac{m}{2}$, so $\vol(V\setminus S) \ge \sum_{v\in V\setminus S}{\bar{m_v}}=\sum_{v\in V\setminus S}{T(v)} \ge \frac{m}{2}-\frac{m}{6}$.

The bound we show on the nominator is $|E(S, V\setminus S)|\le \frac{1018}{1000c}\vol(S)$. We establish this, by bounding, over the different iterations $i=1,\ldots, \kappa-1$, the number of saturated edges from $S_i$ to $V\setminus S$ and the number of non-saturated edges from $S_i$ to $V\setminus S$. By Lemma~\ref{lemma:unit_flow_cut}, there are at most $\frac{6}{1000c}\cdot \vol(S_i)$ such non saturated edges. 
To bound the number of saturated edges from $S_i$ to $V\setminus S$, 
we bound $f'_i(S_i,V\setminus S)$ (their sum is bounded by $\frac{1012}{1000}\vol(S)$, see Lemma~\ref{lem:f(S, A)B}) and then divide by $c$, which is the amount of flow carried by each saturated edge from $S_i$
to $V\setminus S$. We get this bound on $f'_i(S_i,V\setminus S)$ by bounding the net flow $f'_i(S_i,V\setminus S) - f(V\setminus S,S_i)$ and $f'_i(V\setminus S,S_i)$ separately.

The following lemma states that the total sink capacity at $S=\bigcup_{t\ge 1}S_t$ is at most the total mass $f_{total}$. Indeed, by the definition of \unitFlow, every vertex in $S$ has a positive label and hence its sink capacity is saturated ($f(v)\ge T(v)$).

\begin{lemma}
    \label{lemma:S_t_bound}
        $\sum_{t\ge 1}{\sum_{v\in S_t}{\bar{m}_v}} \le f_\text{total}$.
    \end{lemma}
    \begin{proof}
        By the definition of $S_t$ according to \emph{\unitFlow}, any $v\in S_t$ has $l'_t(v) \ge 1$. Since $(\Delta_t, f'_t, l'_t)$ is a valid-solution, for every $v\in S_t \cap A^r$, we get that $f'_t(v)\ge T(v)=\bar{m}_v$ (by the second item of Definition \ref{def:G-valid}). Therefore  $\sum_{t\ge 1}{\sum_{v\in S_t}{\bar{m}_v}} \le \sum_{t\ge 1}{\sum_{v\in S_t}{f'_t(v)}}$, and by Lemma \ref{lemma:f_total_is_sum_of_f}, it is upper-bounded by $f_\text{total}$. So, we get that $\sum_{t\ge 1}{\sum_{v\in S_t}{\bar{m}_v}}\le f_\text{total}$.
    \end{proof}
\begin{corollary}
    \label{cor:vol_S_bar}
        $\vol(V\setminus S) \ge \frac{1}{3}m \ge \frac{1}{6}\vol(S)$.
    \end{corollary}
    \begin{proof}
        Recall that for each vertex $v\in V$, $\bar{m}_v\le \dG{}{v}$. We have
        \begin{align*}
            \vol(V\setminus S)& \ge \sum_{v\in V\setminus S}{\bar{m}_v} = |A^r| - \sum_{t\ge 1}{\sum_{v\in S_t}{\bar{m}_v}}
            \\
            &\numge{1} \frac{1}{2}m - f_\text{total} \numge{2} \frac{1}{2}m - \frac{1}{6}m = \frac{1}{3}m = \frac{1}{6}\cdot 2m \ge \frac{1}{6}\vol(S) \ ,
        \end{align*}
        where Inequalities $(1)$ and $(2)$ follow from Lemmas \ref{lemma:S_t_bound} and \ref{lemma:delta_f_bound}, respectively.
    \end{proof}

     Recall that $A' = V \setminus \bigcup_{t=1}^{\kappa} S_t = V\setminus S$. By the definition of $\Delta_{t+1}$, for every $v\in A'$ it holds that
    \[
    \Delta_{t+1}(v) - \Delta_t(v) 
    = 
    \max\left(0, \sum_{v\in S_t}{f'_t(v,u)}\right)
    =
    \max\left(0, f'_t\left(S_t, \{v\}\right) - f'_t\left(\{v\}, S_t\right) \right),
    \]

	so in particular $f'_t(S_t, \{v\}) - f'_t(\{v\}, S_t) \le \Delta_{t+1}(v) - \Delta_t(v) $. As a direct consequence, we get:

\begin{lemma}
	\label{lem:f(S_t, A_t)}
     $f'_t(S_t, A') - f'_t(A', S_t) \le  \Delta_{t+1}(A') - \Delta_{t}(A')$, for every $t\in [1,\kappa-1]$.
\end{lemma}

\begin{lemma}
\label{cor:f(S, A)}
    It holds that $\sum_{t=1}^{\kappa-1}{f'_t(S_t, A') - f'_t(A', S_t)} \le f_\text{total} - \Delta_1(A')$. 
\end{lemma}
\begin{proof}
    By Lemma~\ref{lem:f(S_t, A_t)} we have $\sum_{t=1}^{\kappa-1}{f'_t(S_t, A') - f'_t(A', S_t)} \le \Delta_\kappa(A') - \Delta_1(A')$. Because $f'_\kappa$ is a preflow on $G[A']$, we get $\Delta_\kappa(A') = \sum_{v \in A'}{f'_\kappa(v)}$. By Lemma \ref{lemma:f_total_is_sum_of_f}, we get $\sum_{v \in A'}{f'_\kappa(v)} \le f_\text{total}$.
\end{proof}
\begin{lemma}
\label{lem:f(S, A)B}
    $\sum_{t=1}^{\kappa-1}{f'_t(S_t, A')} \le \frac{1012}{1000}\vol(S)$. 
\end{lemma}
\begin{proof}
    By Lemma \ref{cor:f(S, A)}, 
    \begin{align*}
        \sum_{t=1}^{\kappa-1}{f'_t(S_t, A') - f'_t(A', S_t)} &\le f_\text{total} - \Delta_1(A') = \Delta_1(S) + \sum_{t > 1}{{\delta f}_t(A_t)}
        \\
        &\numle{2} \sum_{t\ge 1}{\sum_{v\in S_t}{m_v}} + \frac{6}{1000}\sum_{t > 1}{\vol(S_{t-1})} \numle{3} \frac{1006}{1000}\vol(S) \ ,
    \end{align*}
    where  Inequality $(2)$ follows from Lemma \ref{lemma:delta_f_bound}(2) and Inequality $(3)$ follows because $m_v \le \dG{}{v}$ for all vertices $v$.
    By Lemma \ref{lemma:delta_f_bound}(1), we get that $\sum_{t=1}^{\kappa-1}{f'_t(S_t, A')}\le \frac{1012}{1000}\vol(S)$.
\end{proof}

Now we use the bound on the flow crossing the cut to bound the number of edges crossing the cut.

\begin{corollary}
\label{cor:E(S,A)}
    It holds that $|E(S, A')|\le \frac{1018}{1000c}\vol(S)$.
\end{corollary}
\begin{proof}
    By Lemma \ref{lemma:unit_flow_cut}, for each $t\in [1, \kappa-1]$,  the number of unsaturated arcs $(u,v)$  where $u\in S_t$ to $v\in A_t\setminus S_t$ (that is $-c\le f'_t(u, v) < c$) is at most $\frac{6}{1000c}\cdot \vol(S_t)$. Each of the saturated arcs $(u,v)$ where $u\in S_t$ to  $v\in A_t\setminus S_t$ carries $c$ units of mass, so by Lemma \ref{lem:f(S, A)B}, the number of saturated edges in the direction from $S$ to $A'$ is at most $\frac{1}{c}\sum_{t=1}^{\kappa-1}{f'_t(S_t, A')}\le\frac{1012}{1000c}\vol(S)$.
    In total, $|E(S, A')|\le \frac{1018}{1000c}\vol(S)$.
\end{proof}

Finally, we bound the expansion of the cut as follows.

\begin{lemma}
\label{cor:unit_flow_correctness}
    At the end of Algorithm \ref{algo:unit_flow_matching}, if the returned cut $S = \bigcup_{t=1}^{\kappa-1}{S_t}$ is nonempty, it satisfies $\Phi_{G}(S, V\setminus S)\le\frac{7}{c}=O(\phi\log n)$ and $\vol(V\setminus S) \ge \frac{1}{3}m$.
\end{lemma}
\begin{proof}
    Assume $S\neq\emptyset$. By Corollary \ref{cor:E(S,A)},  $|E(S, A')|\le \frac{1018\vol(S)}{1000c}= O(\phi\log n\cdot\vol(S))$. 
    \looseness = -1 By Corollary \ref{cor:vol_S_bar}, we get that $\min(\vol(S), \vol(A'))\ge \frac{1}{6}\vol(S)$. 
    So, the conductance of the cut is $\frac{|E(S,A')|}{\min(\vol(S), \vol(A'))}\le \frac{1018\vol(S)}{1000c\cdot \frac{1}{6} \vol(S)} \le \frac{7}{c} = O(\phi\log n)$.
    By Corollary \ref{cor:vol_S_bar} we also have $\vol(V\setminus S) \ge \frac{1}{3}m$.
\end{proof}

The following analysis of the running time of the algorithm is heavily based on Lemma 3.7 in~\cite{saranurak2019expander}, which is similar to Lemma 3.4 in~\cite{henzinger2017flow}. We include the full proof here for completeness.
\begin{lemma}
\label{lemma:unit_flow_time}
    The total running time of Algorithm \ref{algo:unit_flow_matching} is $O(mh)=O(\frac{m}{\phi})$.
\end{lemma}
\begin{proof}

    In a single iteration of \emph{\unitFlow}, initializing the list of active vertices $Q$ takes $O(m)$ time.

    We charge the work performed by \emph{\unitFlow} to the push and relabel operations. Each push operation is charged $O(1)$, and we charge $O(\deg(v))$ to each relabel of $v$. If we cannot push along an edge $(v,u)$ at some point, we won't be able to push along the edge until the label of $v$ increases. Therefore, we can implement \emph{\unitFlow} such that it checks the residual capacity of each edge incident to $v$ only once between relabels, for a total work of $O(\deg(v))$ per relabel.

    We only relabel $v$ at most $h$ times. So, in total, the relabels pay $O(h\cdot\sum_{v\in V}{\deg(v)})=O(mh)$.

    For the pushes, consider the potential function
    \begin{align*}
        \Lambda=\sum_{v\in A_t}{ex(v)l(v)} \ .
    \end{align*}
    Each push operation of $1$ unit of mass decreases $\Lambda$ by at least $1$. So we can charge the $O(1)$ work by the push to the decrease in $\Lambda$. $\Lambda$ starts at $0$ and is always non-negative. So, we can charge the decreases to the increases of $\Lambda$. 

    First, $\Lambda$ increases in relabels, where a relabel of $v$ adds $ex(v)$ units to $\Lambda$.
    The total number of relabels is at most $n\cdot h$. A relabel of a vertex $v$ with $ex(v) \le 1$, increases $\Lambda$ by at most $1$. Therefore, these relabels increase $\Lambda$ by at most $nh = O(mh)$.

    Consider now relabels of vertices $v$ with $ex(v) > 1$. We always push $1$ unit of mass, and  since we choose to push from an active vertex with a smallest label we  never push to a node with positive excess. So the only way for a vertex $v$ to have an excess larger than $1$ at round $t$, is if it starts the round with this excess. At the end of round $t-1$,\footnote{\ie, with respect to $f'_{t-1}$.} $v$ does not have excess: all vertices with positive excess were cut out in $S_{t-1}$ (this follows from Lemma \ref{lemma:unit_flow_cut}). So when round $t$ starts we have that $ex(v)\le f_t(v)-f'_{t-1}(v)={\delta f}_t(v)$. As long as $ex(v)>1$, no mass enters $v$. Therefore, while $ex(v)>1$, we can charge the increase of $ex(v)$ units in $\Lambda$, caused by a relabel of $v$, to ${\delta f}_t(v)$, such that each unit of ${\delta f}_t(v)$ is charged $O(1)$. In total, each unit of initial excess (of which there are at most ${\delta f}_t(A_t)$ in round $t$) is charged at most $O(h)$. It follows that the total increase of $\Lambda$ at round $t$ due to relabels of vertices $v$ with $ex(v)>1$ is $O({\delta f}_t(A_t)h)$.

    Apart from relabels, we may also increase $\Lambda$ when we remove edges between rounds. In this case, $\Lambda$ can increase by ${\delta f}_t(v)$ multiplied by the label of $v$, for each $v$ with  ${\delta f}_t(v)>0$. Thus, we can charge each unit of ${\delta f}_t(v)$ at most $O(h)$ for this increase of $\Lambda$. Again there are at most ${\delta f}_t(A_t)$ units of mass to charge at round $t$.

    In total, by Lemma \ref{lemma:delta_f_bound}, we charge $O(f_\text{total} h) = O(mh)$ for the increase of $\Lambda$. Additionally, it takes $O(\vol(S_t))$ time to compute $S_t$ in Lemma \ref{lemma:unit_flow_cut}, for a total of $O(m)$. Therefore, the running time is $O(mh)$.

\end{proof}

    \subsection{Proof of Lemma \ref{lemma:unit_flow}}
    \label{appendix:proof_of_unit_flow}
\begin{proof} [Proof of Lemma \ref{lemma:unit_flow}]
	\label{proof:lemma:unit_flow}
	The time bound follows from Lemma \ref{lemma:unit_flow_time}.

        When the algorithm terminates, denote $S=\bigcup_{t\ge 1}{S_t}$. If $S=\emptyset$ we have $A'=V$ and we get a feasible flow for the flow problem $\Pi(G)$, That is,  Case (1) of the lemma follows.
        Otherwise, $S\neq\emptyset$, and we get a feasible flow in $A' = V\setminus S$. By Lemma \ref{cor:unit_flow_correctness}, the cut satisfies the properties required by Case (2) of the lemma.
        Because we only added mass to the vertices throughout the execution (and did not increase the sinks), this feasible flow can in particular route the ``truncated'' flow problem where some of the vertices have reduced source mass.\footnote{Recall that we need to solve the problem in which only vertices of $v\in A^l \setminus S$ have positive source mass, and their source mass is $m_v$. During Algorithm \ref{algo:unit_flow_matching}, it is possible that additional vertices $u\in A'\setminus A^l$ will get $\Delta(u)>0$, or that vertices $v\in A^l\setminus S$ will get $\Delta(v)>m_v$. Crucially, $\Delta_t(v)$ is monotone increasing in $t$, so that $\Delta_\kappa(v)\ge\Delta_1(v)$.} This truncated flow can be found in $O(m \log n)$ time using dynamic trees~\cite{ST83}.

\end{proof}

\end{document}